\documentclass[12pt,letterpaper]{article}
\usepackage[T1]{fontenc}
\usepackage{lmodern}
\usepackage[utf8]{inputenc}
\usepackage{geometry}
\usepackage{booktabs}
\usepackage{amsmath,amssymb,amsfonts,amsthm}
\usepackage{dsfont}
\usepackage{cancel}
\usepackage{bm}
\usepackage{enumitem}
\usepackage{mathtools}
\usepackage{changepage}
\usepackage{verbatim}
\usepackage{graphicx}
\usepackage{emptypage}
\usepackage{newlfont}
\usepackage[all,cmtip]{xy}
\usepackage[bottom]{footmisc}
\usepackage[titletoc,title]{appendix}
\usepackage{chngcntr}
\usepackage{apptools}
\usepackage{listings}
\usepackage{color}
\usepackage{sgame, tikz} 
\usepackage{pgfplots}
\usepackage{kpfonts}    
\usepackage{natbib}
\AtAppendix{\counterwithin{lem}{section}}
\usepackage{caption}
\usepackage{subcaption}
\usepackage{float}
\usepackage{xr-hyper}
\PassOptionsToPackage{hyphens}{url}\usepackage[colorlinks=true,linkcolor=blue, allcolors=blue]{hyperref}
\usepackage{sgamevar}
\usepackage{lipsum}

\newcommand{\normal}{\mathbf{\hat{N}}}

\def\imagebox#1#2{\vtop to #1{\null\hbox{#2}\vfill}}

\theoremstyle{plain} 
\newtheorem{theorem}{Theorem}
\newtheorem*{theorem*}{Theorem}
\newtheorem{corollary}{Corollary}
\newtheorem{lemma}{Lemma}
\newtheorem{example}{Example}
\newtheorem{proposition}{Proposition}
\newtheorem{innercustomthm}{Theorem}
\newenvironment{customthm}[1]
  {\renewcommand\theinnercustomthm{#1}\innercustomthm}
  {\endinnercustomthm}

\theoremstyle{definition} 
\newtheorem{definition}{Definition}

\newtheorem{manualasminner}{Assumption}
\newenvironment{manualasm}[1]{%
  \renewcommand\themanualasminner{#1}%
  \manualasminner
}{\endmanualasminner}

\newcommand{\argmax}{\operatornamewithlimits{argmax}}

\newcommand{\join}{\vee}
\newcommand{\meet}{\wedge}

\usepackage{setspace}
\reversemarginpar

\usepackage[nameinlink]{cleveref}
\crefname{manualasm}{assumption}{assumptions}
\crefname{innercustomthm}{theorem}{theorems}
\crefname{prop}{proposition}{propositions}
\crefname{ex}{example}{examples}
\crefname{defn}{definition}{definitions}

\begin{document}

\title{ARTIFICIAL INTELLIGENCE \\AND SPONTANEOUS COLLUSION\thanks{We are indebted to our advisors Kostas Bimpikis, David Kreps, Michael Ostrovsky, and, in particular, Andrzej Skrzypacz for invaluable guidance and support. We also want to thank Susan Athey, Anirudha Balasubramanian, Lanier Benkard, Emilio Calvano, Daniel Chen, Peter DeMarzo, Andreas Haupt, Ravi Jadageesan, Irene Lo, Alexander MacKay, Suraj Malladi, Paul Milgrom, Ilan Morgenstern, Evan Munro, Ilya Segal, Takuo Sugaya, Stefan Wager, Larry Wein, Gabriel Weintraub, Kuang Xu, Frank Yang, and seminar participants at Stanford, Harvard, Kellogg, IESE, UPenn, BU, Bocconi, TSE,  IIOC, and SITE Conference for their insightful comments.}}
\author{Martino Banchio\thanks{Email: mbanchio@google.com. Address: 1600 Amphitheatre Parkway, Mountain View CA 94043, USA} \\ \small{Google Research} \and Giacomo Mantegazza\thanks{Email: giacomom@stanford.edu. Address: 655 Knight Way, Stanford CA 94305, USA} \\ \small{Stanford GSB}}
\date{}
\maketitle
\setcounter{page}{1}
\begin{abstract}

\noindent 

We develop a tractable model for studying strategic interactions between learning algorithms. We uncover a mechanism responsible for the emergence of algorithmic collusion. We observe that algorithms periodically coordinate on actions that are more profitable than static Nash equilibria. This novel collusive channel relies on an endogenous statistical linkage in the algorithms' estimates which we call \emph{spontaneous coupling}. 
The model's parameters predict whether the statistical linkage will appear, and what  market structures facilitate algorithmic collusion.
We show that spontaneous coupling can sustain collusion in prices and market shares, complementing experimental findings in the literature. Finally, we apply our results to design algorithmic markets.

\vspace{0.1in}
\noindent\textbf{Keywords:} Artificial Intelligence, Learning, Collusion\\
\noindent\textbf{JEL Classification Codes:} D47, D83, L40, L50\\
\end{abstract}

\bigskip

\newpage
\section{Introduction}

Artificial Intelligence (AI) software is becoming more common in a variety of business contexts, 
ranging from bidding in online auctions to pricing on shopping platforms and setting short- and long-term rents. This market shift has been accompanied by concerns that automated pricing and bidding software could facilitate collusive behavior, voiced both by regulatory authorities 
(\citet{OECD2017}, \citet{Competition2018}, \citet{CMA2021}) and  by academic researchers (\citet{Harrington2019}, \citet{Calvano2020}, \citet*{Asker2022}).

In this paper, we identify a novel collusive channel that can facilitate collusion between AI algorithms. 
Posing in sharp contrast with explicit cartels or tacitly-collusive equilibria, this channel does not require deliberate intent to collude from market participants.
Instead, collusive outcomes arise as a result of an endogenous statistical linkage between independent, myopic, profit-maximizing learning algorithms. To highlight these characteristics, we call this collusive channel \emph{spontaneous coupling}. 
We show that spontaneous coupling hinges on the algorithmic nature of market participants, instead of relying on their monitoring technology or their intent to collude.

A major challenge in analyzing these games is that the evolution of play is stochastic and discrete. AI algorithms learn by running randomized experiment, and thus they generate stochastic non-stationary paths of play. To handle these challenges, we first show how to approximate these systems in continuous time. Our model encompasses several algorithms from the machine learning literature, and traditional dynamical systems techniques allow us to characterize the algorithms' outcomes analytically.

We first illustrate our machinery in the most canonical game, a Prisoner's Dilemma.
Despite the strategic simplicity of the game,
AI algorithms such as naive Q-learning (not designed to learn dynamic reward strategies such as tit-for-tat) may enter stochastic cycles sustaining high cooperation rates.
We isolate the mechanism responsible for such cooperation.
The algorithms estimate payoffs from each action by running experiments, and exploit their estimates to collect larger rewards.
When experiments are infrequent, the estimates tend to persist for long periods of play, introducing some estimation error.
These errors are correlated, and they tend to synchronize the algorithms' path of play: agents act symmetrically, often jointly cooperating and defecting in stochastic cycles. This phenomenon disappears when the design of the algorithm incorporates careful exploration or counterfactual modeling, and we characterize a class of algorithms that are immune to spontaneous coupling: these algorithms learn to play undominated strategies. 

Our results apply to a variety of economic settings, where algorithms may collude spontaneously.
Our first application is to price-fixing, one of the most prosecuted collusive practices (e.g., the Lysine cartel\footnote{\url{https://www.justice.gov/atr/case-document/information-33}}  prosecuted in 1996 and popularized in the media, or the price-fixing conspiracy\footnote{\url{https://www.justice.gov/archive/atr/public/press_releases/2005/212002.htm}} in the Dynamic Random Access Memory (DRAM) market of the early 2000).
In the setting of \citet*{Asker2022}, which studies price-setting algorithms in a Bertrand competition model, computational experiments show outcomes consistent with price fixing. Algorithms learn to charge supra-competitive symmetric prices, settling on dominated strategies.
We prove that spontaneous coupling sustains such price-fixing without relying on reward-and-punishment schemes, which highlights the shortcoming in current competition policy noted in \citet{Harrington2019}.
Our second application studies another common market manipulation technique known as market division, or ``market splitting''. Market division may appear in various forms: geographical divisions of market shares (as often happens in open-air drug markets), no-poach agreements, or no-show agreements in auction markets. In a model of online keyword auctions, we show that spontaneous coupling can sustain algorithmic market splitting.
The algorithms learn to ``split the market'' by coordinating on the subset of keywords won by each advertiser.
Motivated by online auctions, in our last application we take the perspective of a market designer. We show that it is possible to design strategy-proof mechanisms that are robust to the participation of algorithmic players. To do this, we prove that the statistical linkage we identify disappears if the designer provides enough feedback to the algorithms to assist their learning, and we characterize the policies which communicate minimally necessary feedback. 

\subsection{Intuition}
Before introducing the formal model, we provide a brief roadmap to the results of the paper, and we include an intuitive description of spontaneous coupling.

In \Cref{sec:model} we introduce a model of algorithmic learning, which we call reinforcer. 
Algorithms in this class maintain a vector of values representing each action's payoff consequences and estimate such vectors by repeatedly interacting with the environment and updating the entries based on the observed payoff.\footnote{For example, this class includes the celebrated Q-learning procedure, itself the building block of many AI algorithms, as well as some variants of the Multiplicative Weights Update.} A policy function maps the estimated payoff vector to the agent's action. The policy is responsible for trading off exploration (running experiments, to estimate payoffs accurately) and exploitation (selecting what is thought to be the optimal action, in order to collect rewards).
For example, an $\varepsilon$-greedy policy selects the action that currently has the highest estimated payoff with probability 1-$\varepsilon$, and with probability $\varepsilon$ explores the action space uniformly at random. 

To fix ideas, consider $\varepsilon$-greedy Q-learning algorithms that play a Prisoner's Dilemma.  \Cref{sec:model} provides the mathematical toolkit that allows us to represent the repeated learning game as a dynamical system. 
In \Cref{sec:PD} instead we pin down the statistical linkage between independent algorithms that undermines the dominant strategy incentives and sustains cooperation. In practice, algorithms play symmetric profiles of actions ``too often''. Although algorithms experiment independently, their estimates are correlated because their payoffs depend on the entire action profile, which limits their ability to evaluate profitable deviations. We dub this phenomenon ``spontaneous'' coupling to stress that even algorithms that explore independently may get linked through correlated play. We show that coupling leads to the continuous counterpart of stochastic cycles, in which the agents cooperate most but not all of the time.\footnote{ 
This explains why we observe that, while such algorithms often learn to collude in simulations, collusion is imperfect. It is because the agents cannot converge to a constant profile of actions that is not a pure Nash Equilibrium --- these AIs always  learn how to best respond to a fixed profile of actions. }

Intuitively, during a period of collusion, each agent estimates the value of colluding \emph{conditional on the opponent colluding as well}. This increases the estimated payoff of collusion. At the same time, since experimenting with the competitive action is profitable in the short run, exploration leads the agents to estimate a high payoff for competition. However, as soon as one agent begins exploiting the competitive action, the opponents will quickly best-respond by competing too, and joint competition will yield reduced payoffs. In particular, the estimated value of the competitive action decreases: the agents estimate the value of competing \emph{conditional on the opponent competing as well}. Instead, because the agent explores seldom, the estimate of collusion's payoff remains close to the value of collusion \emph{conditional on the opponent colluding as well}. This draws the algorithms away from competition and back into a cycle of collusive behavior. This is how spontaneous coupling  sustains dominated outcomes: simultaneous deviations reinforce the estimation error in the algorithms' estimates.

Why do algorithms fall into this trap? In \Cref{sec:results} we show that this process depends on a set of parameters, the relative learning rates of an algorithm, i.e. the speed at which each action's estimate gets updated over time. The slower an algorithm learns about infrequently played actions, the more persistent their estimates become. This persistence, combined with nearly simultaneous deviations, leads to sustained periods of correlated play and ultimately, collusion. 
On the other hand, we show that algorithms with uniform learning rates, where each action's estimate is updated at the same speed, do not fall into these collusive practices. Uniform learning rates guarantee that every action enjoys the same persistence, so that algorithms with this property avoid the stochastic cycles and learn to play only undominated strategies.
We conclude with three applications, highlighting the effects of spontaneous coupling on automated markets.

\subsection{Literature Review}
The literature on algorithmic collusion is growing through both experimental work (see, e.g.,  \citet{Klein2021}, \citet{abada2023}, \citet*{Johnson2022}) and empirical work (see, e.g., \citet{Musolff2021}, \citet{Assad2021}).  The seminal results of  \citet{Calvano2020} focus on strategies as a proxy for collusion: the paper argues that simply looking at outcomes of learning might be insufficient, as collusion might arise as a ``mistake'' by poorly designed algorithms. By modeling the dynamics of learning we obtain comparative statics and we are able to determine what collusive schemes arise in algorithmic markets. We show that ``mistakes'' are sustained by spontaneous coupling. 
\citet*{Asker2022} showed that feedback on demand curves influences the pricing behavior of algorithms in simulated Bertrand oligopoly, and \citet{Banchio2021b} finds that additional feedback in first-price auctions restores competition. We complement these studies by generalizing their intuition and by demonstrating the mechanism that underpins collusion in their cases.

Most theoretical models of algorithmic collusion consider simple adaptive algorithms in the interest of tractability, e.g.
\citet{Brown2021}, \citet{Leisten2022}, and \citet{Lamba2022}. In these papers, algorithms choose prices based on the opponent's last quoted price. These are adaptive strategies, but 
algorithms react only to market conditions. They do not improve their predictions over time, which is a key feature of AI algorithms. 
We focus instead on the \emph{learning} algorithms developed by the research community. 
In another simple model with adaptive algorithms, \citet{harrington2022} shows that a monopolistic algorithm provider selling access through a license may design an algorithm with collusive tendencies to upcharge for the license.

A model that incorporates learning appear in \citet*{Hansen2021}, which finds that supra-competitive prices are sustained by coordinated experiments when bidders use the Upper Confidence Bound algorithm. The key difference is that we allow experiments to happen fully at random. Correlation arises through the estimates, not necessarily in the timing of experimentation. 
\citet{Possnig2023} constructs another model of sophisticated learning, and provides a theoretical analysis of the limiting points of reinforcement learning algorithms. The author allows algorithms to condition on past behavior of their opponents, and characterizes the repeated-game strategies learned by the algorithms. Instead, we abstract away from repeated-game strategies in order to better isolate the statistical linkage we call spontaneous coupling. 

Both economists and computer scientists have examined Reinforcement Learning in games, for example \citet{Erev1998} or \citet{mertikopoulos2016}, but with some notable differences with ours. On the one hand, many have analyzed systems experimentally (\citet*{erev1999}, \citet{lerer2017}). Our approach is complementary: with the aid of our framework, one can tell apart experimental findings from agent design considerations. On the other hand, there are some theoretical results on convergence of learning procedures. For example, learning through reinforcement has been associated with evolutionary game theory by \citet{Borgers1997}. Others have formally analyzed some of the simpler models, as in \citet{Hopkins2005}. These results focus on the connection with replicator dynamics. Our approach is different because we consider a general class of learning procedures from the AI literature. Doing so, we obtain a tool valuable for regulation and design of modern automated markets. Moreover, the approach described in \Cref{sec:model} includes earlier results under a general algorithmic structure.
Finally, a collusive scheme similar to the one we identify is observed in a series of papers (\citet{Karandikar1998}, \citet*{Bendor2001a}, \citet*{Bendor2001b}) on aspiration-based learning. Instead of viewing the learning process as a behavioral rule, we study algorithms developed in the context of machine learning, which turn out to have similar characteristics. 

A stream of literature analyzes continuous-time approximation of AI algorithms, mostly in the single-agent setting. Related to ours is \citet*{Tuyls2005}: the authors examine a continuous-time approximation of multi-agent $Q$-learning with Boltzmann (logit) exploration, and show a link with the Replicator Dynamics from the Evolutionary Game Theory (EGT) literature. Building on this work, \citet{Stefanos2022} characterizes the tradeoff between exploration and exploitation in the same setting. Our approximations and results hold in more general settings: we analyze a general class of AI algorithms, and our results leverage their discontinuities. 
Both \citet{Gomes2009} and \citet*{Wunder2010} propose a continuous-time approximation of $Q$-learning in a multi-agent setting with $\varepsilon$-greedy algorithms: their approximations are mutually inconsistent. Most importantly, those approximations remain model-dependent; our method instead applies to general algorithmic forms. The result is a recipe to analyze equilibria through the lens of dynamical systems, abstaining from heuristic modeling choices.

The research of \citet{Benaim1996} and \citet{Borkar2000} often serve as a foundation for stochastic approximations in learning: see, e.g., \citet*{Russo2021}, which uses continuous-time approximations to analyze the finite-time statistical properties of single-agent Temporal Difference learning. Our approach relies instead on \citet{Kurtz1970}, which provides a more flexible tool to handle general learning procedures and multi-agent settings. Additionally, we adopt the formalism of differential inclusions to analyze points of non-differentiability, which are generally overlooked in traditional stochastic approximations analyses, but that prove to be central in our study.\footnote{One exception is the paper by \citet*{Wunder2010}, which however abandons this route in favor of simulations.}

\section{Model}\label{sec:model}
Our model encompasses several canonical learning algorithms: from Q-learning variants, such as $\varepsilon$-greedy Q, to Multiplicative Weights Update and EXP3. What all these algorithms have in common is that they reinforce successful actions and penalize unsuccessul ones while interacting repeatedly. For this reason, we call an algorithm that follows our learning model a \emph{reinforcer}.

\subsection{Learning Algorithms in Games}\label{sec:adaptive}
Consider a finite normal-form game $G = (N,(A_i)_{i \in N}, (r^i)_{i \in N})$ with $N$ players. 
We are interested in what happens when agents repeatedly play this game and choose actions using a learning algorithm. Let us focus on one agent, Alice, whose action set $A_i$ has cardinality $d_i$; Alice delegates decision-making to an algorithm that attempts to maximize her utility. Her utility function, $r^i(a^i,a^{-i})$, depends on her opponents' actions and her own.
 
\begin{example}\label{ex:Q}
Alice employs $\varepsilon$-greedy Q-learning. This algorithm consists of a vector $Q(k)$ for every period $k$ and a decision rule $\pi_\varepsilon$.
\begin{itemize}
    \item Each entry $Q_a(k)$ is an estimate of the long-run value of action $a \in A$. The update for entry $Q_a(k+1)$ in period $k+1$ is given by
    \begin{equation}
    \label{eq:q learning update}
        Q_{a}(k+1)=\begin{cases}
        Q_{a}(k)+\alpha \left[ r(k) +\gamma \max_{a^\prime} Q_{a^\prime}\left(k\right)-Q_{a}\left(k\right)\right] & \text{if } a=a(k)\\
        Q_{a}(k) & \text{else}.
    \end{cases}
    \end{equation}
    where $r(k)$ is the payoff, and $a(k)$ is the action the algorithm took, in period $k$. Parameters are $\gamma \in [0,1)$, a discount factor, and $\alpha \in  (0,1)$, called learning rate.
    \item Given the vector $Q(k)$, the algorithm takes actions according to a $\varepsilon$-greedy policy. The decision rule $\pi_\varepsilon \colon \mathbb{R}^{|A|} \to \Delta(A)$ selects the following probability distribution over actions:
    \[
    \pi_{\varepsilon}\left(Q(k)\right) =
    \begin{cases}
        \dfrac{1}{| \argmax_{a \in A_i} Q_a(k)|} \quad \forall a \in \argmax_{a \in A_i} Q_a(k)  &  \text{ with probability } 1-\varepsilon \\
        \frac{1}{d_i} \qquad \qquad \qquad \qquad \; \forall a \in A_i  & \text{ with probability } \varepsilon
    \end{cases}
    \]
\end{itemize}
Intuitively, the Q-vector estimates the long-run value of actions by iterating over a Bellman equation. In period $k$, the value of an action is a convex combination of its previous estimate (weighted by $1-\alpha$) and a new Bellman estimate (weighted by $\alpha$).\footnote{$Q$-learning is a more general version of the \citet{Erev1998} and \citet{Borgers1997} reinforcement learning models. We analyze in this work a different, more straightforward decision rule.} The weight $\alpha$ controls the persistence of the estimates, with lower values implying larger persistence of past experiences.
\end{example}

The $\varepsilon$-greedy policy offers a straightforward way to balance exploration and exploitation. Specifically, with probability $1-\varepsilon$, the algorithm selects the action corresponding to the highest entry of the Q-vector, which reflects the agent's belief about the best course of action at that time. In contrast, with probability $\varepsilon$, the algorithm takes a random action, enabling the agent to explore the action space. Because of its simplicity and attractive properties in single-agent environments, it is often used as a benchmark for more complex exploration policies.

Notice that Q-learning is misspecified when it is used as a learning and decision rule in a game. This is because Alice's Q-vector estimates a continuation payoff $Q_{a^i}$ as a function of Alice's own actions only, while payoffs depend also on the opponents' unobserved actions, $a^{-i}$. 
If the opponents' profile of strategies were fixed, Q-learning would converge on the best response to that profile,\footnote{The proof is a simple adaptation of the arguments in \citet{Watkins1992}.} but in a strategic setting where all agents are learning, there are no guarantees of convergence.

Q-learning is a representative of a learning model that promotes successful play. We call this the reinforcer model.

\begin{definition}\label{def:reinforcer}
A \emph{reinforcer} for agent $i$ is a pair ($\theta^i$, $\pi^i$) consisting of 
\begin{itemize}
    \item A $d_i$-dimensional stochastic process $\theta^i$ that evolves according to 
    \[
    \theta^i(k+1) = \theta^i(k) + \alpha^i D^i(a^i(k),r^i(k),\theta^i(k)),
    \]
    where $a^i(k) \in A_i$ is the action taken by agent $i$ in period $k$, $\theta^i \in \mathscr{T} \subset \mathbb{R}^{d_i}$, and $\alpha^i \in \mathbb{R}_+^{d_i}$ are learning rates for all actions $a^i \in A_i$.
    \item A \emph{policy} $\pi^i$, that is a map $\pi^i \colon \mathscr{T} \to \Delta(A_i)$, which selects a distribution of  actions for each value of the process $\theta^i \in \mathscr{T}$. 
\end{itemize}
\end{definition}

A reinforcer carries a statistic for each available action and selects an action in period $k$ according to its policy. Both agent $i$'s and her opponents' policies introduce randomness in the process $\theta^i$. The update function $D^i$ depends Alice's realized actions directly and on her opponents' actions through her utility. Given a (possibly random) initial value for $\theta^i(0)$, which we call the \emph{initialization} of $\theta^i$, this stochastic process is well-defined.
The following assumption is maintained throughout the paper:

\begin{manualasm}{A1}\label[manualasm]{manualasm:lipschitz}
The functions $D^i(a^i,r,\theta)$ and $\pi^i(\theta)$ are Lipschitz-continuous almost-everywhere in $\theta$.
\end{manualasm}

Let us return to \Cref{ex:Q}. The policy $\pi_{\varepsilon}$ is constant on the subspace of $\mathbb{R}^{d_i}$ where $\argmax_{a} \theta^i_a$ is fixed and unique. The argmax changes only along the lines of the form $\left\{\theta^i | \theta^i_a = \theta^i_{a'} = \max \theta^i\right\}$, which have zero Lebesgue measure. The discontinuities of the update function of Q lie on the same lines, and thus the update is also a.e. Lipschitz. The $\varepsilon$-greedy Q-learning satisfies \Cref{manualasm:lipschitz}.

When multiple agents employ reinforcers, we represent the system as a single vector of dimension $d_1 + \dots + d_N$ by stacking the individual algorithms, denoted by $\theta(k) = (\theta^1(k),\dots,\theta^N(k))$. Similarly, $D$ and $\pi$ indicate the collection of update functions and policies when missing a superscript. 

\subsection{Approximation in Continuous Time}\label{sec:approximation}

We are interested in describing the dynamics of learning of a reinforcer.
While relatively simple to analyze in a stationary, single-decision-maker environment, reinforcers become unpredictable when learning to play against each other. For example, all convergence properties of Q-learning rely on stationarity assumptions. Analyzing the learning path of any number of Q-learning agents requires describing discrete, stochastic updates and how these interact over time.

We deal with these difficulties using a continuous-time approach known as \emph{fluid approximations}.  We adopt the formalism first introduced by \citet{Kurtz1970}: the idea behind fluid approximations is to analyze the limit of the systems as the jumps get small and their frequency increases, which yields a typically more tractable ODE system. We expect to accurately model most online marketplaces, where decisions are taken at very high frequencies and the impact of individual decisions is usually small.

The first step of the approximation is re-casting a reinforcer $\theta$ as a process in continuous time, a procedure called \emph{Poissonization}. Formally, we consider a Poisson clock with rate $\lambda^1 = 1$ and define $\theta^1(t)$ as the pure-jump continuous-time process such that $\theta^1(0)=\theta(0)$ and that is constant for all $t$ except at the ticks $\tau$ of the Poisson clock, when it gets updated according to $D$. Intuitively, one can think of a sequence of stage games that are played only at the tick of the clock. We introduced a new layer of randomness, but the path traced by $\theta^1(t)$ remains identical to the path of the discrete $\theta(t)$.

We generate a sequence of processes $(\theta^n)_{n \in \mathbb{N}}$ by increasing the rate of the Poisson clock to $\lambda^n = n$. The goal is to regularize the learning dynamics; therefore, we need 
to compensate for frequent updates by reducing the contribution of each jump to the total estimate.
We do this by dividing each jump by $\frac{1}{n}$: in the parlance of \Cref{def:reinforcer}, at a given arrival time $\tau$ of the Poisson process, 
\[\theta^n(\tau) = \theta^n(\tau^-) + \frac{\alpha}{n} D\big(a(\tau),r_{\tau},\theta^n(\tau^-)\big).\]
Each process $\theta^n$ has the same infinitesimal generator: the sequence $(\theta^n)_{n \in \mathbb{N}}$ preserves the instantaneous rate of change uniformly. We can prove the following result:

\begin{theorem}\label{thm: fluid approx thm}
Let $H \subset \mathscr{T}$ be such that $D$ and $\pi$ are Lipschitz over $H$, and let $y_0 \in H$ be the initialization point of $\theta$. Then, the sequence of continuous-time stochastic processes $(\theta^n)_{n \in \mathbb{N}}$ converges in probability to the solution of the following Cauchy problem:
\[\begin{cases}
\frac{d\Theta^i(t)}{dt} = \alpha \mathbb{E}_{\pi^i,\pi^{-i}}\left[D^i\big(a^i,r(a^i,a^{-i}),\Theta^i(t)\big) \right] \\ 
\Theta^i(0) = y^i_0
\end{cases}\]
for all $i$.
That is,  $\lim_{n\to\infty} P\Big\{\sup_{t\leq T} \Big\lVert \theta^n(t) - \Theta(t)\Big\rVert > \eta \Big\}=0$ for all $T\geq 0$ and $\eta>0$ such that $\{\Theta(t)\}_{t\leq T} \subset H$.

\end{theorem}

We provide a formal construction of the sequence $\theta^n$ and a proof of this result in \Cref{app:proofs}.
The process $\Theta$ is the fluid approximation to $\theta$: it is a deterministic dynamical system whose time-derivative is the expected update that the discrete process $\theta$ would incur over one unit of time. \Cref{thm: fluid approx thm} guarantees that the sequence of pure-jump processes $(\theta^n)_{n \in \mathbb{N}}$ draws closer and closer (in probability) to the continuous process $\Theta$.  
The theorem relies on a law-of-large-numbers argument: when updates occur at high frequency and each update is small, the process behaves similarly to its expectation. 
We leverage the fundamental theorem of fluid approximations by \citet{Kurtz1970}, a stochastic-processes version of the law of large numbers, to conclude convergence in probability.

\subsection{Reinforcers in Continuous Time}
Instead of analyzing the discrete reinforcers, we focus on their fluid limits. We now introduce a number of concepts from the literature on dynamical systems that we use throughout the paper to formalize our statements. 

\begin{definition}
Given a dynamical system $\frac{d\theta}{dt}$, the \emph{flow} of $\theta$ starting in $x$ is the map
\[
\theta(-,x) \colon \begin{array}[t]{c c c} 
          [0,+\infty) &\rightarrow& \mathscr{T} \\ 
          t &\mapsto& \theta(t,x) 
         \end{array}
\]
that satisfies \Cref{eq:reinforcer} with initial condition $\theta(0,x) = x$. A \emph{trajectory} of the dynamical system is the graph of the flow, $\big\{(t,\theta(t,x)) \colon t \in [0,+\infty)\big\}$. The set $\Gamma_x = \big\{ \theta(t,x) \colon t \in [0,+\infty) \big\}$ is the \emph{orbit} of $\theta$ starting from $x$. 
If a sequence $(t_n)_{n \in \mathbb{N}}$ is such that 
\[
    \begin{cases}
    \underset{n \to \infty}{\lim} t_n = +\infty \\
    \underset{n \to \infty}{\lim} \theta(t_n,x) = y 
    \end{cases}
\]
we say that $y$ belongs to the \emph{forward limit set} for the orbit $\Gamma_x$.
A \emph{steady state} of the dynamical system is a fixed point of its law of motion, i.e. $\frac{d\theta}{dt}(t) = 0$. 
\end{definition}

\Cref{def:reinforcer} specifies that the actions taken by a reinforcer depend directly on its current estimates. Thus, analyzing the dynamical system of $\theta^i$ is a good proxy for the path of play. Moreover, studying the forward limit set of a reinforcer allows us to focus on estimated values instead of realized actions. Since many policies, such as  $\varepsilon$-greedy, prescribe randomly exploring at small rates arbitrarily ahead in the future, even if the reinforcer settles its estimates and identifies the action with the largest expected reward, it will keep playing sub-optimal actions (albeit with a small probability). Thus, focusing on the path of estimates instead of the path of play provides a natural way to define stationarity in the case of learning: a reinforcer becomes stationary if its estimates reach a steady state. We adopt this view for simplicity and clarity of exposition.
\begin{definition}\label{def:steadystate}
We say a reinforcer $(\theta^i, \pi^i)$ initialized at $x$ \emph{converges} if $\theta^i(t)$'s flow started at $x$  converges on a steady-state $\overline{\theta}^i$.
We say the agent \emph{converges on action $a_\text{ss}$} if the steady-state $\overline{\theta}^i$ is such that $a_\text{ss} \in \argmax_{a \in A_i} \overline{\theta}^i_a$. If the argmax is not unique, we say $\overline{\theta}^i$ is a \emph{pseudo-steady-state}.
\end{definition}

Requiring that a steady state always exist can be stringent, which motivates the following definition of \emph{learning} for this paper: agents learn action $a_l$ if the estimate of that action is always the largest in the limit.

\begin{definition}
We say the reinforcer \emph{learns action $a_l$} if there exists a $T>0$ such that $a_l \in \argmax_{a \in A} \theta^i_a(t)$  for all $t \geq T$. Equivalently, for all $\overline{\theta}^i$ in the forward limit set of a given orbit,  $\overline{\theta}^i_{a_l} = \max_{a \in A} \overline{\theta}^i_a(t)$.
\end{definition}

It follows that if the forward limit set of $\theta$ is a singleton, an agent who learns action $a_l$ also converges on action $a_l$. 
Finally, notice that an agent can learn action $a_l$ even if she doesn't play said action in each period (for example because an agent explores with positive probability in the limit). 

The definition of reinforcers is rather permissive, allowing for many procedures well known in the AI literature. For the sake of tractability, in the rest of the paper we will focus on what we call \emph{separable} reinforcers.

\begin{definition}\label{def:regularity}
    A reinforcer $(\theta^i,\pi^i)$ is said to be \emph{separable} if the fluid approximation of $\theta^i$ is of the form
    \begin{equation}\label{eq:reinforcer}
        \frac{d\theta^{i}_{a}(t)}{dt} = \alpha^{i}_{a}\big(\theta^i(t)\big)\Big[U\big(\theta^{i}_a(t),\mathbb{E}_{\pi_{-i}}[r^i(a,\pi_{-i})]\big) + V\big(\theta^i(t)\big)\Big].
    \end{equation}
    where $\alpha^i_a(\theta^i) \colon \mathscr{T} \rightarrow [0,1]$. Moreover we require that $\alpha$, $U$, and $V$ be a.e. Lipschitz in all components, $U$ be Lipschitz everywhere and increasing in $\mathbb{E}[r(a,\pi_{-i})]$ and decreasing in $\theta^i_a$, and $\frac{\partial V}{\partial \theta^{a^i}} < - \frac{\partial U}{\partial \theta^{a^i}}$ almost everywhere. 
\end{definition}

Two parts make up a separable reinforcer: a component that is equal for all actions, $V(\theta^i)$, and a component that is action specific but Lipschitz over the entire domain $\mathscr{T}$, $U(\theta^i_{a},\mathbb{E}[r^i])$. The function $U$, uniform across all actions, operates on a given action's estimates. The first of the monotonicity assumptions amounts to requesting that a reinforcer's updates increase with good news. The second states that a reinforcer likes surprises: the update shrinks if the agent already holds an action in high regard. 

\addtocounter{example}{-1}
\begin{example}[continued]
    $\varepsilon$-greedy Q-learning is a separable reinforcer. For Q-learning, 
    \begin{align*}
    U(Q^i_a(t),\mathbb{E}[r^i]) &= \mathbb{E}_{\pi^{-i}}[r^i(a^i,a^{-i})] - Q^i_a(t) \\
    V(Q^i) &= \gamma \max_a Q^i_a(t)\\
    \alpha^i_a(Q^i(t)) &= \alpha^i \cdot (\pi_\varepsilon(Q^i(t)))_a
    \end{align*}
    Thus, $U$ is linear in its arguments and Lipschitz everywhere, increasing in rewards and decreasing in $Q^i_a(t)$. The common component $V$ is constant over its Lipschitz domains and since $\gamma<1$ its derivative is dominated by $U$'s. Finally, the learning rate $\alpha^i_a(Q^i(t))$ is the product of the common rate $\alpha^i$ and the probability of selecting action $a$ given $Q^i(t)$.
\end{example}

The restriction to separable reinforcers does not rule out standard algorithms.
Separability mainly requires that an algorithm "treats each action the same": the only term that can differ across actions is the learning rate $\alpha^i_a$. All entries of the vector $\theta^i$ adopt the same action-specific component, and any interaction term appears uniformly in all actions' updates.

\section{Q-learning in the Prisoner's Dilemma}\label{sec:PD}
This section analyzes how Q-learning algorithms behave in a family of Prisoner's Dilemmas to build intuition towards more general statements. Simulations show that Q-learners can coordinate to cooperate over a wide range of settings, but their coordination is not perfect and they randomly switch between cooperation and defection.
Analyzing the fluid approximation allows us to pin down the key mechanism behind cooperation.

\paragraph{Model.} We consider a family of Prisoner's Dilemma games, as illustrated in  \Cref{fig:stage game payoff}. Alice and Bob both start with an equal endowment of two US dollars and they simultaneously decide whether to invest their funds in a shared pool that grows in value by a factor denoted as $g$, where $g$ falls between $1$ and $2$. The accumulated wealth in the pool is then divided equally between Alice and Bob, and those who didn't invest get to keep their endowments as well.
\begin{figure}[h]
\centering
    \begin{game}{2}{2}[Alice][Bob]
    \> $C$ \> $D$\\
    $C$ \>$2g,2g$ \>$g,2+g$\\
    $D$ \>$2+g,g$ \>$2,2$
    \end{game}
    \caption{Payoffs of the stage game, $1< g < 2$.}
    \label{fig:stage game payoff}
\end{figure}
This game is a canonical model of the free-rider problem, equivalent to a Prisoner's Dilemma. The parameter $g$ models the attractiveness of joint cooperation: the larger $g$, the more attractive cooperation (action C) becomes. However, for all $g \in (1,2)$, the dominant strategy, and the only Nash equilibrium, is to play ``defect" (action D) and keep one's change. 

\paragraph{Simulations.} We simulate the path of play of Alice and Bob when they adopt $\varepsilon$-greedy Q-learning in this free-rider problem $100$ times for various values of $g$. \Cref{fig:Bar Chart} shows the results of these numerical experiments. The algorithmic agents learn to play 
the dominant strategy equilibrium $\{D, D\}$ only for low values of the parameter $g$. 
\begin{figure}[h!]
\centering
    \begin{subfigure}[t]{0.49\textwidth}
    \includegraphics[width=80mm]{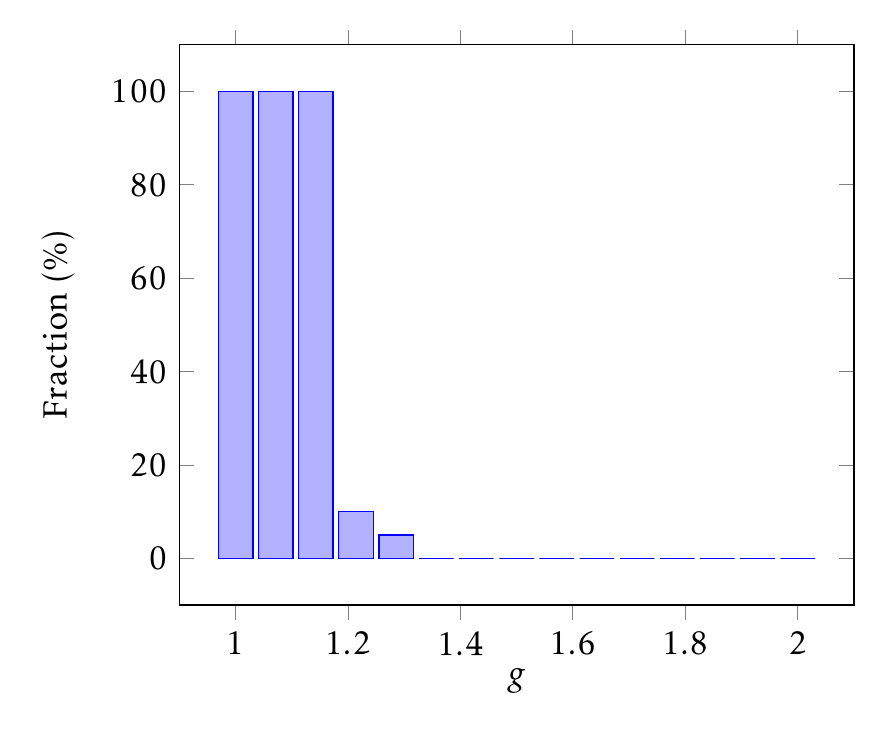}
    \caption{Each bar represents the fraction of the 1000 simulations for a given $g$ where the agents learn the Nash Equilibrium $\{D,D\}$ (agents play the pair $(D,D)$ in over $50\%$ of iterations).}
    \label{fig:Bar Chart}
    \end{subfigure}
    \hfill
    \begin{subfigure}[t]{0.49\textwidth}
    \centering
    \includegraphics[width=75mm]{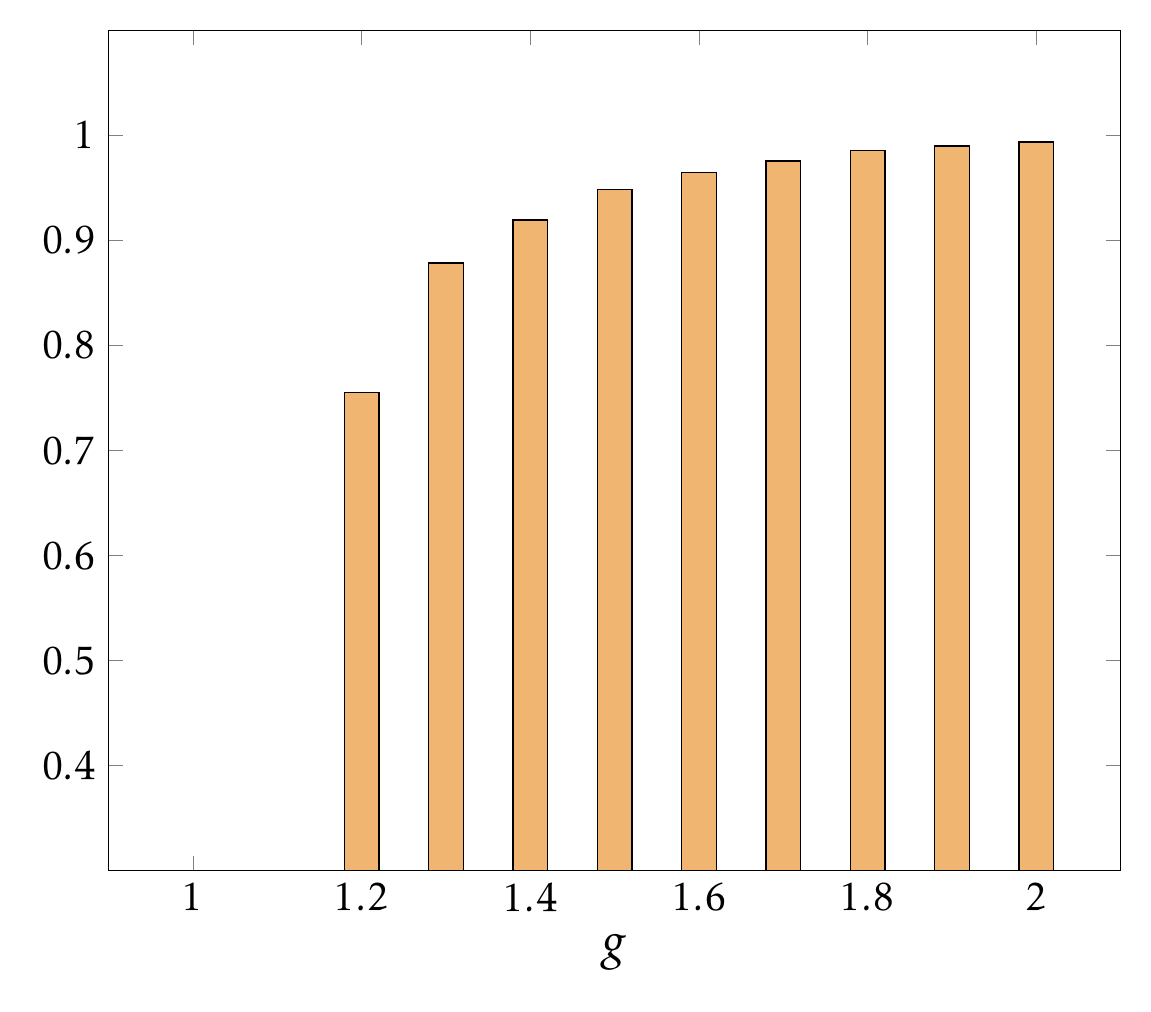}
    \caption{Each bar represents the ratio of periods $k \in [80,001, 100,000]$ where cooperation is a player's preferred action ($Q_C(k) \geq Q_D(k)$) over the total number of periods in the same interval ($20,000$), for a sample experiment with value $g$.}
    \label{fig:localtime}
    \end{subfigure}
    \caption{For these simulations, we chose parameters $\varepsilon = 0.1$, $\alpha = 0.05$ and $\gamma = 0.9$. In each simulation, we initialize both algorithms optimistically, i.e. with values larger than the maximum available continuation value, $\frac{2g}{1-\gamma}$.}
    \label{fig:Experiment results}
\end{figure}
Instead, when $g$ is large, the agents cooperate, albeit imperfectly. Both cooperation and defection appear in unpredictable but recurrent cycles, as shown in \Cref{fig:Alice,fig:Bob}: the value of collaboration is generally above that of defection, but after some time it drops below $Q_D(k)$, so that agents switch to playing $D$. The value of defection then decreases almost immediately, and players revert to cooperation. 
These simulations highlight a few puzzles. First, cooperation arises only for large values of $g$, even though defection is always a dominant strategy. Second, cooperation seems to consist of cycles, but one cannot impute these to ``retaliatory" strategies, since Q-learning does not carry memory of past play. \Cref{fig:localtime} shows that in the long run Alice and Bob play cooperation for a large fraction of the time. 
\begin{figure}[h!]
\centering
    \begin{subfigure}[b]{0.49\textwidth}
    \includegraphics[width=70mm]{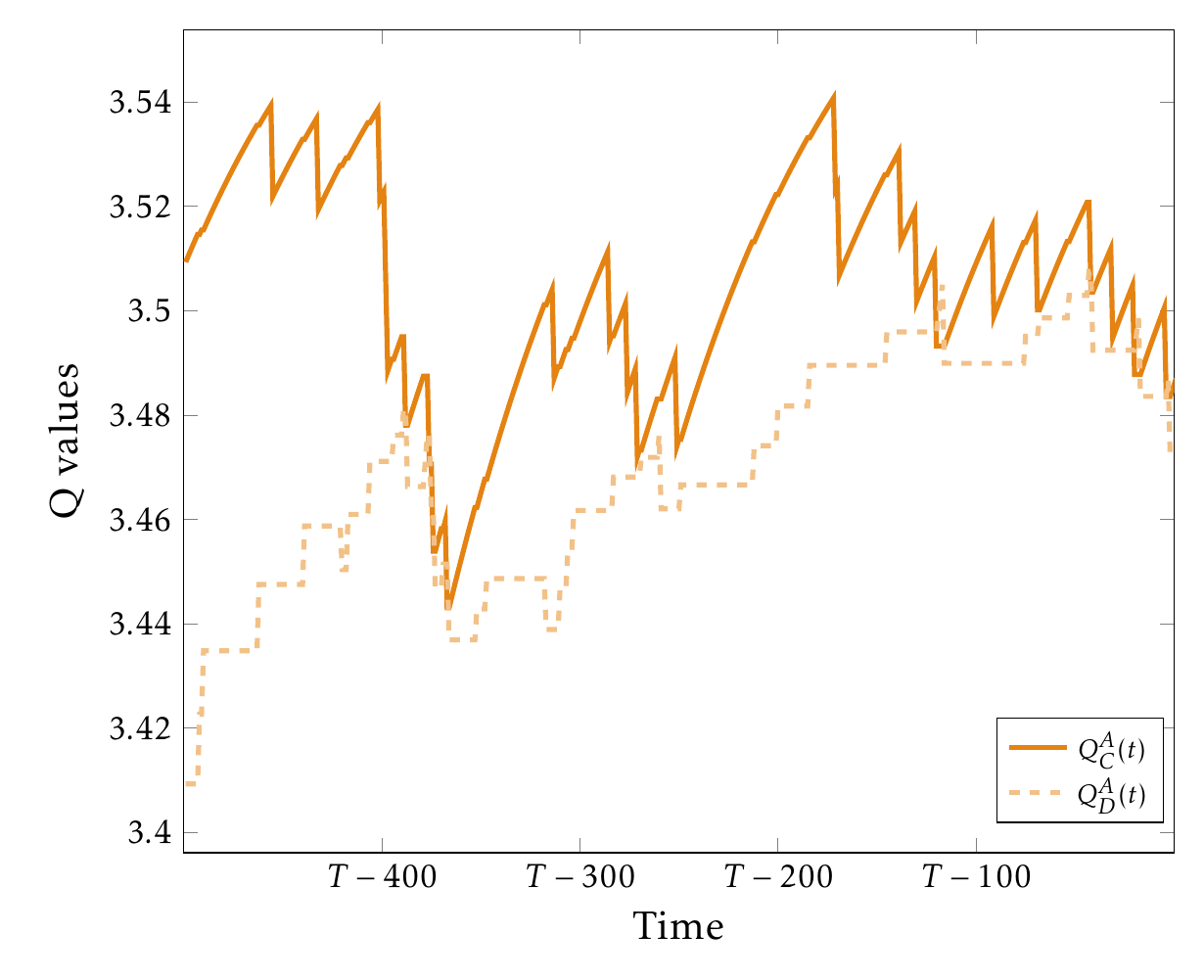}
    \caption{Evolution of Alice's Q-values.}
    \label{fig:Alice}
    \end{subfigure}
    \hfill
    \begin{subfigure}[b]{0.49\textwidth}
    \centering
    \includegraphics[width=70mm]{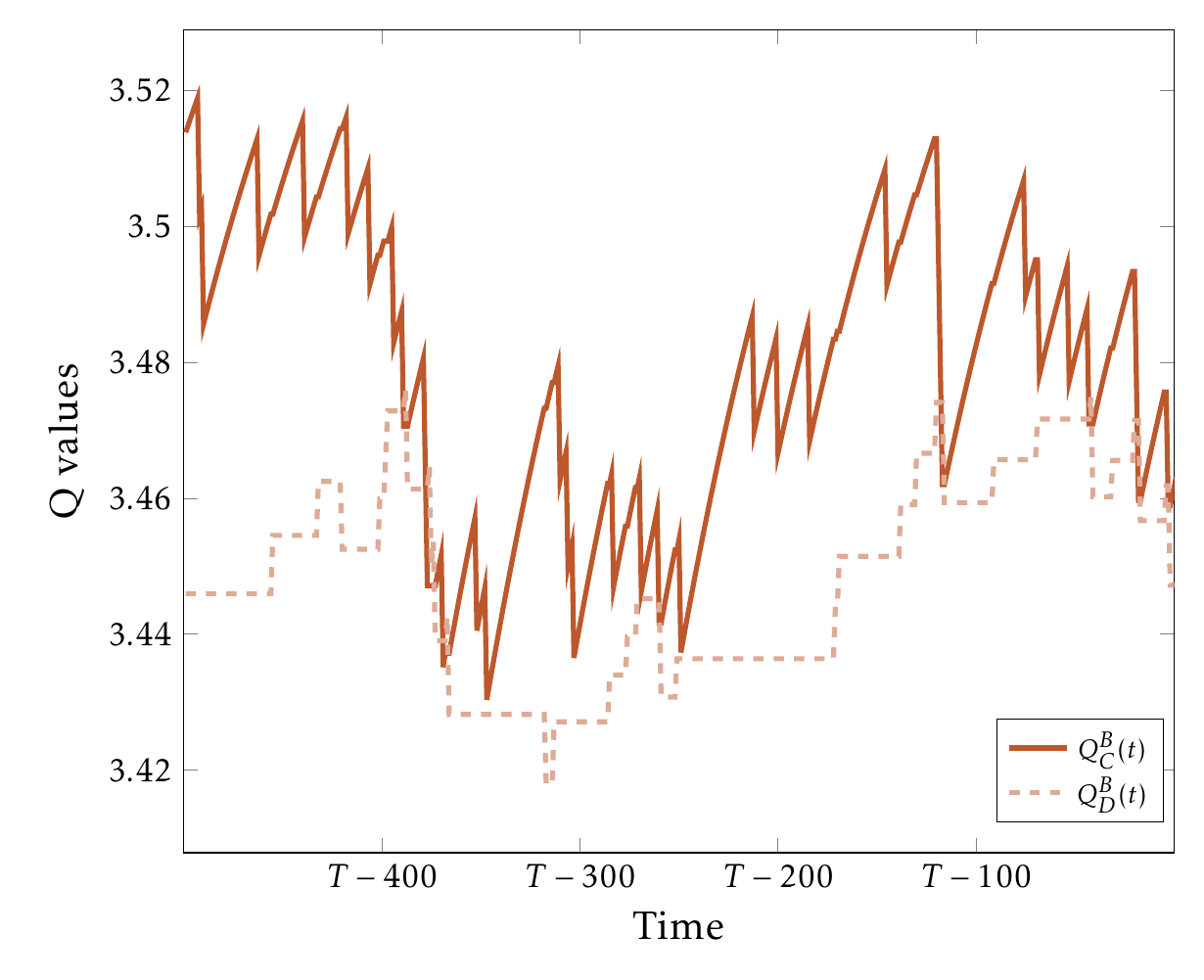}
    \caption{Evolution of Bob's Q-values.}
    \label{fig:Bob}
    \end{subfigure}
    \caption{The two graphs depict the last 500 iterations of a sample simulation with $g=1.8$. Other parameters are the same as in \Cref{fig:Experiment results}.}
    \label{fig:jagged}
\end{figure}

\subsection{Theoretical Results}\label{sec:results PD}

We begin the theoretical analysis by explicitly deriving the continuous-time approximation of $\varepsilon$-greedy Q-learning. In every iteration, Alice and Bob play the action corresponding to their largest payoff estimate (with high probability). Both learn something about the value of their action and then update their payoff estimate.
In the parlance of \Cref{def:reinforcer}, both Alice and Bob evolve according to a function $D^i$ that is discontinuous along the surface $Q^i_D(k) =  Q^i_C(k)$. To sidestep this discontinuity in the right-hand side of the discrete-time system, we first apply \Cref{thm: fluid approx thm} to $Q(k)$ over the largest open sets such that $D = (D^A,D^B)$ is everywhere Lipschitz. We call these sets \emph{maximal continuity domains}: in the case of $\varepsilon$-greedy Q-learning these are sets $\omega_{a,b}$ of vectors $Q$ such that Alice's greedy action is $a$ and Bob's greedy action is $b$; let $\overline{\omega}_{a,b}$ be their closure.

Over $\omega_{C,C}$ the greedy action for both players is $C$, so that in every period Alice cooperates with probability\footnote{That is, $C$ is selected with probability $1-\varepsilon$ if the randomization device instructed the agent to be greedy and with probability $\frac{\varepsilon}{2}$ if the agent plays a random action.} $1-\frac{\varepsilon}{2}$ and defects with probability $\frac{\varepsilon}{2}$. Hence, with probability $(1-\frac{\varepsilon}{2})^2$ she collects reward $2g$ --- similarly for other profiles. Therefore, the fluid limit in $\omega_{C,C}$ solves
\begin{equation}
    \label{eq:ODE in CC}
    \begin{cases}
    \dfrac{d\mathbf{Q}^A_C(t)}{dt}=\alpha\left(1-\dfrac{\varepsilon}{2}\right)\left[\left(1-\dfrac{\varepsilon}{2}\right)2g+ \dfrac{\varepsilon}{2} g+(\gamma-1)\mathbf{Q}_C^A(t)\right]\\
    \dfrac{d\mathbf{Q}^A_D(t)}{dt}=\alpha \dfrac{\varepsilon}{2} \left[\left(1-\dfrac{\varepsilon}{2}\right)(2+g) +2\dfrac{\varepsilon}{2}+\gamma \mathbf{Q}_C^A(t)- \mathbf{Q}_D^A(t)\right] \\
    \dfrac{d\mathbf{Q}^B_C(t)}{dt}=\alpha\left(1-\dfrac{\varepsilon}{2}\right)\left[\left(1-\dfrac{\varepsilon}{2}\right)2g+ \dfrac{\varepsilon}{2} g+(\gamma-1)\mathbf{Q}_C^B(t)\right]\\
    \dfrac{d\mathbf{Q}^B_D(t)}{dt}=\alpha \dfrac{\varepsilon}{2} \left[\left(1-\dfrac{\varepsilon}{2}\right)(2+g) +2\dfrac{\varepsilon}{2}+\gamma \mathbf{Q}_C^B(t)- \mathbf{Q}_D^B(t)\right]
    \end{cases}
\end{equation}
where we denote in bold the fluid approximation of the algorithms of Alice and Bob. Similar linear systems appear in all continuity domains. We can write the dynamical system in matrix form as 
\begin{equation}
    \label{eq:allODE}
    \dot{\mathbf{Q}}(t) =
    \begin{cases}
    A_{C,C}\mathbf{Q}(t) + b_{C,C} \text{ for } \mathbf{Q}(t) \in \omega_{C,C}\\
    A_{C,D}\mathbf{Q}(t) + b_{C,D} \text{ for } \mathbf{Q}(t) \in \omega_{C,D}\\
    A_{D,C}\mathbf{Q}(t) + b_{D,C} \text{ for } \mathbf{Q}(t) \in \omega_{D,C}\\
    A_{D,D}\mathbf{Q}(t) + b_{D,D} \text{ for } \mathbf{Q}(t) \in \omega_{D,D}
    \end{cases}
\end{equation}
Dynamical systems with discontinous right-hand sides can be analyzed using the tools of differential inclusions: \citet{Filippov1988} guarantees that there exists a forward solution to \Cref{eq:allODE}.

The dynamics of this $4$-dimensional piecewise-linear system are complex: the system exhibits chaotic behavior over various parametrizations and initial conditions (see \Cref{app:chaos} for more details on chaos theory and its analysis in the Prisoner's Dilemma). To make progress in the analysis, we begin by restricting attention to the subspace $\{Q \in \mathbb{R}^4 |Q^A_a = Q^B_a \text{ for } a=C,D\} \cong \mathbb{R}^2$, where the system is bound to remain if the initial condition is symmetric. We will then make the following assumption for the rest of \Cref{sec:PD}:
\begin{manualasm}{S}\label[manualasm]{manualasm:symmetry}
Let $\mathbf{Q}^A_a(0) = \mathbf{Q}^B_a(0)$ for $a = C,D$. 
\end{manualasm}
Under \Cref{manualasm:symmetry}, we prove the following proposition, which characterizes the limiting behavior of the continuous-time approximation.

\begin{proposition}\label{prop:two eq PD}
Let $\underline{\varepsilon}(g) =  1-\sqrt{\frac{2-g}{g}}$. The forward limit set of $\mathbf{Q}$ is a singleton for any initial condition. 
If $\varepsilon \geq \underline{\varepsilon}(g)$, all initial conditions lead to the unique steady state 
\[
q^{eq}_D = \left( \frac{2\varepsilon + 2g - \varepsilon g}{2} +\frac{ \gamma \left( 4+\varepsilon g \right)}{2(1-\gamma)}, \frac{  4+\varepsilon g }{2(1-\gamma)}  \right)
\]
which lies in $\omega_{D,D}$.

If $\varepsilon < \underline{\varepsilon}(g)$, initial conditions may lie in one of two regions of attractions, $R_D$ and $R_C$. All initial conditions in $R_D$ lead to the steady state $q^{eq}_D$. All initial conditions in $R_C$ instead lead to the pseudo-steady state (see \Cref{def:steadystate})
\[
q^{eq}_{C}= (y,y) \text{ where } y= \frac{1+g + \sqrt{(g-1)(g-1-\varepsilon g+ \frac{\varepsilon^2g}{2} )}}{(1-\gamma)},
\]
which lies in $ \overline{\omega}_{D,D} \cap \overline{\omega}_{C,C}$. 
\end{proposition}

In the symmetric subspace, the limiting behavior of the system can be twofold. When $\varepsilon$ is larger than the critical level $\underline{\varepsilon}(g)$, the algorithms converge on a steady state $q_D^{eq}$ where both players find defection to be the preferred action. However, when $\varepsilon$ is below the critical level, there exists additional steady state on the boundary that separates the two regions $\omega_{D,D}$ and $\omega_{C,C}$ --- a pseudo-steady-state according to \Cref{def:steadystate}, where Alice's and Bob's estimates of cooperation and defection coincide.
In this case the Q estimates lie between the long-run value of mutual defection, $ \frac{r(D,D)}{1-\gamma}$, and the long-run value of mutual cooperation, $\frac{r(C,C)}{1-\gamma}$.
Alice and Bob are indifferent between cooperation and defection at $q^{eq}_C$, and they play cooperation (defection) for a fraction $\tau$ ($1-\tau$) of the time.

\begin{corollary}\label{cor:Local time}
In the pseudo-steady-state $q_C^{eq}$ agents spend $\tau_c$ fraction of their time cooperating, where \[\tau = \dfrac{\frac{\varepsilon^2 g}{2}+\varepsilon-2-q^{eq}_C(\gamma-1)(1-\varepsilon)}{2(\varepsilon-1)(1+g+(\gamma-1)q^{eq}_C)} \in \left[\frac{1}{2},1\right].\]
\end{corollary}

The pseudo-steady-state $q_C^{eq}$ corresponds to the imperfect cooperation we observed in the experiments. In particular, the analytic expression for the time spent cooperating approximates its discrete-time experimental counterpart closely, as shown in \Cref{fig:Local Time}.\footnote{In \Cref{app:chaos} we observe that the general asymmetric 4-D system gravitates around a similar equilibrium, where $Q_C = Q_D$. We interpret it as additional support for the restriction to a symmetric subspace.}

\subsection{Sketch of the Proofs}\label{sec:sketch}

With \Cref{manualasm:symmetry}, learning evolves according to this piecewise-linear dynamical system:
\begin{equation*}
    \dot{\mathbf{Q}}(t) =
    \begin{cases}
     A_{C,C}\mathbf{Q}(t) + b_{C,C} \text{ for } \mathbf{Q}(t) \in \omega_{C,C} \\
     A_{D,D}\mathbf{Q}(t) + b_{D,D} \text{ for } \mathbf{Q}(t) \in \omega_{D,D}
    \end{cases}
\end{equation*}

The flows within each $\omega_{a,a}$ characterize the evolution of the Q-functions. However, the system is discontinuous: we would like to preserve continuity of any given flow along the boundary $\overline{\omega}_{C,C} \cap \overline{\omega}_{D,D}$. 
\Cref{prop:Inclusion} below guarantees that we can suitably extend the flows on the boundary, similarly to continuous pasting techniques, such that the flow defined by the fluid limit is globally defined forward in time. 
\begin{proposition}\label{prop:Inclusion}
Let $F_a$ be the field defined as above over $\omega_{a,a}$ for all $a\in\{C,D\}$. There exists a global solution in the sense of \citet{Filippov1988} to the differential inclusion
\begin{align*}
    \frac{d\mathbf{Q}_{t}}{dt} &= F_a(\mathbf{Q}_{t}) \qquad \qquad \qquad \qquad \ \ \ \text{ over }\omega_{a,a} \text { for } a = C,D\\
    \frac{d\mathbf{Q}_{t}}{dt} &\in co\big\{F_a(\mathbf{Q}_t) \big| \ \forall a =C,D\big\} \quad \text{ when }\mathbf{Q}_t \in
    \overline{\omega}_{C,C} \cap \overline{\omega}_{D,D}
    \end{align*}
where $co\{\cdot\}$ denotes the convex hull of a set, and $\overline{\omega}$ is the closure of $\omega$.
\end{proposition}

Both vector fields $F_C$ and $F_D$ are well-defined also on the boundary between $\omega_{C,C}$ and $\omega_{D,D}$. This boundary is called a \emph{switching surface}, because the laws of motion must switch from one field to the other.  Adopting the convention of Filippov, we can define a vector field on the switching surface such that all flows extend to a global solution and such that certain continuous-pasting conditions are satisfied.\footnote{Uniqueness in general is not guaranteed: behavior on the switching surface can lead to multiplicities in the behavior of the system. The law of motion on the switching surface needs only to belong to the convex hull of the laws of motion on either side; by constraning the motion to satisfy certain continuous-pasting conditions we obtain well-defined orbits.} 

Let us call $\normal$ the unit normal vector to the switching surface $\overline{\omega}_{C,C}\cap \overline{\omega}_{D,D}$, so that $\normal \cdot F_a$ is the component of $F_a$ orthogonal to the surface. 
We divide the switching surface in three regions, according to the signs of the normal components of the two vector fields $F_C$ and $F_D$.
A \emph{crossing region} occurs when both normal fields to the boundary are of the same sign: either $F_C$ or $F_D$ can be used to define the law of motion on the surface, since any orbit leaves the switching boundary immediately. A \emph{repulsive region} occurs where both normals face away from the boundary, which will then never be reached; unless initialized here, an orbit will never intersect the repulsive region, thus we do not need to define a field here. Finally, a \emph{sliding region} occurs when both normal components of the two vector fields point towards the boundary. Every flow hitting the switching surface in the sliding region must continue on the switching surface, \emph{sliding along the boundary}. The sliding vector field is defined as a convex combination of the two vector fields $F_{C}$ and $F_{D}$ with parameter $\tau$ such that the normal component to the switching surface vanishes, i.e. $\tau (\normal \cdot F_C) + (1-\tau) (\normal \cdot F_D) = \vec{0}$. The vector field $\tau F_C + (1-\tau) F_D$ is the unique vector field in the convex hull of $F_C$ and $F_D$ whose flow is confined to the switching surface and that satisfies the differential inclusion requirements. \Cref{fig:slidecross} plots the vector fields around the boundary and shows examples of sliding and crossing boundaries with a sample orbit. 

\begin{figure}[hbtp]
    \centering
    \begin{subfigure}{.5\textwidth}
    \centering
    \includegraphics[width=55mm]{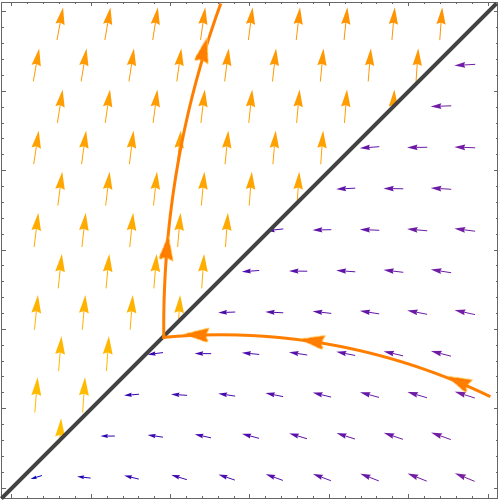}
    \caption{Orbit crossing the boundary.}
    \end{subfigure}%
    \begin{subfigure}{.5\textwidth}
    \centering
    \includegraphics[width=55mm]{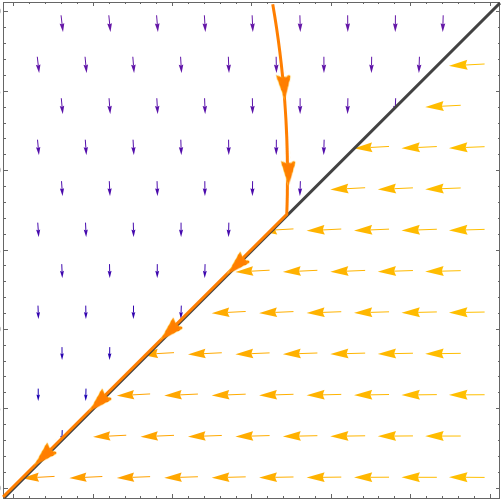}
    \caption{Orbit sliding along the boundary.}
    \end{subfigure}
    \caption{Depiction of two discontinuous flows around a switching surface, in the crossing and sliding case. }
    \label{fig:slidecross}
\end{figure}

\begin{figure}[hbtp]
    \centering
    \begin{subfigure}{\textwidth}
    \centering
    \includegraphics[width=155mm]{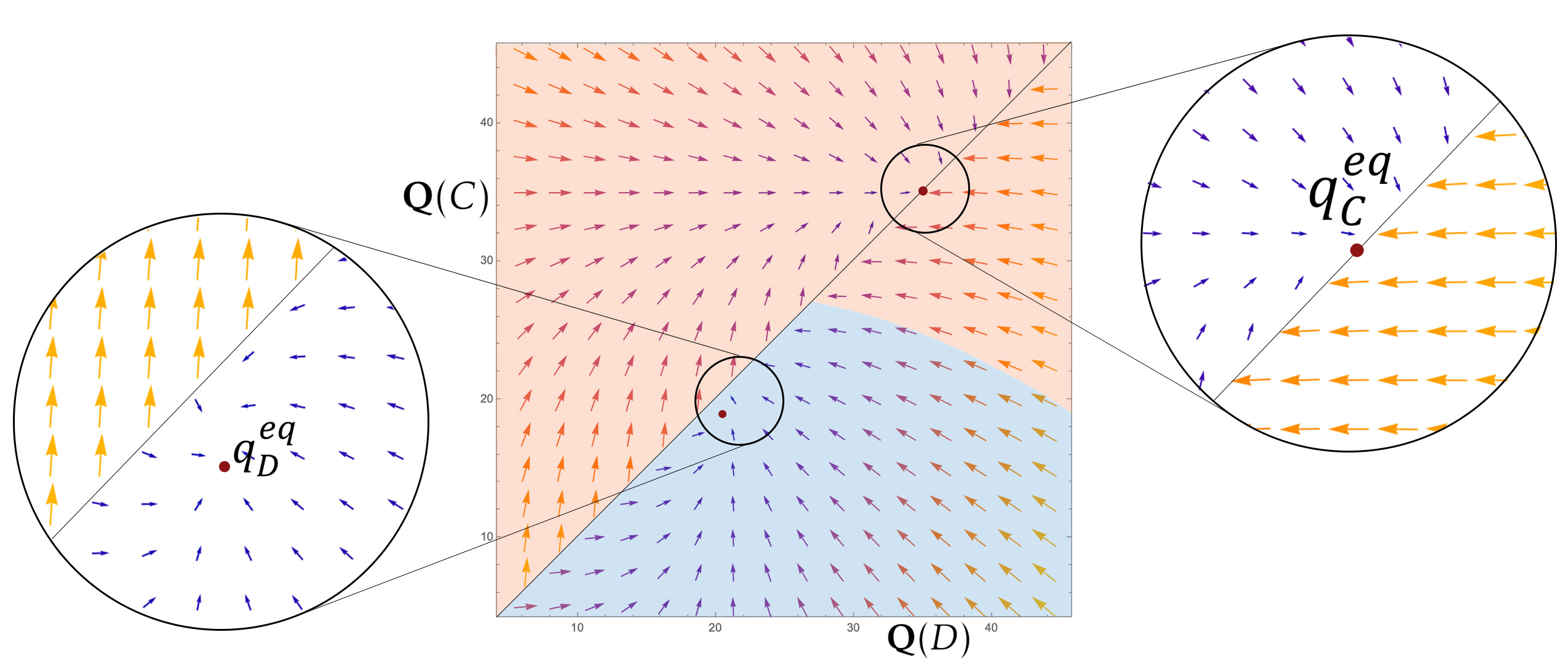}
    \caption{Phase space with $g=1.8$.}
    \label{fig:stability18}
    \end{subfigure}
    \begin{subfigure}{\textwidth}
    \centering
    \includegraphics[width=127mm]{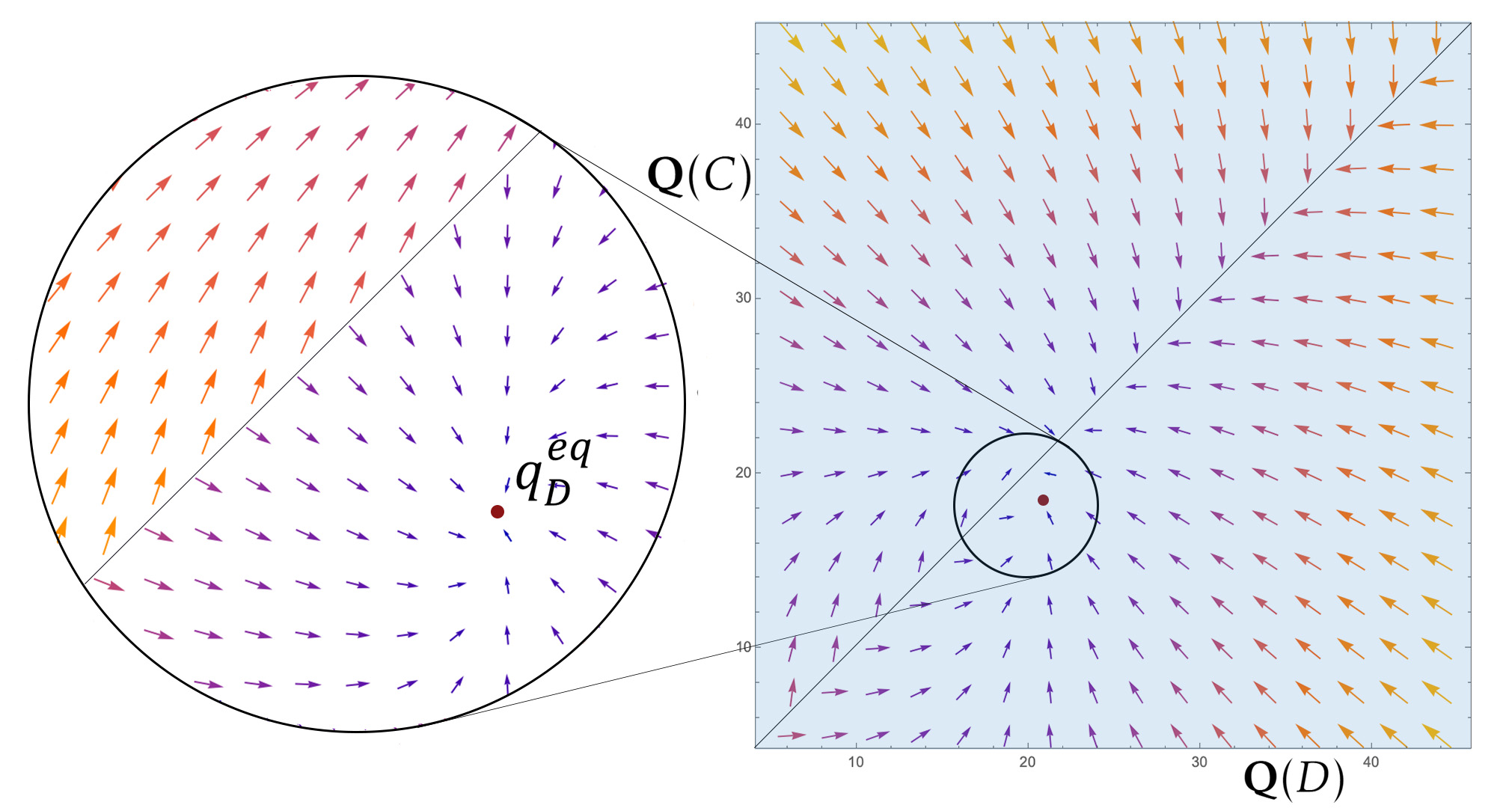}
    \caption{Phase space with $g=1.1$. }
    \label{fig:stability11}
    \end{subfigure}
    \caption{Stationary points are marked with a red dot. The  domain of attraction of the cooperative outcome is green-shaded, and the one for the non-cooperative outcome is blue-shaded. }
    \label{fig:vector fields}
\end{figure}

One can then look for stationary points by finding the solutions to $F_C=0$, $F_D=0$, and $\tau F_C + (1-\tau) F_D=0$. We show that $F_C=0$ never has a solution in $\omega_{C,C}$, $F_D=0$ always have a unique one in $\omega_{D,D}$, and that $\tau F_C + (1-\tau) F_D=0$ has a solution only when $\varepsilon \geq \underline{\varepsilon}(g)$. In \Cref{fig:vector fields} we plot the vector fields that define $\mathbf{Q}$ for two different values of $g$: notice two stationary points when $g$ is large --- the pseudo-steady-state on the boundary disappears for low values of $g$. The analytical characterization of the two points mentioned in \Cref{prop:two eq PD} follows the derivation in \Cref{app:proofs}.

\begin{figure}[hbtp]
    \centering
    \includegraphics[width=0.5\textwidth]{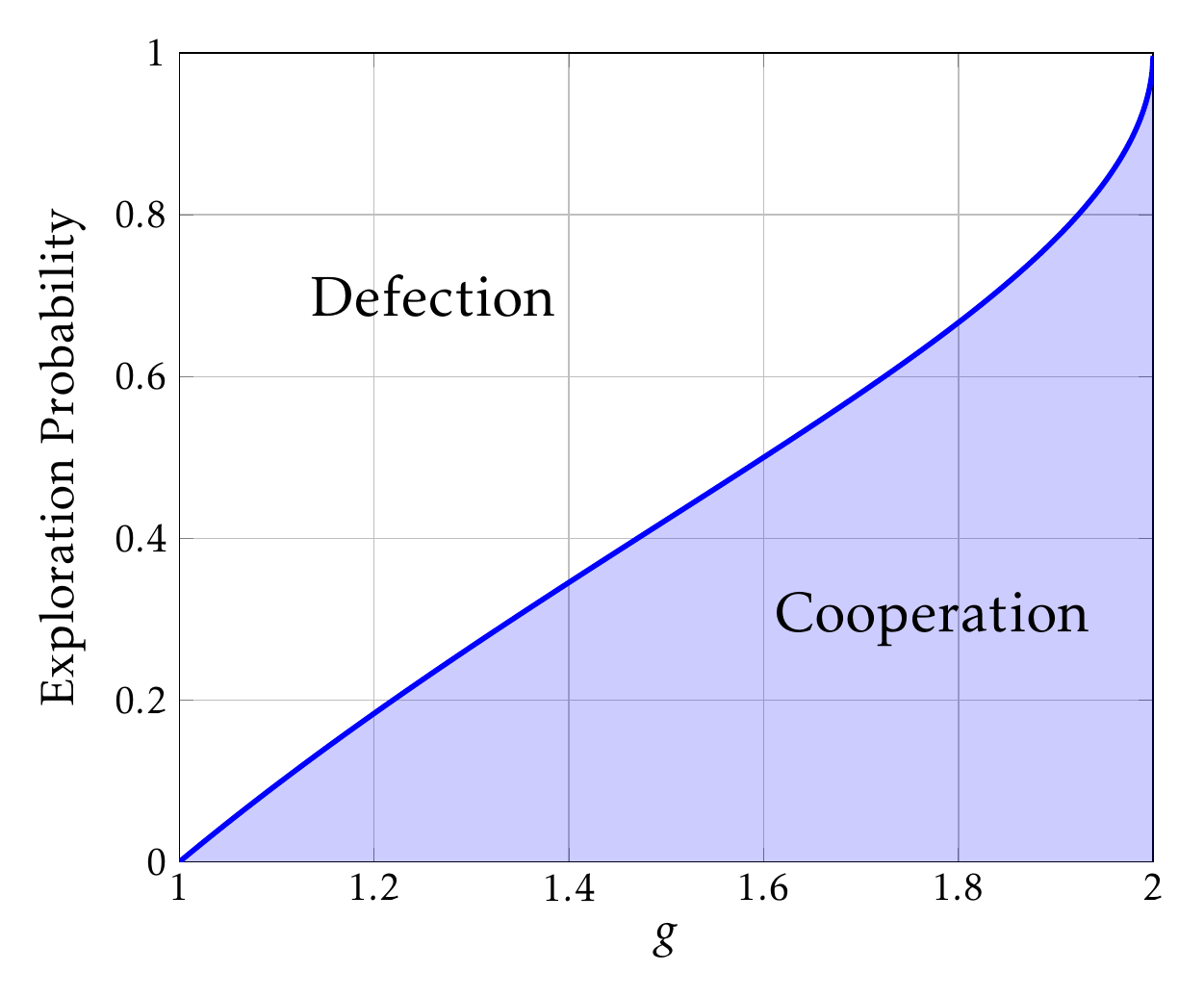}
    \caption{Maximum exploration rate $\underline{\varepsilon}$ to support the cooperative equilibrium, as a function of the growth rate. For all $\varepsilon> \underline{\varepsilon}(g)$ there does not exist a steady-state where both algorithms learn to cooperate.}
    \label{fig:minimum exploration}
\end{figure}

\subsection{Interpretation}\label{sec:coupling}

According to \Cref{prop:two eq PD}, the \emph{cooperative steady-state}, $q^{eq}_C$, exists only when the exploration rate is large enough. 
\Cref{fig:minimum exploration} shows the threshold level of exploration $\underline{\varepsilon}$ that guarantees a cooperative steady-state and how it varies with the value of cooperation. Intuitively, as $g$ increases, the relative benefit of defecting decreases (and vanishes for $g=2$), so more and more exploration is needed to realize that $D$ is a dominant action. For example, if $g=1.8$ the exploration rate required to guarantee convergence on $\{D,D\}$ is about $70\%$, which is considerably larger than the standard employed in practice.\footnote{The literature on Q-learning in games usually employs $\varepsilon=0.1$ or smaller, either fixed or decreasing over time. See, e.g., \citet{Gomes2009}.}

How does Q-learning sustain long-run cooperation, and why does the exploration rate matter? We answer these questions by analyzing the cooperative steady state, and its representations in \Cref{fig:ODE cycles} together with \Cref{fig:Alice,fig:Bob}.\footnote{
Notice that, when simulated with small but discrete time steps, the dynamical system closely mimics the path of play of the discrete Q-learning.}
Suppose $C$ is the preferred action of both players. Alice and Bob cooperate, but with probability $\frac{\varepsilon}{2}$ one defects and is ``surprised'' the unilateral benefit of defecting. Through repeated experiments, over time $Q^i_D$ rises above $Q^i_C$. Suppose this happens first to Alice; as soon as Alice begins defecting, Bob will defect immediately after --- cooperating makes Bob considerably worse off when Alice defects, and his estimate $Q^B_C$ quickly falls below $Q^B_D$. But now mutual defection decreases the value of $Q^i_D$ for Alice and Bob. Instead, experiments are so infrequent that $Q^i_C$ changes very slowly. After a brief phase of joint defection, $Q^i_C$ will dominate $Q^i_D$ again, restarting the cycle.

Effectively, when the exploration rate is low, the AI agents play a symmetric profile of actions too often. This is the type of statistical linkage between algorithm's estimates that we call \emph{spontaneous coupling} --- the estimates of independent Q-learners tend to evolve symmetrically. 
When algorithms start defecting, they experiment with cooperation infrequently, which makes the estimate $Q^i_C$ of cooperation's value persistent. Alice's estimate $Q^A_C$ remains close to the long-run value of mutual cooperation. Instead, her estimate $Q^A_D$ drops dramatically once both agents begin defecting. Alice never realizes the downside of cooperating when Bob defects because both revert to cooperation simultaneously.

\begin{figure}[hbtp]
    \centering
    \begin{subfigure}{0.49\textwidth}
    \includegraphics[width=80mm]{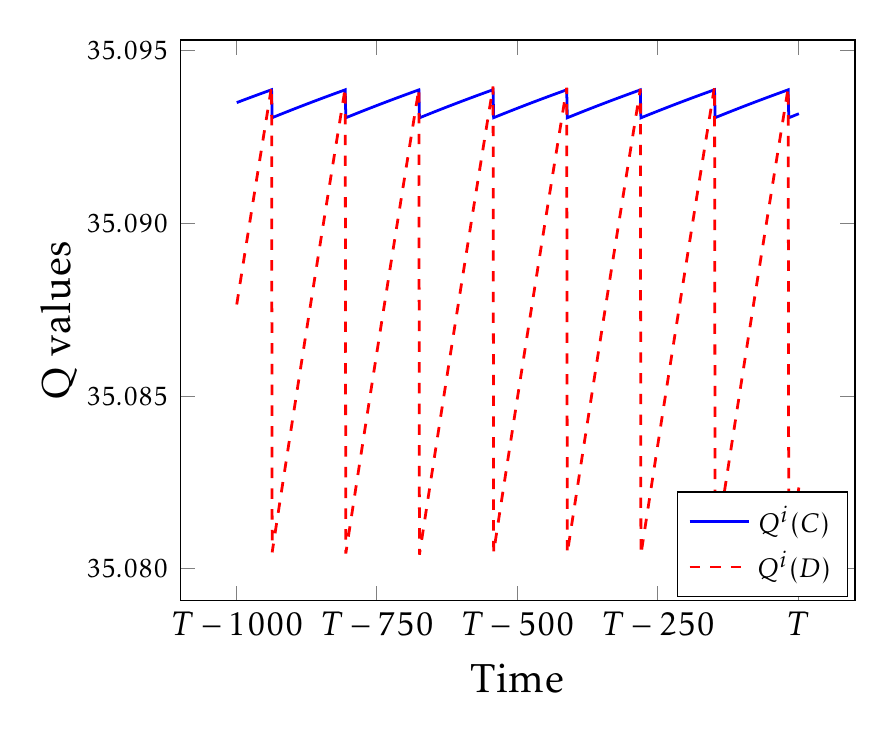}
    \caption{Cycles of play around the cooperative equilibrium in the discretized ODE system.}
    \label{fig:ODE cycles}
    \end{subfigure}
    \hfill
    \begin{subfigure}{0.49\textwidth}
    \centering
    \includegraphics[width=80mm]{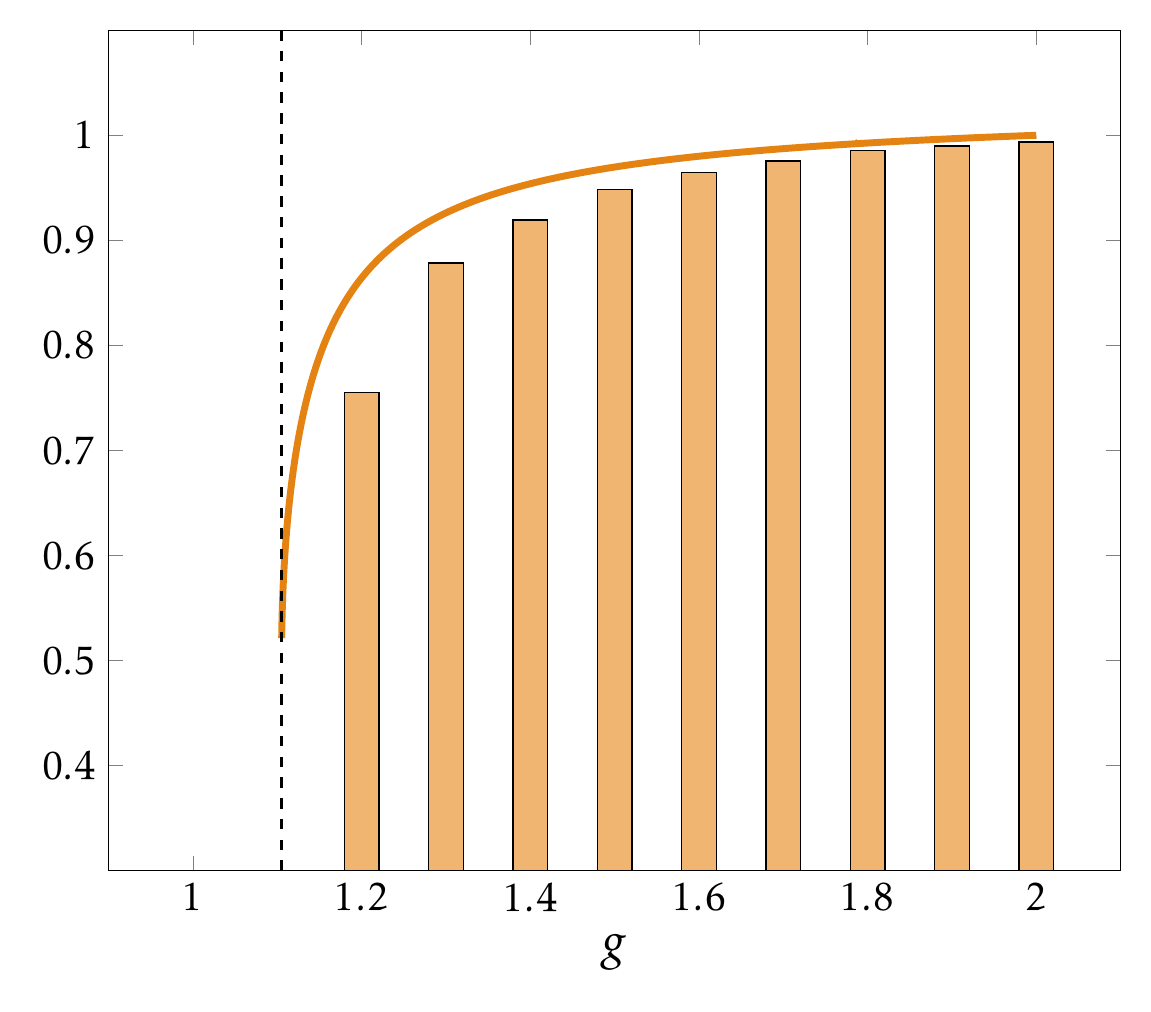}
    \caption{Proportion of time spent playing the cooperative outcome in $q_C^{eq}$, analytically and in the simulations.}
    \label{fig:Local Time}
    \end{subfigure}
    \caption{}
    \label{fig:ODE results}
\end{figure}
The pseudo-steady-state on the boundary is the continuous-time counterpart of these cycles in the discrete system. The discrete cycles tend to a single point in the continuous-time limit. 
In the pseudo-steady-state, agents are ``indifferent'' between cooperation and defection. Both spend some ``local time'' playing either action: we interpret the weights $\tau, 1-\tau$ assigned to the fields on the boundary as the fraction of time dedicated to cooperation and defection, respectively.\footnote{This intuition can formalized using the idea of hysteresis loop around the boundary; see \citet{Dibernardo2008}} The local time is such that the infinitesimal incentives around the stationary point are balanced.

We argue that spontaneous coupling is a form of collusion. In particular, the outcomes of a system undergoing such statistical linkage are indistinguishable from a market with imperfect monitoring in which participants sustain collusion via punishment-reward schemes. Differently from classic stick-and-carrot strategies, however, spontaneous coupling hinges on the algorithmic nature of market participants. A novel incentive that does not rely on monitoring, coupling is, as far as we know, the first collusion scheme to be formalized that is exclusive to Artificial Intelligence algorithms.

We imposed \Cref{manualasm:symmetry} at the beginning of the formal results of this section. While symmetry is necessary to gain analytical tractability, spontaneous coupling appears even without this assumption. For example, recall that the discrete simulated system always exhibits chaotic behavior. Nonetheless, in \Cref{fig:Local Time} the expression $\tau$ appears to fit the data reasonably well. We refer to \Cref{app:chaos} for an analysis of the chaotic system when initialization is not symmetric. There we argue that under chaotic behavior spontaneous coupling displays as a bounded chaotic attractor with similar time-average properties.

\section{General Reinforcers}\label{sec:results}

The mechanics of spontaneous coupling, identified in \Cref{sec:PD}, appear tied to the policy of the reinforcer. We now show that the driving force of inadvertent coordination lies in the algorithm's uneven learning rates across different actions, of which the exploration rate is just a determinant. Learning at different rates about different actions facilitates coordination for general reinforcers. We establish a sufficient condition on the parameters of the algorithm which avoids coupling: learning rates must be uniform across actions.

\subsection{Reinforcers in the Prisoner's Dilemma}
The analysis in the Prisoner's Dilemma of \Cref{sec:PD} brings out a mechanism which sustains collusive behavior between algorithms.
When Alice's preferred action is cooperation, she learns about the payoffs of cooperating at rate $1-\frac{\varepsilon}{2}$, while much more slowly, at rate $\frac{\varepsilon}{2}$, about defection. Essentially, exploration regulates how quickly Alice learns about the value of deviations. When algorithms are coupled, low rates of exploration impair the ability of the algorithm to correctly estimate the value of deviations.

This phenomenon can be formalized without referring to a specific policy. In fact, a given policy only affects the learning rates $\alpha_a$ of different actions. For example, in the case of $\varepsilon$-greedy Q-learning the differential equation within the maximal continuity domain $\omega_{D,D}$ reads as:
\begin{equation*}
    \begin{cases}
    \dfrac{d\mathbf{Q}^A_C}{dt}(t)=\quad \ \pmb{\alpha \dfrac{\varepsilon}{2}} \ \quad \left[\left(1-\dfrac{\varepsilon}{2}\right)g+ \dfrac{\varepsilon}{2}2g+(\gamma-1)\mathbf{Q}_C^A(t)\right]\\
    \dfrac{d\mathbf{Q}^A_D}{dt}(t)=\pmb{\alpha\left(1-\dfrac{\varepsilon}{2}\right)}\left[\left(1-\dfrac{\varepsilon}{2}\right)2 +\dfrac{\varepsilon}{2}(2+g)+\gamma \mathbf{Q}_C^A(t)- \mathbf{Q}_D^A(t)\right] 
    \end{cases}
\end{equation*}
In the language of separable reinforcers, the decision rule of player $i$ affects only $\alpha^i_a(\theta^i)$. In the fluid approximation the only role of the policy is to determine the relative learning rates, since in \Cref{def:regularity} expectation is taken only with respect to distribution over profiles generated by the opponents' policies. Since the absolute magnitude of the learning rates does not matter, we focus on the relative rates.

\begin{definition}
The relative learning rate of action $a^i$ is the ratio
\[RLR(a^i) = \frac{\alpha_{a^i}}{\sum_{a \in A_i} \alpha_a}.\]
\end{definition}

\Cref{thm:RLS} below shows how differences in RLR across actions generally give rise to spontaneous coupling: the coupling appears in a Prisoner's Dilemma for any pair of agents using any (separable) reinforcer. For simplicity, we focus on reinforcers with maximal continuity domains equal to $\omega_{C,C}$ and $\omega_{D,D}$, and constant learning rates $\alpha_C, \alpha_D$. 

\begin{theorem}\label{thm:RLS}
Let each agent learn through the same greedy reinforcer in any Prisoner's Dilemma, and let \Cref{manualasm:symmetry} be satisfied.
There exist an open set $A \subset \mathbb{R}^4_+$ such that for all parameters $\{\alpha_{j}(\omega_{k,k})\}_{j,k = C,D} \in A$ there exists a pseudo-steady-state on the boundary $\overline{\omega}_{C,C} \cap \overline{\omega}_{D,D}$.
\end{theorem}
The proof is constructive: we show that there exist a symmetric tuple of $\alpha$'s such that the sliding vector field on the indifference boundary is null. In particular, we show this when the vector fields on either side of the indifference boundary are opposite, i.e. the local time is $\frac{1}{2}$. We apply homotopical arguments to show that there exists such $\alpha$s, and then we perturb the problem and we obtain the result by continuity. \Cref{thm:RLS} highlights an outstanding fact: even in dominant-strategy-solvable games, reinforcers will not play the dominant strategy for various ranges of parameters.

\subsection{Reinforcers with Uniform Learning Rates}

We provide a simple condition that guarantees reinforcers converge on dominant strategies: reinforcers' relative learning rates must be uniform across all actions. Intuitively, even if algorithms are coupled, uniform learning rates allow agents to evaluate deviations correctly, thus leading them to their dominant strategy.

We first need the following technical assumption:
\begin{manualasm}{A3}[Thickness]\label[manualasm]{manualasm:thickness}
Let $G^t_{-i}(a)$ be the  distribution over actions of all players but $i$ at time $t$. Then, there exists a $\chi>0$ and a $T$ such that for all $t>T$, $G^t_{-i}(a) \geq \chi$ for all $i$ and for all $a \in A_{-i}$.
\end{manualasm}
\noindent Thickness ensures that sufficient exploration is carried out by all players in the limit. For example, any game where all agents adopt a $\varepsilon$-greedy policy satisfies this assumption. More generally, thickness states that each action profile is played with positive (albeit small) probability in the limit. 
\begin{theorem}\label{thm:domstr}
Suppose \Cref{manualasm:thickness} is satisfied for all players.\footnote{We require this assumption because we formulate the Theorem for games solvable by weak dominance. We can instead drop \Cref{manualasm:thickness} when the best-response to any of the opponent's strategies is unique.} In any game with a dominant strategy equilibrium, a reinforcer with $RLR(a^i) = RLR(\tilde{a}^i)$ for all $a^i,\tilde{a}^i \in A_i$ learns the dominant strategy. If the forward limit set of $\theta$ is a singleton, then $\theta$ converges on the dominant strategy.
\end{theorem}

This result is quite general, as we make almost no assumption about the opponent's play: as long as all actions are played with some positive probability even in the limit, the reinforcer will learn to play its dominant strategy. If the game is solvable by strict domination, we do not even require \Cref{manualasm:thickness}. Uniform learning rates ensure any reinforcer will learn the dominant strategy for any number of opponents, as even if they adopt different learning algorithms, or the same learning algorithm with different parameters. The assumption that relative learning rates be uniform across actions may appear stringent: it might for example require restricting the exploration of the algorithm to try each action uniformly at random. Instead, we propose a different strategy to achieve identical RLR across actions.

Consider again Q-learning: the algorithm updates the statistic of action $a^i$ when action $a^i$ is taken and its reward is observed, but it leaves other statistics unattended when the corresponding action is not selected. Suppose however that the agents were able to compute counterfactuals. That is, suppose that, after choosing an action $a^i$ in period $t$, the algorithmic agent was able to back out $r_t(\tilde{a}^i,a^{-i})$ for all $\tilde{a}^i \neq a^i$. Then, the statistics of all actions could be updated simultaneously, using the reward that each action would have procured had it been played in that period. Simultaneous updates are sometimes referred to as \emph{synchronous} learning:\footnote{The term synchronous appears in \citet*{Asker2022}, but the idea of agents learning from counterfactuals is present already in \citet{Tumer2009}.} in this case learning happens at the same rate for all actions. 
The ability to compute counterfactuals affects the learning rates: the second factor in $\alpha^i_a \cdot (\pi_\varepsilon(\theta^i(t)))_a$ disappears when agents update synchronously.
When asymmetric learning rates arise from missing or asymmetric experiments, counterfactual information (i.e., a correct model of the environment) is sufficient to eliminate the asymmetry.
The following corollary formalizes this intuition.
\begin{corollary}\label{cor:counterfactuals}
Under the same assumptions of \Cref{thm:domstr}, a reinforcer who can compute counterfactuals always learns the dominant strategy. If its forward limit set is a singleton, it converges to the dominant strategy.
\end{corollary}
It is natural that counterfactuals help to learn to play equilibria. In fact, the theory of Nash equilibrium is based on the assumption that agents can compute the payoff that would have obtained if they had played a different action, treating the opponents' strategies as fixed. This in turn allows them to evaluate incentives to deviate. \Cref{cor:counterfactuals} establishes that reinforcer algorithms successfully rule out dominated strategies, provided they have access to a method to compute counterfactuals. Reinforcers with counterfactuals will learn to play the (unique) equilibrium.

We can extend the intuition that uniform learning rates ensure that players correctly estimate the value of deviations to more general setting. This is particularly relevant for pricing games that do not have equilibria in dominant strategies. Let us consider the procedure of iterated elimination of strictly of dominated strategies (IESDS). However, we rescrict deletion to strategies strictly dominated by another \emph{pure} strategy, because reinforcers do not deal well with mixed strategies.\footnote{\Cref{def:reinforcer} makes clear that it is impossible for reinforcers to learn the value of randomizing across actions.}
\begin{definition}
We say that an action $a^i \in A_i$ is \emph{pure-rationalizable} if there is an order of IESDS such that $a^i$ survives the IESDS procedure.
\end{definition}
In general, for a certain order of deletion of dominated strategies action $a^i$ might get eliminated. However, as long as there is an order such that $a^i$ survives the IESDS process, we consider $a^i$ pure-rationalizable.
Our next theorem shows that reinforcers with access to counterfactuals only play pure-rationalizable strategies in the limit.

\begin{theorem}\label{thm:IESDS}

Let all players in game $G$ learn through a reinforcer using a $\varepsilon$-greedy policy with $RLR(a^i) = RLR(\tilde{a}^i)$ for all $a^i,\tilde{a}^i \in A_i$ for all $i \in N$. Assume $\varepsilon>0$ is small enough so that the reward's order is preserved, i.e. if $\underline{a}^{-i}$ is the preferred profile of actions of agent $i$'s opponent, $\mathbb{E}_{\pi^{-i}}[r(a^i,a^{-i})]>\mathbb{E}_{\pi^{-i}}[r(\tilde{a}^i,a^{-i})]$ when $r(a^i,\underline{a}^{-i}) > r(\tilde{a}^i,\underline{a}^{-i})$.
Then, all actions learned by the players are pure-rationalizable in the game $G$ under the same IESDS order.

\end{theorem}

\Cref{thm:IESDS} implies that $\varepsilon$-greedy reinforcers with uniform learning rates always learn to play Nash equilibrium strategies in a supermodular game with a unique equilibrium. More generally, in any pure-dominance-solvable game, reinforcers will learn the pure-strategy Nash equilibrium.

We can interpret relative learning rates as the relative ability to work out counterfactuals for a given action. Because the utility of a given action depends on the opponent's path of play, uneven learning rates generate biased estimates. Small relative learning rates fail to account for asymmetric play, impairing the ability of the algorithm to best-respond. Uniform learning rates instead guarantee correct counterfactual estimates. With unbiased counterfactuals, abandoning dominated strategies is immediate. In the following sections we show how our results can be used to understand and design markets where AI agents operate.

\section{Application: Price Fixing}\label{sec:bertrand}
In this section we consider a prototypical example of price competition between learning algorithms. In particular, we look at the setting studied by \citet*{Asker2022} and show how the results they obtained in simulation experiments can be given theoretical foundations. The authors simulate algorithmic competition in a Bertrand oligopoly, and find that the emergence of collusion depends critically on what the authors call ``synchronicity'' of the algorithm; their conclusion is in fact a consequence of \Cref{thm:domstr,thm:IESDS}, which validates our approach based on fluid approximations.

Consider first a simplified version of price competition. There are two firms, Alice Inc. and Bob Ltd., that face a common demand for their product. Assume that the market demand is $D(p_A,p_B) = 3-\min\{p_A,p_B\}$, and if the two firms charge the same price they split demand equally. Suppose that each firm has $0$ marginal cost, and for simplicity let the firms choose only between two prices: $p \in \{ 0.5,2\}$. Profits equal price times individual demand. This Bertrand game has only one static Nash equilibrium, the profile $\{0.5,0.5\}$, since posting the lower price is dominant.
\citet*{Asker2022} assumes that the two firms learn using Q-learning and they consider two variations, both of which are \emph{greedy}, i.e., the action taken is always the one with the highest estimated value.
\begin{itemize}
    \item [(i)] Asynchronous Greedy Q-learning: the algorithm updates only the Q-value of the action taken in each period; 
    \item [(ii)] Synchronous Greedy Q-learning: the algorithm updates all Q-values in each period, with the return that it could have obtained had he played the other action instead, but holding the opponent's action fixed. 
\end{itemize}

It is clear that the syncronous greedy Q-learning is a reinforcer with uniform relative learning rates, since the Q-values of both prices are updated at every time step. Applying \Cref{thm:domstr}, it is then immediate to deduce that in this case the two learning firms should converge to posting the lower price, as concluded by \citet*{Asker2022}. Asynchronous learning can be regarded as the opposite situation: the relative learning rate of a price is either 0 or 1. In this case we expect spontaneous coupling to occur, and indeed \citet*{Asker2022} finds supra-competitive prices. 

For this simple model with two prices we can carry out a detailed analysis similar to \Cref{sec:PD}, which shows that the same mechanism at play there is also driving the outcomes in this setting. \Cref{fig:bertrand} plots the vector fields obtained applying \Cref{thm: fluid approx thm} to this model assuming symmetric initialization. With asynchronous learning there are two stationary regions. For values of $\mathbf{Q}_2 \leq \mathbf{Q}_{0.5} = \frac{5}{8} $ the algorithms converge on the competitive outcome; however, for $\mathbf{Q}_{0.5}\leq \mathbf{Q}_2 = 1$ the algorithms collude. These results are robust to $\varepsilon-$greedy exploration: the \emph{spontaneous coupling} introduced in \Cref{sec:PD} again sustains collusion. Because all observations of the returns from the supra-competitive price are obtained when colluding, $\mathbf{Q}_2$ remains consistent with mutual collusion also during a competitive phase, which brings the system back to collusion. Instead, if both firms adopt synchronous learning, there is only one stationary point, at $\mathbf{Q}_{0.5} = \frac{5}{8}, \mathbf{Q}_2 = 0$, and it is a global attractor. When the two firms are colluding, all arrows point upward: the algorithms correctly estimate the value of a one-shot deviation, without internalizing the effect that defecting from a collusive outcome will have on returns in the future. Once the firms begin competing it is then impossible to revert back to collusion: the counterfactual return of a deviation is zero, and since the relative learning rates are equal $\mathbf{Q}_2$ is bound to remain lower than $\mathbf{Q}_{0.5}$. In other words, the coupling disappears and therefore the value of joint collusion is short-lived after competition begins. 

\begin{figure}[h]
    \centering
    \begin{subfigure}{0.49\textwidth}
    \includegraphics[width=80mm]{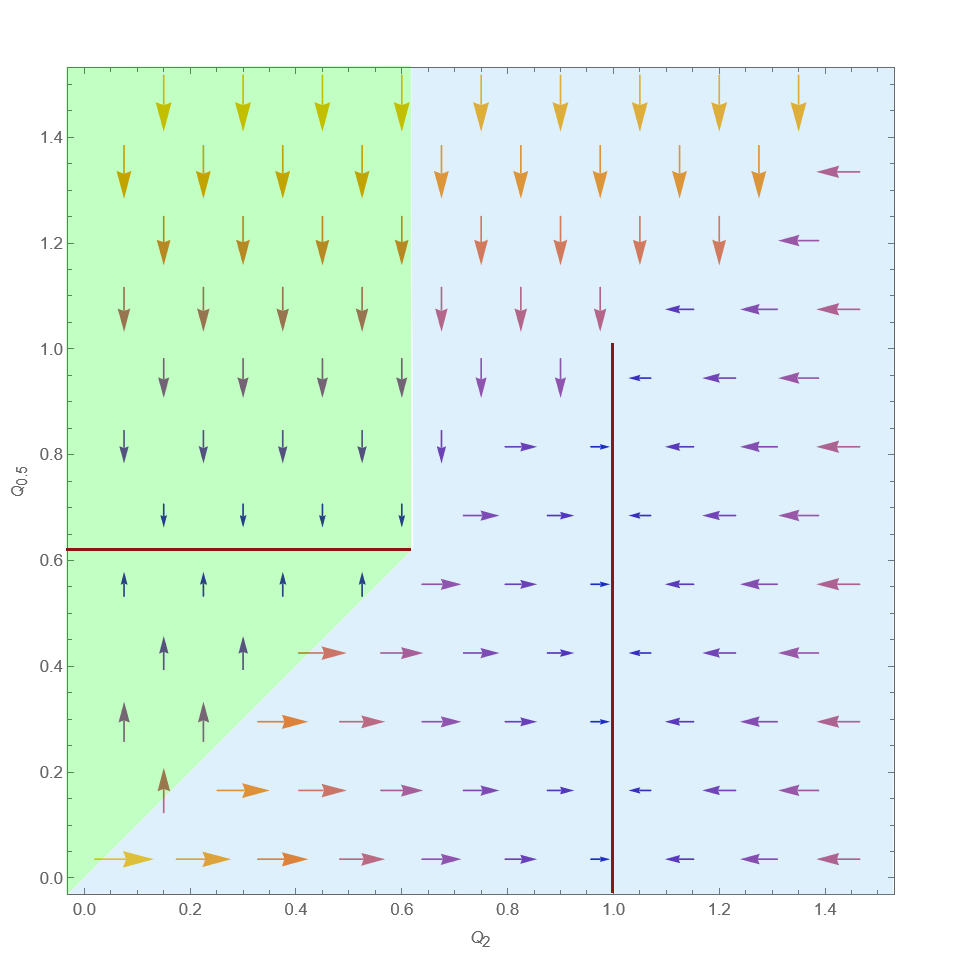}
    \caption{Asynchronous updating.}
    \label{fig:asynch}
    \end{subfigure}
    \hfill
    \begin{subfigure}{0.49\textwidth}
    \centering
    \includegraphics[width=80mm]{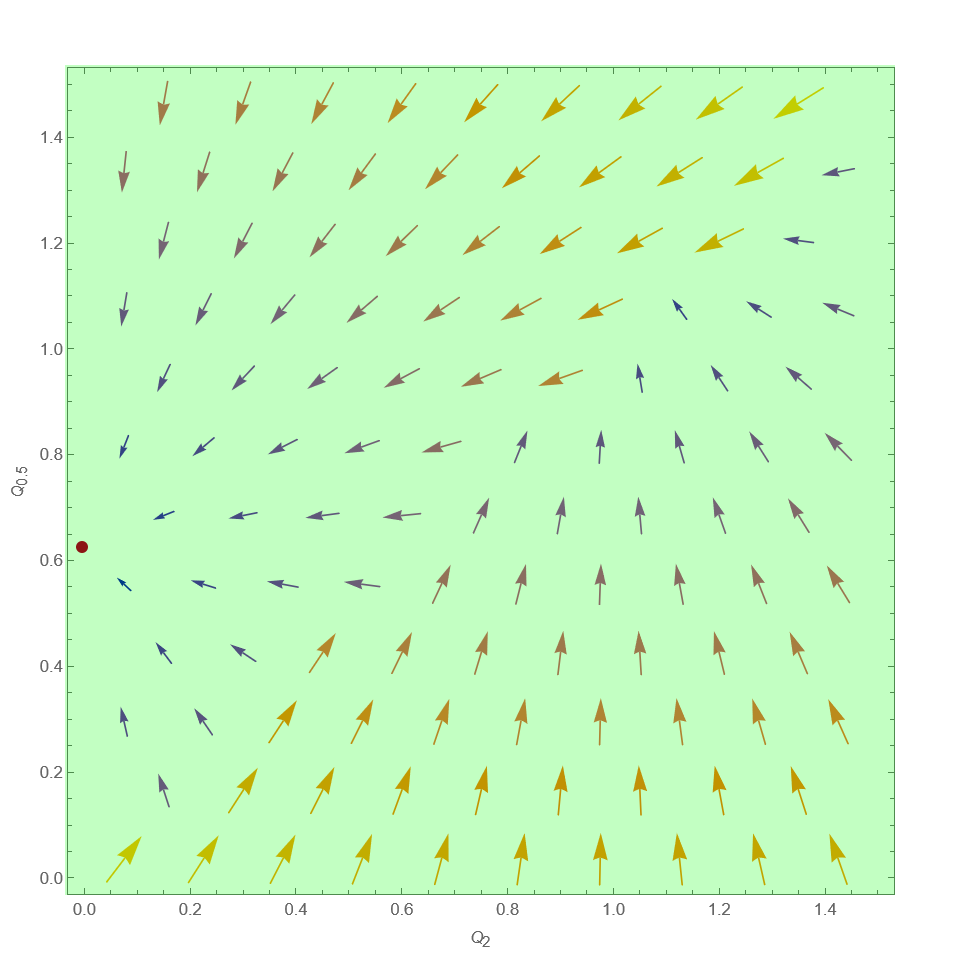}
    \caption{Synchronous updating.}
    \label{fig:synch}
    \end{subfigure}
    \caption{On the vertical axis is Q-value of the low price, and on the horizontal axis the value of the high price. The green-shaded area denotes the domain of attraction of the competitive outcome, while the blue-shaded area is the domain of attraction of the collusive outcome. The red dot and red lines are the equilibria of the systems. Obtained with $\gamma = 0$ (reflects the specification of \citet*{Asker2022}).}
    \label{fig:bertrand}
\end{figure}

\subsection{General Bertrand Competition} The simple model above reduces the Bertrand game to a dominant-strategy game. 
It is a convenient simplification for the purposes of inspecting and plotting the dynamical systems, but the theory developed in \Cref{sec:results} allows us to deal with more general models. The following is the specification from \citet*{Asker2022}.

Alice Inc. and Bob Ltd. have now constant marginal costs $c_A = c_B = 2$. They sell homogeneous goods and compete by setting prices. The set of feasible prices is composed of $100$ equally spaced numbers between $0.01$ and $10$, inclusive. The set of prices is denoted by $P = \{p^1,\dots, p^{100}\}$. Consumers buy from the firm with the lowest price, and demand is parametrized as 
\[d_i(p_i,p_{-i}) = \begin{cases}
        1 \text{ if } p_i < p_{-i} \text{ and } p_i \leq 10 \\
        \frac{1}{2} \text{ if } p_i = p_{-i} \text{ and } p_i \leq 10 \\
        0 \text{ otherwise }\\
\end{cases}\]
As the authors note, there are two Nash equilibria of this game, one ($E_1$) with $p_A = p_B = 2.0282$ and one ($E_2$) with $p_A = p_B = 2.1291$. The multiplicity is a consequence of the discretization of the space in equally spaced prices. \Cref{thm:IESDS} allows us to deal with settings where there is no dominant strategy, and in particular we can deduce that also with 100 prices, if the firms employ synchronous learning, they can only converge to competitive pricing. Moreover, the following proposition also suggests that the parameters of the algorithms are irrelevant for the convergence result.

\begin{proposition}\label{prop:bertrand}
In a Bertrand oligopoly, if Alice Inc. and Bob Ltd. adopt any $\varepsilon$-greedy separable reinforcers with a small $\varepsilon>0$ such that the relative speed of learning is the same across all prices, they will learn to play either $p_1 = 2.0282$ or $p_2 = 2.1291$.
\end{proposition}
\begin{proof}
This proposition follows almost immediately from \Cref{thm:IESDS}. The Theorem guarantees that two $\varepsilon$-greedy reinforcers will learn a pure-rationalizable strategy. Discretized homogeneous Bertrand games have only two pure-rationalizable strategies, the two lowest prices above marginal cost, which are also the game's Nash Strategies.
\end{proof}

\section{Application: Market Division}\label{sec:keyword}
We made the argument in \Cref{sec:results} that spontaneous coupling is not a collusive behavior per se, but is instead a mechanism that may underpin a broad set of market manipulations by algorithmic agents. In support of this claim, in this section we show that spontaneous coupling can stifle competition by sustaining an anticompetitive conduct known as \emph{market splitting}. With this conduct, market participants coordinate to concentrate each on a subset of the market and decline to participate in others, so that effectively each competitor is a monopolist in their reference market.  We study a model of search advertising, where competing advertisers submit bids for keywords of various values. We find that advertisers learn not to bid on their competitor's favorite keyword, thus implicitly splitting the market. The outcome is supported by the spontaneous coupling trap, which emerges endogenously from the learning process.

Consider a market where two advertisers, Alice.net and Bob.com, compete for ad slots on three distinct keywords: $a$, $m$, and $b$; there is one slot available for each keyword. Each advertiser $I$ derives a certain value from a click on their ad, denoted by $v^I_j \sim U[1,2]$, where $j \in \{a,m,b\}$ denotes the keyword. All values $v^I_j$ are re-drawn independently in every auction. 
We assume that the probability of a click varies depending on the keyword. Specifically, an ad on an advertiser's ``branded'' keyword yields a higher click-through rate (CTR) than an ad on a neutral or their competitor's branded keyword. We represent this variation by specifying advertiser-keyword specific click-through rates (CTR), denoted by $ctr^I_j$ for advertiser $I$ in keyword $j$. We choose $ctr^A_a = 1$, $ctr^A_m = 0.6$, and $ctr^A_b = 0.2$ for Alice.net, and symmetrically for Bob.com.\footnote{Keywords denoted by $a$ and $b$ are "branded" for Alice.net and Bob.com, respectively. Think of a query that communicates a clear intent to reach Alice.net's website. The query denoted by $m$ instead is neutral --- it may lead to a click on either advertiser with the same probability.}

The ad slots are allocated via (separate) second-price auctions with a reserve price of $1$: Alice.net and Bob.com may submit their values $v^I_j$ for all three keywords, and the highest bidder for the $j$ auction wins the ad slot for the keyword $j$. The winner pays $t^I_j$ per click, where $t^I_j$ the second highest bid or the reserve price if they were the only bidder. 
The expected payoff for the winner is given by their CTR multiplied by the difference between their bid and their payment, $ctr^I_j \times (v^I_j - t^I_j)$. The loser receives a payoff of 0. If an advertiser does not submit a bid, they receive a payoff of 0. Notice that bidding on all three keywords is a dominant strategy in this game. This is because the payoff for winning a keyword is always positive, while the payoff for losing a keyword is 0, and there is no budget constraint.

Within the setting described, we let advertisers A and B use a learning algorithm to determine which keywords to bid on. Specifically, we assume that A and B implement a Q-learning algorithm that selects a subset of keywords to bid on. Note that the algorithm does not select a bid: we assume that when agents bid on a given keyword they bid their value. Thus, the action space for this algorithm is defined as the power set of the set of available keywords, which results in a total of $2^3$ possible keyword combinations.

\begin{figure}[h]
    \centering
    \begin{subfigure}[b]{.49\textwidth}
    \centering
    \includegraphics[width = 80mm]{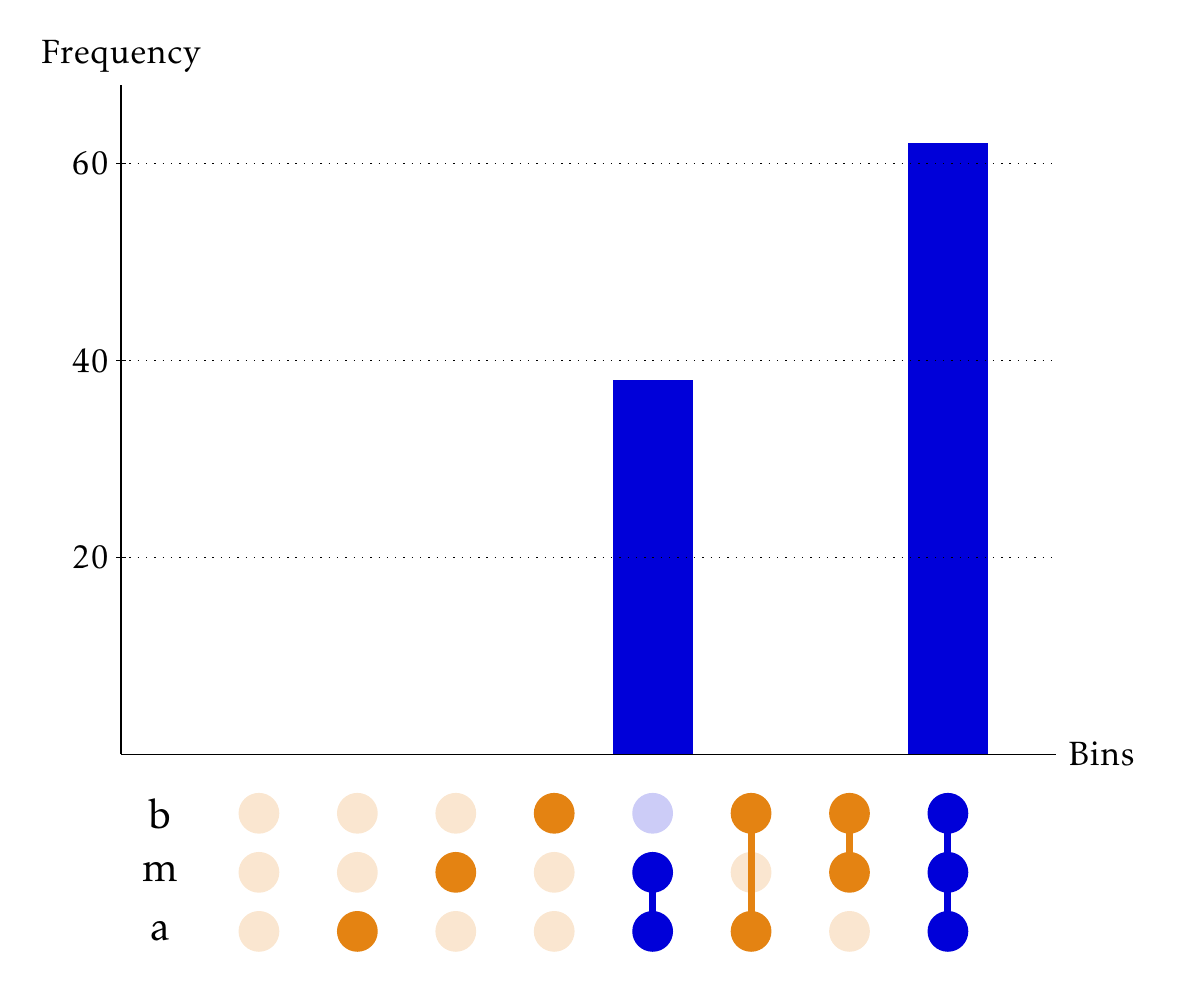}
    \subcaption{Alice's play}
    \end{subfigure}
    \hfill
    \begin{subfigure}[b]{.49\textwidth}
    \centering
    \includegraphics[width = 80mm]{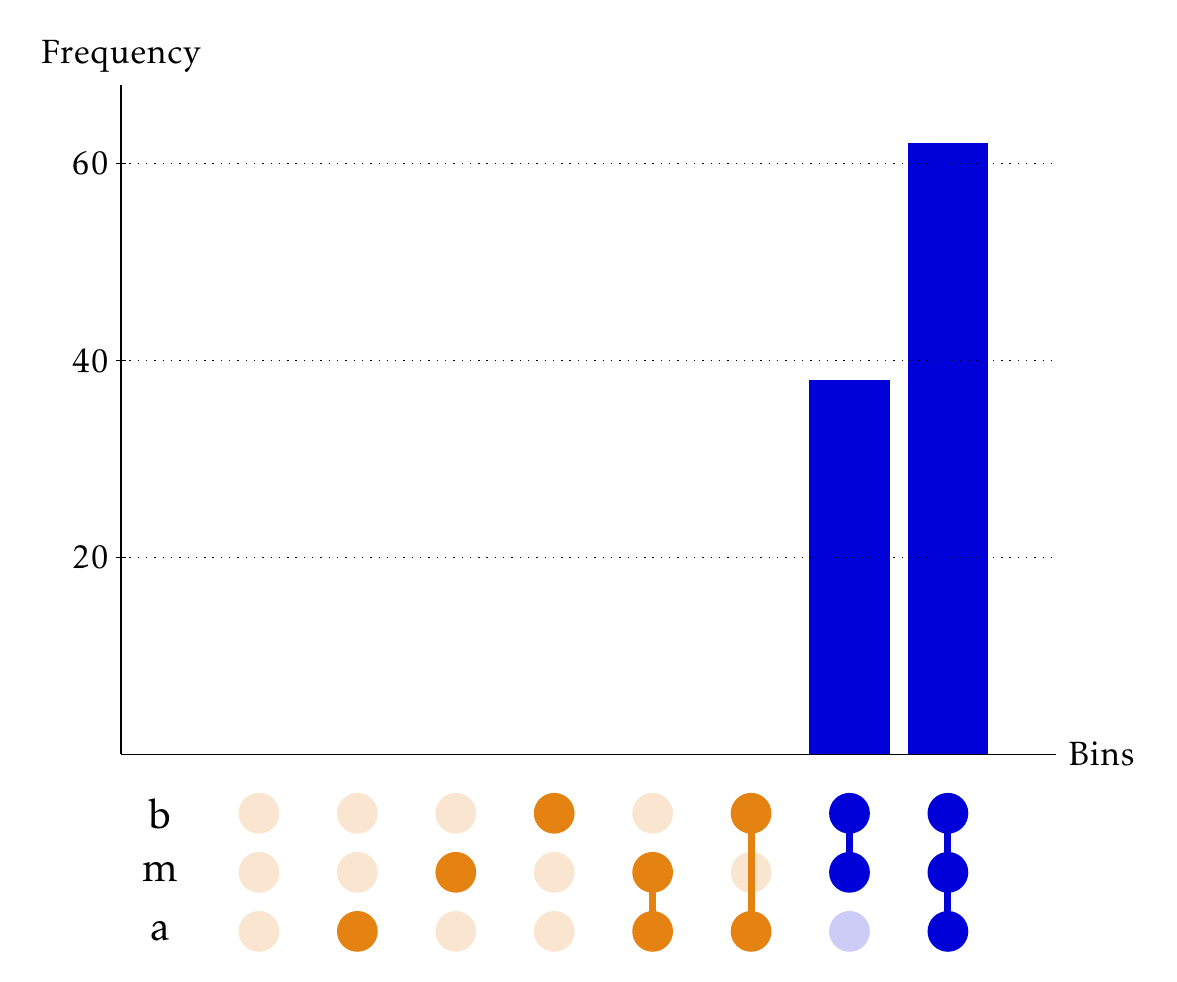}
    \subcaption{Bob's play}
    \end{subfigure}
    \caption{The figure displays a frequency histogram of independent, randomly initialized learning trajectories that converge on a given strategy. The upper bar graph represents the fraction of these trajectories that converge on each strategy, while the lower graph displays all possible advertisers' strategies.}
    \label{fig:keywords}
\end{figure}

To investigate the dynamics of the Q-learning algorithm, we simulate its behavior and visualize the results in \Cref{fig:keywords,fig:Qhistory}.
In \Cref{fig:keywords}, we observe that both algorithms exhibit convergence towards the dominant strategy in more than $50\%$ of the independent, randomly initialized learning trajectories. However, in 18 out of 50 simulations, the two advertisers learn to collude by splitting the market.
Under this collusion scheme, each advertiser only bids on his own branded keyword and the neutral one. In this collusive scheme, reminiscent of collusion in spatial models, neither advertiser bids on the opponent's branded keyword, despite the fact that it is strategically suboptimal. It is important to notice that the outcome of collusion by market splitting is Pareto dominant: both players achieve better outcomes than if they bid on all keywords. In fact, bidding on all keywords yields an expected payoff of $\frac{1}{6}\big(1+\frac{6}{10} + \frac{2}{10}\big) =\frac{3}{10}$ per round, whereas market splitting yields an expected payoff of $\frac{1}{2} + \frac{1}{6}\frac{6}{10} =\frac{6}{10}$ per round.

\begin{figure}[h]
    \centering
    \begin{subfigure}[b]{.49\textwidth}
    \centering
    \includegraphics[width = 7cm]{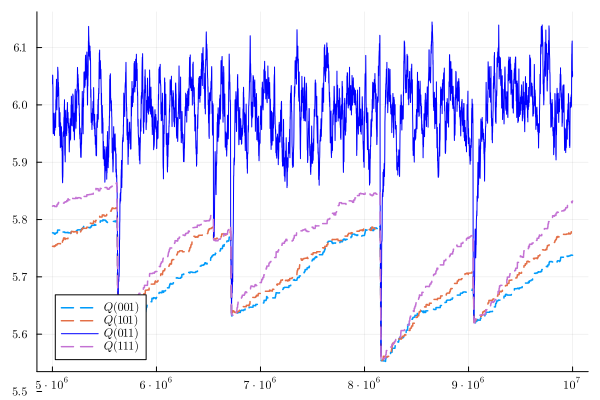}
    \end{subfigure}\hfill
    \begin{subfigure}[b]{.49\textwidth}
    \centering
    \includegraphics[width = 7cm]{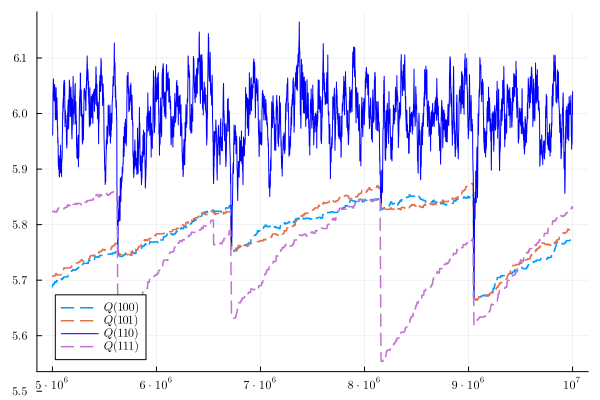}
    \end{subfigure}
    \caption{Learning trajectories from a simulation with $\gamma = 0.9$, $\varepsilon = 0.01$, $\alpha = 0.005$. The solid line represents $Q(m,a)$ and $Q(b,m)$ on the left and right, respectively. The dotted lines represent the $Q$ values associated with some bundles containing the opponent's keyword.}
    \label{fig:Qhistory}
\end{figure}

We now focus on \Cref{fig:Qhistory} to examine in greater detail the collusive dynamics. This figure portrays the evolution of the Q-learning algorithm for a selection of actions taken by each advertiser. As the figures reveal, the collusion behavior is not perfect. Although most of the time the two advertisers engage in collusive splitting, experimentation eventually brings about the realization that they could benefit by bidding on the opponent's branded keyword. As soon as they realize this, one of the advertisers begins bidding on the opponent's branded keyword, causing a sudden drop in the opponent's payoffs, which leads her to abandon the splitting strategy and revert to bidding on all keywords. This sequence of events causes both players to experience significantly lower payoffs, but their expectations (their $Q$) about the market-splitting outcome persistently remain high.
Ultimately, both advertisers return to splitting the market simultaneously in a stochastic recurrent cycle typical of spontaneous coupling, as described in \Cref{sec:results}.

\section{Application: Learning-Robust Mechanism Design}\label{sec:learnable}

The previous sections demonstrated that reinforcers can fail to learn to play their dominant strategy. Besides being a cornerstone of game theoretic analysis, dominant strategies are also fundamental in mechanism design; given the widespread adoption of mechanisms in the online economy, the failure of reinforcers is all the more concerning. However, \Cref{thm:domstr} and \Cref{cor:counterfactuals} provide a solution that can be embedded in mechanism design problems: providing counterfactual feedback may avoid spontaneous coupling's occurrence. We are thus interested in designing a dominant strategy mechanism with a feedback rule that guarantees that if players use reinforcers and update their estimates according to their feedback, they will learn the dominant strategy. Moreover, we are interested in finding the minimal feedback necessary to accomplish this goal, because we are concerned with unintended consequences of providing the algorithms more information than necessary about the play of the other players.

Consider a canonical model of implementation with private information. Let $\mathcal{X}$ be the set of possible outcomes, and let there be a set $N$ of agents with types $ (\lambda^i)_{i\in N} \in \varprod_{i\in N}\Lambda_i = \Lambda$ fixed over time.\footnote{It is simple to extend this result to allow for types drawn i.i.d. in every period, but for simplicity we stick to a constant type in this section.} Type $\lambda^i$ determines agent $i$'s preferences $u^i \colon \mathcal{X} \times \Lambda_i \to \mathbb{R}$ over outcomes.
A direct revelation mechanism requires each agent to report a type $\hat{\lambda}^i$. The mechanism then maps the reported type profile $\hat{\lambda}$ to an outcome, $f(\hat{\lambda})$.
We say a mechanism is \emph{strategy-proof} if it is a direct revelation mechanism and reporting truthufully is a dominant strategy:

\[u^i(f(\lambda^i,\lambda^{-i}),\lambda^i) \geq u^i(f(\hat{\lambda}^i,\lambda^{-i}), \lambda^i) \qquad \forall \hat{\lambda}^i \neq \lambda^i.\]

Assume further that a subset of agents $L\subseteq N$ act according to the choices of their own reinforcer. Agents in $L$, learners, assess the value of each individual report over multiple iterations of the mechanism. Agents in $N \setminus L$ are instead rational and myopic. That is, they rationally play their static dominant strategy instead of trying to manipulate the learning agents. In each period both rational and learning agents choose a report. Once the mechanism selects the outcome, payoffs realize and learners update their estimates. We call this setting a \emph{hybrid market}, because rational and learning agents coexist.

If the mechanism is strategy-proof, myopic rational agents play their dominant strategy and report truthfully. However, as we have seen in \Cref{sec:PD,sec:results} coupling between independent algorithms may lead to behavior different from dominant strategy, consistent with collusive agreements, so that simple strategy-proof mechanisms may fail the designer. We thus seek learning-robust strategy-proof mechanisms (LRSM).
\begin{definition}
Suppose agents in $L$ adopt a (separable) reinforcer.
A strategy-proof mechanism is \emph{learning-robust} if each agent $l \in L$ learns the truthful report $\lambda^l_\text{truthful}$.
\end{definition}

We assume that the designer can provide information to the participants of the mechanism, and that such information is used by the algorithms to make inference about payoffs. Then, our results in \Cref{sec:results} show that the designer can ensure robustness of a mechanism by supplying enough information to the reinforcers so that they can evaluate counterfactuals.\footnote{Recall from \Cref{cor:counterfactuals} that updating the estimates $\theta_a$ according to counterfactual information enforces uniform learning rates, but this is only a sufficient condition.} 
The designer assists the algorithms in their counterfactual calculations by revealing some private information. We refer to this ex-post revelation as \emph{feedback provision}. For any given a strategy-proof mechanism $f$, we look for a LRSM for $f$, that is, an ex-post feedback policy which allows agents to compute counterfactuals.

A feedback policy for agent $i$ is essentially a partition of the space $\Lambda_{-i}$ of opponents' types. After the mechanism, the designer communicates to each agent what element of the partition the opponents' reports belonged to. The partitions must be such that agents can compute what outcomes they could have enforced by reporting differently.

Of course, a designer can always opt for a full revelation feedback policy.\footnote{The full revelation policy is the finest partition of $\Lambda_{-i}$.} Revealing everyone's report after the mechanim allows algorithms to compute payoffs from every report, but it may induce additional concerns. First, insisting on revealing all private information may facilitate tacit or explicit collusion, and, second,it may result in large communication costs. Finally, it is not necessarily true that, when provided with all reports, computing the allocation is a simple task. In certain combinatorial auctions, translating reports to prices and allocations requires solving a complex optimization problem.

To address some of these concerns, we look for LRSM with reduced communication burden. Formally, we define a privacy order on the space of feedback policies. We can show that this (partial) order is a lattice, and thus there exist both a minimally- and a maximally-private feedback policy. The former is the full revelation policy, consistent with our intuition. We characterize the latter: a policy that communicates just enough information to ensure that each agent can compute its counterfactuals, and nothing more. It turns out that such feedback is informationally-equivalent to the well-known \emph{menu} formulation of the mechanism (\citet{Hammond1979}).

\begin{theorem}\label{thm:LRSM}
Let $f$ be a strategy-proof mechanism. Then,
\begin{enumerate}
\item There always exists a LRSM for $f$,
\item The maximally-private LRSM for $f$ is a menu description.
\end{enumerate}
\end{theorem}

Menu descriptions are the ex-post feedback counterpart of menu mechanisms. \citet{Hammond1979} defines menu mechanism as providing each agent with a menu, which depends on the profile of reports of the opponents, and let the agent choose his preferred outcome. Instead, we provide feedback in the form of menus: after having received all reports and implemented the outcome, the designer sends a menu, also dependent on the profile of reports of the opponents, that lists what other outcomes would have occurred if the agent had report differently, so that algorithms can compute what would have happened had they reported a different type.

The feedback of menu descriptions is an aggregator of market information, which helps agents evaluate the true value of truthful reporting.\footnote{The recent \citet*{Gonczarowski2023} discusses the simplicity properties of menu descriptions in experiments with human participants. Here, we argue that menu descriptions may facilitate strategy-proof play from algorithmic participants as well.}  
\citet{Parkes2004} argues that mechanism design can play an important role in shaping algorithmic systems. The author describes \emph{learnable} mechanism design --- the idea of explicitly designing mechanisms to maximize and improve performance considering the agent's adaptive behavior. As he suggests, \emph{``a useful learnable mechanism would provide information, for example via price signals, to maximize the effectiveness with which individual agents can learn equilibrium strategies''}. 
We formalized this intuition by showing that feedback design can make traditional strategyproof settings robust to adaptive algorithmic players by providing price signals, or menu descriptions.

Finally, note that the world of online auctions has partly begun to provide feedback to its participants. Google's auctions for display advertising provide feedback, in the form of a ``minimum bid to win'', after each auction has concluded. The minimum bid to win is indeed the menu mechanism for a single-item allocation problem. In the next subsection we will show what a menu description would look like in a simple VCG auction for online advertising, such as one used by Yandex.ru for their search advertising business.

All formal statements and proofs in this section are presented in \Cref{app:feedback}.

\subsection{VCG for Online Search Advertising}
Consider the following simple model of search advertising. There are $N$ bidders for a given query, each with their own value $v_i \in \{0,0.1,\dots,10\}$ for each click. Without loss we can order bidders by their valuations: $v_1 \geq v_2 \geq \dots v_N$.
The search site offers two ad slots. The first ad slot will bring a predicted traffic volume of $100$ units, while the second ad slot will bring in only $80$ units of traffic. The search site is running VCG: all players submit one bid each, representing the price they would pay for each click, and the winners of the two slots will be the agents with the two largest bids. Let us assume for simplicity that ties are broken in favor of the agent with smaller index $i$. Both winners will pay the largest losing bid for $80$ units of traffic. Additionally, the winner of the first ad slot will pay the bid of the winner of the second ad slot for the extra $20$ units he receives. This because the pivotal bidder for the last $20$ units is the winner of the second ad slot, not the loser with the largest bid. 

Suppose first that all agents report truthfully. Then, agent $1$ wins the first ad slot, and agent $2$ wins the second. Agent $2$ pays an estimated $80v_{3}$, while agent $1$ pays $80v_3 + 20v_{2} $.
Now, imagine agent $k$ was attempting to learn how to play by bidding according to a reinforcer. The designer would want to provide feedback to the agent, to ensure that he be able to compute what would have happened, had he bid an amount $\hat{v}_k \neq v_k$, keeping everyone else's reports fixed. The feedback required is simple: agent $k$ needs a price for the second ad slot, and a price for the first. 
In this example, the designer would communicate prices $v_{2}$ and $v_1$.

To see why these prices are sufficient, consider agent $k$'s calculations. There are only three possibilities. If he bid $\hat{v}_k \leq v_{2}$, then he would receive zero payoff, the same as if he was to bid truthfully. Suppose he bid $v_2 < \hat{v}_k \leq v_1$ instead: then agent $k$ would win the second ad slot, and pay $80v_2$. Finally, if agent $k$ bid $v_1 < \hat{v}_k$, he would win the first ad slot. His payment would then be $80v_2 + 20v_1$. All three counterfactual payoffs only require two prices: the bids of the two winners. 

Similarly, the winner of the second slot requires two prices: $v_1$ and $v_3$. The winner of the first slot instead requires $v_2$ and $v_3$. In a VCG setting communication reduces to revealing the values of the bidders that are pivotal for the specific agent. This is indeed a menu description, and it is much more private than the full-feedback policy, which would require communicating all reports $v_{-i}$ to every agent $i$ in the auction.

\section{Conclusion}\label{sec:conclusion}
This paper analyzes collusion in games played by online learning algorithms. We take a theoretical perspective and, complementing burgeoning empirical and numerical evidence, we identify a new driver of collusive behavior specific to algorithmic players. We first address the issue of analytical intractability of strategic interaction among algorithms by showing that it can be approximated with a dynamical system. Then we apply this framework to dominant-strategy games, and we show that ($\varepsilon$-)greedy algorithms can learn to collude. We identify the mechanism sustaining collusion, a statistical linkage we call \emph{spontaneous coupling}: when algorithms are slow to realize the value of the competitive action, joint collusion appears more attractive. Involuntary coupling yields self-fulfilling biases in the estimates: we demonstrate this intuition in a Prisoner's Dilemma with Q-learning agents.
We expect the techniques developed to analyze the simple Prisoner's Dilemma to yield insight in games with more complex strategic structure. In particular, we believe similar techniques can help understanding how AI algorithms reach tit-for-tat strategies when given monitoring technology.

We show that spontaneous coupling may sustain various forms of market manipulation, from price fixing to market splitting. We expect the same arguments to apply to other forms of anticompetitive conduct as well. In particular, we wish to highlight the limited role that monitoring plays in a conduct sustained by spontaneous coupling. Regulation will need to adapt its current policies to account for this truly tacit collusion.
We then design mechanism that are robust to the presence of algorithmic agents. The ideas we outline are sensible design principles even when dealing with algorithms that are not reinforcers. For example, regret-minimizing algorithms obtain better guarantees when provided with counterfactual information.

We view our paper as contributing to the growing literature studying strategic interaction of algorithmic agents. Algorithms shape the dimensions of rationality of these decision makers, and allow us to carry out a disciplined analysis of equilibria and market design for such boundedly-rational agents. There are many other dimensions of interest in the study of strategic algorithmic interactions that we do not touch upon in this paper. For example, we focused on dominant strategy games, which intrinsically make collusion the hardest to sustain: the outcomes of games where the separation between competition and collusion is less stark remains unclear, and worthwhile to pursue. Our algorithms interact with the environment and adapt according to the feedback they receive, but many deployed market algorithms are instead trained offline. Offline training is often prone to unwanted feedback loops, but as our analysis points to correlation as a key driver of collusion, we suggest that offline algorithms may be less prone to collusive behavior.
Another interesting aspect of algorithmic collusion is whether coordination on collusive outcomes would be even easier if algorithms were able to retain memory of recent payoff-relevant quantities. We suspect that when algorithms are enabled to learn dynamic reward-punishment strategies their collusive behavior will increase substantially, as highlighted in the literature. 

An interesting question is what could a sophisticated player achieve when competing against automated decision-makers. The manipulability of these algorithms deserves further analysis, and we believe a setting similar to the one offered in this work could prove helpful in understanding these questions. 
Finally, algorithmic decision-making can be seen as a form of bounded rationality. This implies that the set of implementable outcomes is, in general, wider than that of rational agents. \Cref{thm:IESDS} suggests that arguments similar to \citet{Abreu1992} could prove useful in enlarging the set of implementable outcomes; characterizing the set of implementable outcomes for purely adaptive decision makers is beyond the scope of this paper, but of independent interest.

\bibliographystyle{ecta} 
\bibliography{RL_bibliography}

\begin{thebibliography}{45}
\newcommand{\enquote}[1]{``#1''}
\expandafter\ifx\csname natexlab\endcsname\relax\def\natexlab#1{#1}\fi

\bibitem[\protect\citeauthoryear{Abada and Lambin}{Abada and
  Lambin}{2023}]{abada2023}
\textsc{Abada, I. and X.~Lambin} (2023): \enquote{Artificial intelligence: Can
  seemingly collusive outcomes be avoided?} \emph{Management Science}.

\bibitem[\protect\citeauthoryear{Abreu and Matsushima}{Abreu and
  Matsushima}{1992}]{Abreu1992}
\textsc{Abreu, D. and H.~Matsushima} (1992): \enquote{Virtual implementation in
  iteratively undominated strategies: complete information,}
  \emph{Econometrica}, 993--1008.

\bibitem[\protect\citeauthoryear{Asker, Fershtman, and Pakes}{Asker
  et~al.}{2022}]{Asker2022}
\textsc{Asker, J., C.~Fershtman, and A.~Pakes} (2022): \enquote{{Artificial
  Intelligence, Algorithm Design and Pricing},} \emph{American Economic Review,
  P\&P}, forthcoming.

\bibitem[\protect\citeauthoryear{Assad, Clark, Ershov, and Xu}{Assad
  et~al.}{2023}]{Assad2021}
\textsc{Assad, S., R.~Clark, D.~Ershov, and L.~Xu} (2023): \enquote{Algorithmic
  Pricing and Competition: Empirical Evidence from the German Retail Gasoline
  Market,} \emph{Journal of Political Economy}, forthcoming.

\bibitem[\protect\citeauthoryear{Banchio and Skrzypacz}{Banchio and
  Skrzypacz}{2022}]{Banchio2021b}
\textsc{Banchio, M. and A.~Skrzypacz} (2022): \enquote{Artificial Intelligence
  and Auction Design,} Working paper.

\bibitem[\protect\citeauthoryear{Benaim}{Benaim}{1996}]{Benaim1996}
\textsc{Benaim, M.} (1996): \enquote{A Dynamical System Approach to Stochastic
  Approximations,} \emph{SIAM Journal of Control and Optimization}, 34,
  437--472.

\bibitem[\protect\citeauthoryear{Bendor, Mookherjee, and Ray}{Bendor
  et~al.}{2001{\natexlab{a}}}]{Bendor2001a}
\textsc{Bendor, J., D.~Mookherjee, and D.~Ray} (2001{\natexlab{a}}):
  \enquote{Aspiration-Based Reinforcement Learning in Repeated Interaction
  Games: An Overview,} \emph{International Game Theory Review}, 03, 159--174.

\bibitem[\protect\citeauthoryear{Bendor, Mookherjee, and Ray}{Bendor
  et~al.}{2001{\natexlab{b}}}]{Bendor2001b}
---\hspace{-.1pt}---\hspace{-.1pt}--- (2001{\natexlab{b}}):
  \enquote{Reinforcement Learning in Repeated Interaction Games,} \emph{The
  B.E. Journal of Theoretical Economics}, 1.

\bibitem[\protect\citeauthoryear{Bhandari, Russo, and Singal}{Bhandari
  et~al.}{2021}]{Russo2021}
\textsc{Bhandari, J., D.~Russo, and R.~Singal} (2021): \enquote{A finite time
  analysis of temporal difference learning with linear function approximation,}
  \emph{Operations Research}, 69, 950--973.

\bibitem[\protect\citeauthoryear{Borkar and Meyn}{Borkar and
  Meyn}{2000}]{Borkar2000}
\textsc{Borkar, V.~S. and S.~P. Meyn} (2000): \enquote{The O.D.E. Method for
  Convergence of Stochastic Approximation and Reinforcement Learning,}
  \emph{SIAM J. Control. Optim.}, 38, 447--469.

\bibitem[\protect\citeauthoryear{Brown and MacKay}{Brown and
  MacKay}{2021}]{Brown2021}
\textsc{Brown, Z.~Y. and A.~MacKay} (2021): \enquote{Competition in Pricing
  Algorithms,} \emph{American Economic Journal: Microeconomics}, forthcoming.

\bibitem[\protect\citeauthoryear{Börgers and Sarin}{Börgers and
  Sarin}{1997}]{Borgers1997}
\textsc{Börgers, T. and R.~Sarin} (1997): \enquote{Learning Through
  Reinforcement and Replicator Dynamics,} \emph{Journal of Economic Theory},
  77, 1--14.

\bibitem[\protect\citeauthoryear{Calvano, Calzolari, Denicolo, and
  Pastorello}{Calvano et~al.}{2020}]{Calvano2020}
\textsc{Calvano, E., G.~Calzolari, V.~Denicolo, and S.~Pastorello} (2020):
  \enquote{Artificial intelligence, algorithmic pricing, and collusion,}
  \emph{American Economic Review}, 110, 3267--97.

\bibitem[\protect\citeauthoryear{{Competition \& Markets
  Authority}}{{Competition \& Markets Authority}}{2021}]{CMA2021}
\textsc{{Competition \& Markets Authority}} (2021): \enquote{Algorithms: How
  they can reduce competition and harm consumers,} Discussion paper.

\bibitem[\protect\citeauthoryear{{Competition Bureau}}{{Competition
  Bureau}}{2018}]{Competition2018}
\textsc{{Competition Bureau}} (2018): \enquote{Big Data and Innovation:
  Implications for Competition Policy in Canada,} Discussion paper.

\bibitem[\protect\citeauthoryear{di~Bernardo, Budd, Champneys, and
  Kowalczyk}{di~Bernardo et~al.}{2008}]{Dibernardo2008}
\textsc{di~Bernardo, M., C.~Budd, A.~R. Champneys, and P.~Kowalczyk} (2008):
  \emph{Piecewise-smooth dynamical systems: theory and applications}, vol. 163,
  Springer Science \& Business Media.

\bibitem[\protect\citeauthoryear{Dieci and Lopez}{Dieci and
  Lopez}{2011}]{dieci2011}
\textsc{Dieci, L. and L.~Lopez} (2011): \enquote{Sliding motion on
  discontinuity surfaces of high co-dimension. A construction for selecting a
  Filippov vector field,} \emph{Numerische Mathematik}, 117, 779--811.

\bibitem[\protect\citeauthoryear{Erev, Bereby-Meyer, and Roth}{Erev
  et~al.}{1999}]{erev1999}
\textsc{Erev, I., Y.~Bereby-Meyer, and A.~E. Roth} (1999): \enquote{The effect
  of adding a constant to all payoffs: experimental investigation, and
  implications for reinforcement learning models,} \emph{Journal of Economic
  Behavior \& Organization}, 39, 111--128.

\bibitem[\protect\citeauthoryear{Erev and Roth}{Erev and Roth}{1998}]{Erev1998}
\textsc{Erev, I. and A.~E. Roth} (1998): \enquote{Predicting How People Play
  Games: Reinforcement Learning in Experimental Games with Unique, Mixed
  Strategy Equilibria,} \emph{The American Economic Review}, 88, 848--881.

\bibitem[\protect\citeauthoryear{Filippov}{Filippov}{1988}]{Filippov1988}
\textsc{Filippov, A.~F.} (1988): \emph{Differential Equations with
  Discontinuous Right-Hand Sides}, Springer Science \& Business Media.

\bibitem[\protect\citeauthoryear{Gomes and Kowalczyk}{Gomes and
  Kowalczyk}{2009}]{Gomes2009}
\textsc{Gomes, E.~R. and R.~Kowalczyk} (2009): \enquote{Dynamic Analysis of
  Multiagent Q-Learning with $\epsilon$-Greedy Exploration,} in
  \emph{Proceedings of the 26th Annual International Conference on Machine
  Learning}, 369–376.

\bibitem[\protect\citeauthoryear{Gonczarowski, Heffetz, and
  Thomas}{Gonczarowski et~al.}{2023}]{Gonczarowski2023}
\textsc{Gonczarowski, Y., O.~Heffetz, and C.~Thomas} (2023):
  \enquote{Strategyproofness-Exposing Mechanism Descriptions,} Working paper.

\bibitem[\protect\citeauthoryear{Hammond}{Hammond}{1979}]{Hammond1979}
\textsc{Hammond, P.~J.} (1979): \enquote{Straightforward Individual Incentive
  Compatibility in Large Economies,} \emph{The Review of Economic Studies}, 46,
  263--282.

\bibitem[\protect\citeauthoryear{Hansen, Misra, and Pai}{Hansen
  et~al.}{2021}]{Hansen2021}
\textsc{Hansen, K., K.~Misra, and M.~Pai} (2021): \enquote{Algorithmic
  Collusion: Supra-Competitive Prices via Independent Algorithms,}
  \emph{Marketing Science}, 40, 1--12.

\bibitem[\protect\citeauthoryear{Harrington}{Harrington}{2018}]{Harrington2019}
\textsc{Harrington, J.~E.} (2018): \enquote{{Developing Competition Law for
  Collusion by Autonomous Artificial Agents},} \emph{Journal of Competition Law
  \& Economics}, 14, 331--363.

\bibitem[\protect\citeauthoryear{Harrington}{Harrington}{2022}]{harrington2022}
---\hspace{-.1pt}---\hspace{-.1pt}--- (2022): \enquote{The effect of
  outsourcing pricing algorithms on market competition,} \emph{Management
  Science}, 68, 6889--6906.

\bibitem[\protect\citeauthoryear{Holden}{Holden}{2014}]{holden2014chaos}
\textsc{Holden, A.~V.} (2014): \emph{Chaos}, vol. 461, Princeton University
  Press.

\bibitem[\protect\citeauthoryear{Hopkins and Posch}{Hopkins and
  Posch}{2005}]{Hopkins2005}
\textsc{Hopkins, E. and M.~Posch} (2005): \enquote{Attainability of boundary
  points under reinforcement learning,} \emph{Games and Economic Behavior}, 53,
  110--125.

\bibitem[\protect\citeauthoryear{Johnson, Rhodes, and Wildenbeest}{Johnson
  et~al.}{2023}]{Johnson2022}
\textsc{Johnson, J., A.~Rhodes, and M.~Wildenbeest} (2023): \enquote{Platform
  Design when Sellers Use Pricing Algorithms,} \emph{Econometrica},
  forthcoming.

\bibitem[\protect\citeauthoryear{Karandikar, Mookherjee, Ray, and
  Vega-Redondo}{Karandikar et~al.}{1998}]{Karandikar1998}
\textsc{Karandikar, R., D.~Mookherjee, D.~Ray, and F.~Vega-Redondo} (1998):
  \enquote{Evolving Aspirations and Cooperation,} \emph{Journal of Economic
  Theory}, 80, 292--331.

\bibitem[\protect\citeauthoryear{Klein}{Klein}{2021}]{Klein2021}
\textsc{Klein, T.} (2021): \enquote{Autonomous algorithmic collusion:
  Q-learning under sequential pricing,} \emph{The RAND Journal of Economics},
  52, 538--558.

\bibitem[\protect\citeauthoryear{Kurtz}{Kurtz}{1970}]{Kurtz1970}
\textsc{Kurtz, T.~G.} (1970): \enquote{Solutions of ordinary differential
  equations as limits of pure jump Markov processes,} \emph{Journal of Applied
  Probability}, 7, 49--58.

\bibitem[\protect\citeauthoryear{Lamba and Zhuk}{Lamba and
  Zhuk}{2022}]{Lamba2022}
\textsc{Lamba, R. and S.~Zhuk} (2022): \enquote{Pricing with Algorithms,}
  working paper.

\bibitem[\protect\citeauthoryear{Leisten}{Leisten}{2022}]{Leisten2022}
\textsc{Leisten, M.} (2022): \enquote{Algorithmic Competition, with Humans,}
  Working paper.

\bibitem[\protect\citeauthoryear{Leonardos and Piliouras}{Leonardos and
  Piliouras}{2022}]{Stefanos2022}
\textsc{Leonardos, S. and G.~Piliouras} (2022):
  \enquote{Exploration-exploitation in multi-agent learning: Catastrophe theory
  meets game theory,} \emph{Artificial Intelligence}, 304.

\bibitem[\protect\citeauthoryear{Lerer and Peysakhovich}{Lerer and
  Peysakhovich}{2017}]{lerer2017}
\textsc{Lerer, A. and A.~Peysakhovich} (2017): \enquote{Maintaining cooperation
  in complex social dilemmas using deep reinforcement learning,} \emph{arXiv
  preprint arXiv:1707.01068}.

\bibitem[\protect\citeauthoryear{Mertikopoulos and Sandholm}{Mertikopoulos and
  Sandholm}{2016}]{mertikopoulos2016}
\textsc{Mertikopoulos, P. and W.~H. Sandholm} (2016): \enquote{Learning in
  games via reinforcement and regularization,} \emph{Mathematics of Operations
  Research}, 41, 1297--1324.

\bibitem[\protect\citeauthoryear{Musolff}{Musolff}{2021}]{Musolff2021}
\textsc{Musolff, L.~A.} (2021): \enquote{Algorithmic Pricing Facilitates Tacit
  Collusion: Evidence from E-Commerce,} Working paper.

\bibitem[\protect\citeauthoryear{OECD}{OECD}{2017}]{OECD2017}
\textsc{OECD} (2017): \enquote{Algorithms and Collusion: Competition Policy in
  the Digital Age,} Technical report.

\bibitem[\protect\citeauthoryear{Parkes}{Parkes}{2004}]{Parkes2004}
\textsc{Parkes, D.~C.} (2004): \enquote{On learnable mechanism design,} in
  \emph{Collectives and the Design of Complex Systems}, Springer, 107--131.

\bibitem[\protect\citeauthoryear{Possnig}{Possnig}{2023}]{Possnig2023}
\textsc{Possnig, C.} (2023): \enquote{Reinforcement Learning and Collusion,}
  Working paper.

\bibitem[\protect\citeauthoryear{Tumer and Khani}{Tumer and
  Khani}{2009}]{Tumer2009}
\textsc{Tumer, K. and N.~Khani} (2009): \enquote{Learning from actions not
  taken in multiagent systems,} \emph{Advances in Complex Systems}, 12,
  455--473.

\bibitem[\protect\citeauthoryear{Tuyls, Hoen, and Vanschoenwinkel}{Tuyls
  et~al.}{2005}]{Tuyls2005}
\textsc{Tuyls, K., P.~J. Hoen, and B.~Vanschoenwinkel} (2005): \enquote{An
  Evolutionary Dynamical Analysis of Multi-Agent Learning in Iterated Games,}
  \emph{Autonomous Agents and Multi-Agent Systems}, 12, 115--153.

\bibitem[\protect\citeauthoryear{Watkins and Dayan}{Watkins and
  Dayan}{1992}]{Watkins1992}
\textsc{Watkins, C.~J. and P.~Dayan} (1992): \enquote{Q-learning,}
  \emph{Machine learning}, 8, 279--292.

\bibitem[\protect\citeauthoryear{Wunder, Littman, and Babes}{Wunder
  et~al.}{2010}]{Wunder2010}
\textsc{Wunder, M., M.~L. Littman, and M.~Babes} (2010): \enquote{Classes of
  Multiagent Q-learning Dynamics with epsilon-greedy Exploration,} in
  \emph{Proceedings of the 27th Annual International Conference on Machine
  Learning}, 1167--1174.

\end{thebibliography}

\newpage
\begin{appendices}
\crefalias{subsection}{appendix}
\crefalias{section}{appendix}
\section{}\label{app:proofs}

We restate \Cref{thm: fluid approx thm} in its more formal version.

\begin{theorem*}[1]
Let $(\theta, \pi)$ be a collection of reinforcers that satisfy \Cref{manualasm:lipschitz} individually, and let the domain $\mathscr{T}\subset \mathbb{R}^{\sum_i d_i}$ of $\theta$ be a compact set. Let $(H_j)_{j \in J}$ be the collection of $\theta$'s maximal Lipschitz-continuity domains, that is, let $H_j$ be the largest open set such that $\theta$ is Lipschitz over $H_j$ and there is a discontinuity on $\partial H_j$.
For all $j \in J$
the collection of Cauchy problems 
\[\begin{cases}
\frac{d\Theta^i(t)}{dt} = \alpha \mathbb{E}_{\pi^i,\pi^{-i}}\left[D^i(\pi^i,r(\pi^i,\pi^{-i}),\Theta^i(t)) \right] \\ 
\Theta^i(0) = y^i_0
\end{cases}\]
has a solution $\Theta^i$ for all $i$ over $H_j$ for all $y_0 \in H_j$. There exists a sequence of processes $\{\theta^n\}_{n \in \mathbb{N}}$ such that: 
\begin{itemize}
    \item $\mathbb{E}\left[\theta^1(\tau(k))\right]  = \mathbb{E}\left[\theta(k)\right]$ for all $k$, $\tau(k) = \inf \left\{t \ | \ \theta^1(t)\text{ jumped } k \text{ times}\right\}$,
    \item the infinitesimal generators $\mathcal{A}D_n(\theta)$ are all identical to $\mathcal{A}D_1(\theta)$ for all $\theta \in H_j$ and $n$, 
    \item $\lim_{n\to\infty} P\Big\{\sup_{t\leq T} \Big\lVert \theta^n(t) - \Theta(t)\Big\rVert > \eta \Big\}=0$ for all $T\geq 0$ and $\eta>0$ such that $\{\Theta(t)\}_{t\leq T} \subset H_j$.
\end{itemize}
\end{theorem*}

\paragraph{Proof of \Cref{thm: fluid approx thm}.}
The existence of a solution for the Cauchy problem is guaranteed by \Cref{manualasm:lipschitz} and the restriction to the maximal continuity domains $H_j$.
In particular, notice that one can write $D^i(\pi^i(\theta),r(\pi^i(\theta^i), \pi^{-i}(\theta^{-i}),\theta^i)$ as a map $\hat{D}(\theta,Y)$ where $Y$ is a random variable representing the uncertainty introduced by the policies $\pi$.

We can divide the proof of \Cref{thm: fluid approx thm} in two main steps:
\begin{enumerate}
    \item Finding the correct continuous-time embedding of the reinforcer $\theta$;
    \item Identifying a scaling that guarantees limits are well-defined.
\end{enumerate}

First, let us fix a compact ball of radius $r$ in $\mathbb{R}^{\sum_i d_i}$. We will consider the set $E = H \cap B(r)$  with the Borel intersection sigma algebra. Since we can choose $r$ to be as large as we want, the approximation will hold for any finite values of $\theta$. Let us add one component to the vector $\theta$, in position $\sum_i d_i +1$, which will keep track of the iteration $k$.

As far as the first step is concerned, let us define a Poisson process $N_1$ of rate $\lambda_1 = 1$. Consider the sequence of (stochastic) arrival times ${0< \tau_1 < \tau_2 < \tau_3< \dots}$. We define the process $\theta^1(t)$ as
\[\theta^1(t) = \theta(k) \qquad \text{if }\tau_{k+1} > t \geq \tau_k\]
for all times $t \geq 0$. The process $\theta^1(t)$ is a compound Poisson process, cadlag and Markov. Naturally, its $\sum_i d_i + 1$ component always coincides with the iteration $k$.
At arrival times the algorithm is equal to its continuous time equivalent $\theta^1$, which proves item 1 of the Theorem.

We now want to increase the pace of the updates while retaining the same uncertainty in expectation. Intuitively, we can ``speed up'' the Poisson arrivals, but we also need to ``dampen'' the jumps accordingly, otherwise the process will diverge to infinity. 
Formally, we consider a sequence of Continuous-Time Markov Chains indexed by $n \in \mathbb{N}$ as follows:
\begin{itemize}
    \item The jump rate $\lambda_n$ is defined as 
    $\lambda_n = n$.
    
    \item At each jump, the update in the first $\sum_i d_i$ components is\footnote{Note that we omit the dependence on the iteration since iterations are now part of the process $\theta^n$.} 
    \[\theta^n(t) - \theta^n(t^-) = \frac{1}{n} \hat{D}(\theta^n(t^-),Y)\]
    
    \item At each jump, the update in the coordinate $\sum_i d_i +1$ is 
    \[\theta^n(t) - \theta^n(t^-) =\frac{1}{n}\]
\end{itemize}

Intuitively, the updates of the original process $\theta$ are scaled down by a factor $n$ and the last coordinate keeps track of how many updates have occurred scaled by $n$.

Consider then the probability measure $\mu^n(x,dz)$ of the size of the updates starting at $x$, with
\[\mu_n(x,dz) = \mathbb{P}\Big\{\theta^n(\tau^n) \in dz | \theta^n(0) = x \Big\}\]
where $\tau^n$ is the first exit time of $\theta^n$ from $x$. We define the component-wise function
\begin{equation}\label{eq:F}
F_n(x)^{m} = \lambda_n \int (z^{m}-x^{m}) \mu_n(x,dz),\end{equation}
which intuitively describes the expected jump of $\theta^n$ from $x$  along the $m$-th component over one unit of time.
In fact, $F_n$ can be rewritten as
\[F_n(x)^m = n\int \frac{\alpha}{n} \hat{D}^m(x,Y) \mu = \int \alpha \hat{D}^m(x,Y) \mu.\]
We chose the scaling in such a way that the function $F_n$ is independent of $n$. 
The function $F_n$ is exactly the infinitesimal generator of the compound Poisson process $\theta^n(t)$, therefore item (2) of the Theorem is proved. 

Let $F(x) \coloneqq \lim_{n\to\infty} F_n(x)$. It is clear that $F(x) = F_n(x)$ for all $n$. Moreover the function $F(x)$ is Lipschitz in every component. 
We will need a technical lemma:

\begin{lemma}\label{lemma:technical}
Let $E$ be a compact set in $\mathbb{R}^{|\mathcal{I}|\times|\mathcal{S}|\times|\mathcal{A}|}$.  There exists a sequence $\{\varepsilon_n\}_n>0$ with $\lim_{n \to \infty} \varepsilon_n = 0$ such that 
\[\lim_{n \to \infty} \sup_{x \in E} \lambda_n \int_{|z-x|>\varepsilon_n} |z-x|\mu_n(x,dz) = 0\]
Moreover, 
\[\sup_n \sup_{x \in E} \lambda_n \int_{E} |z-x|\mu_n(x,dz) < \infty\]
\end{lemma}
\begin{proof}
Since $E$ is compact $\hat{D}(\theta,Y)$ must be bounded. Let $M = \sup_{\theta\in E} \hat{D}(\theta,Y)$ across dimension, and note that $M < +\infty$.
Let then $\{\frac{M}{n}\}_n$ be a sequence satisfying the assumptions of the Lemma, and notice that $\int_{|z-x|>\varepsilon_n} |z-x|\mu_n(x,dz) = 0$ for all $x$ and $n$. We thus proved the first claim. The second claim follows a simple observation: $ |z-x|\mu_n(x,dz) \leq \frac{M}{n}$ for all $x$. Since $\lambda_n = n$, we obtain that
\[\sup_n \sup_{x \in E} \lambda_n \int_{E} |z-x|\mu_n(x,dz) =M < \infty\] which concludes the proof.
\end{proof}

This lemma verifies one of the conditions of the following Theorem taken from \citet{Kurtz1970}:
\begin{customthm}{2.11}\label{customthm:Kurtz}
Suppose there exists $E\subset \mathbb{R}^k$, a function $F\colon E\to\mathbb{R}^k$ and a constant $M$ such that $|F(x) - F(y)|\leq M|x-y|$ for all $x,y \in E$ and \[\lim_{n\to\infty}\sup_{x\in E_n\cup E} |F_n(x) - F(x)| = 0\]
Let $X(t,x_0), 0\leq t\leq T, x_0 \in E$ satisfy
\[X(0,x_0) = x_0, \qquad \dot{X}(t,x_0) = F(X(t,x_0)) \]
Suppose additionally that the sequence $F_n$ satisfies the conditions of \Cref{lemma:technical}, then for every $\eta>0$
\[\lim_{n \to\infty} P\big\{\sup_{t\leq T}|X_n(t) - X(t,x_0)|\geq \eta\big\}=0\]
\end{customthm}
If $X(t,x_0) =\Theta$, we can verify that the assumptions of the Theorem hold:
\begin{itemize}
    \item since $\hat{D}$ is Lipschitz, and $F_n=F$ are integrals of Lipschitz functions, it is clear that $|F(x) - F(y)|\leq M|x-y| $ holds,
    \item $\lim_{n}\sup_{x\in H} |F_n(x,t) - F(x,t)| = 0$ is satisfied by definition of $F_n = F$,
    \item the conditions of $\Cref{lemma:technical}$ are verified.
\end{itemize}
Then, Theorem 2.11 implies that for every $\eta>0$
\[\lim_{n \to\infty} P\big\{\sup_{t\leq T}|\theta^n(t) - \Theta(t)|\geq \eta\big\}=0\]
which proves item 3 of the Theorem and concludes the proof.

As advanced at the beginning, the  process $\Theta$ is a deterministic process with all uncertainty collapsed into the drift component. \qed 

\paragraph{Proof of \Cref{prop:Inclusion}.} This result relies on the following theorem from \citet{Filippov1988}:
\begin{customthm}{1, Chapter 2, Section 7} Let $S$ be a compact domain. Let $G(t,x)$ be a nonempty, bounded, closed, convex set-valued function that is upper semicontinuous in $t,x$ for all $(t,x) \in S$. Then for any point $(t_0,x_0)\in S$ there exists a solution of the problem \[\dot{x} \in G(t,x), \qquad x(t_0) = x_0\]
and if the domain $S$ contains a cylinder $Z(t_0\leq t \leq t_0+a, \ |x-x_0|\leq b)$, the solution exists at least on the interval
\[t_0\leq t \leq t_0+d, \qquad d= \min\Big\{a;\frac{b}{m}\Big\} \qquad m = \sup_Z |G(t,x)|\]
\end{customthm}
Let $G(t,x)$ be defined as the vector field given in the statement of \Cref{prop:Inclusion}, and notice that $G$ is time-invariant; also let $S$ be a compact ball in $\mathbb{R}\times\mathbb{R}^{|\mathcal{I}|\times|\mathcal{A}|}$. We show that $G$ satisfies all the conditions:
\begin{itemize}
    \item it is nonempty over $\mathbb{R} \times \mathbb{R}^{|\mathcal{I}|\times|\mathcal{A}|}$,
    \item it is bounded everywhere, since each individual $F_a$ is bounded over $\omega_{a,a} \cap S$,
    \item it is closed and convex, as it is clear from the definition above,
    \item it is upper semicontinuous: to see this, select a convergent sequence in the domain. If the sequence is entirely contained in a $\omega_j$, then continuity is clear. Instead, suppose that the sequence lies in an $\omega_{a,a}$ but its limit lies on $\partial \overline{\omega}_{a,a}$. Upper semicontinuity is guaranteed by the definition of $G$ as a convex combination of $F$ over the overlapping boundaries. 
\end{itemize}
Additionally, note that the domain is any compact ball, therefore we can find a $S$ that contains any cylinder $Z$  and a solution to this differential inclusion is global within any compact subset of $\mathbb{R}^{|\mathcal{I}|\times|\mathcal{A}|}$.
\qed

\paragraph{Proof of \Cref{prop:two eq PD}}
The existence of the equilibrium $q^{eq}_{D}$ follows directly from setting the field over $\omega_{D,D}$ to zero. 

We prove existence of $q^{eq}_C$ and its related property for one agent; by symmetry, the other agent's Q-values enjoy the same properties. The boundary is defined as $\Sigma = \{q \in \mathbb{R}^2 \colon c\cdot q = 0\}$ where $c = (1,-1)$ and $\cdot$ denotes the usual dot product. Using the Filippov convention, we can further divide $\Sigma$ in three regions:
\begin{itemize}
    \item a \emph{crossing region}, $\Sigma^c =\{q \colon (c\cdot (A_Cq+b_C))(c\cdot (A_{D}q+b_{D}))>0\}$
    \item a \emph{repulsive region}, $\Sigma^r =\{q \colon c\cdot (A_Cq+b_C) >0, \  c\cdot (A_Dq+b_{D})<0\}$
    \item a \emph{sliding region}, $\Sigma^s =\{q \colon c\cdot (A_Cq+b_C) <0, \ c\cdot (A_{D}q+b_{D})>0\}$
\end{itemize}
where we have
\begin{align*}
    A_C=\begin{bmatrix}
    \alpha \left(1-\frac{\varepsilon}{2}\right) (\gamma-1) & 0\\
    \alpha \gamma \frac{\varepsilon}{2} & -\alpha \frac{\varepsilon}{2}
    \end{bmatrix} \qquad
    A_D=\begin{bmatrix}
    -\alpha \frac{\varepsilon}{2} & \alpha \gamma \frac{\varepsilon}{2}\\
    0 & \alpha \left(1-\frac{\varepsilon}{2}\right) (\gamma-1)
    \end{bmatrix},
\end{align*}

and 

\begin{align*}
    b_C=\begin{bmatrix}
    \alpha \left(1-\frac{\varepsilon}{2}\right)\left(2-\frac{\varepsilon}{2}\right) g\\
    \alpha \frac{\varepsilon}{2} \left(2 + g - g \frac{\varepsilon}{2}\right)
    \end{bmatrix} \qquad b_D=\begin{bmatrix}
    \alpha \left(1+\frac{\varepsilon}{2}\right)\frac{\varepsilon}{2} g\\
    \alpha \left(1-\frac{\varepsilon}{2}\right)\left(2+\frac{\varepsilon}{2}g\right)
    \end{bmatrix}.
\end{align*}
We can define the sliding solution as the field $\frac{d\mathbf{Q}}{dt} = F^s(\mathbf{Q})$ over the sliding region where 
\[F^s(\mathbf{Q}) = \frac{c\cdot(A_{D}\mathbf{Q}+b_{D}) (A_C\mathbf{Q}+b_C) - c\cdot(A_C\mathbf{Q}+b_C)(A_{D}\mathbf{Q}+b_{D}) }{c\cdot (A_{D}\mathbf{Q}+b_{D}) - c\cdot (A_C\mathbf{Q}+b_C)}\]
The relative time spent on $\omega_{C,C}$ at point $\mathbf{Q}$ is defined as \[\tau_C = \frac{c\cdot(A_{D}\mathbf{Q}+b_{D})}{c\cdot (A_{D}\mathbf{Q}+b_{D}) - c\cdot (A_C\mathbf{Q}+b_C)}\]

The sliding vector field becomes 
\[\frac{d\mathbf{Q}_j}{dt} = \frac{\alpha\big(\frac{1}{2}\varepsilon g (2-\varepsilon)(g-1) + (2g+(\gamma-1)\mathbf{Q}_j) (2+(\gamma-1)\mathbf{Q}_j)\big)}{2(1+g+(\gamma-1)\mathbf{Q}_j)}\]
for every direction $j$. By setting the field equal to zero and solving for $\mathbf{Q}_j$, we find that there is an equilibrium at 
\[q^{eq}_{C,j} = \frac{1+g + \sqrt{(g-1)(g-1-\varepsilon g+ \frac{\varepsilon^2g}{2} )}}{(\gamma-1)}\]
for all j. This equation has a feasible solution for all $\varepsilon < 1- \sqrt{\frac{2-g}{g}}$. \qed

\paragraph{Proof of \Cref{cor:Local time}.}
The result follows immediately from the proof of \Cref{prop:two eq PD}. In particular, it is sufficient to compute \[\tau_C = \frac{c\cdot(A_{D}\mathbf{Q}+b_{D})}{c\cdot (A_{D}\mathbf{Q}+b_{D}) - c\cdot (A_C\mathbf{Q}+b_C)}\]
at $\mathbf{Q}=q^{eq}_C$.
\qed

\paragraph{Proof of \Cref{thm:RLS}.}

First off, notice that payoffs in a Prisoner's Dilemma are ordered as follows:
\[r(D,C) > r(C,C) > r(D,D) > r(C,D)\]
We will prove existence of a stationary point on the switching surface for a system defined as follows: 
\[\begin{cases} \dot{\theta}^C = \alpha \left[U(\theta^C,r(C,C)) + V(\theta)\right] \\ \dot{\theta}^{D} = (1-\alpha) \left[U(\theta^{D},r(D,C)) + V(\theta)\right]
\end{cases}\]
in the half-plane where $C$ is the preferred action, and 
\[\begin{cases} \dot{\theta}^C = (1-\alpha) \left[U(\theta^C,r(C,D))+ V(\theta)\right] \\ \dot{\theta}^{D} = \alpha\left[U(\theta^{D},r(D,D))+ V(\theta)\right]
\end{cases}\]
in the half-plane where $D$ is the preferred action. Note that we are assuming WLOG that the learning speeds sum to $1$, since all it matters is the relative speed of learning. To start, let us assume $\alpha$ is identical across the half planes.

We want to show that there exist an $\alpha$ and a $\theta^*$ such that:
\begin{equation}\label{eq:nolocaltime}
\begin{cases}
\alpha U(\theta^*,r(C,C)) + (1-\alpha) U(\theta^*,r(C,D)) + V(\theta^*)= 0\\
(1-\alpha) U(\theta^*,r(D,C)) + \alpha U(\theta^*,r(D,D)) + V(\theta^*)= 0
\end{cases}\end{equation}
Because we used the same $\alpha$ on both sides, we are simply looking for a translation of \Cref{eq:nolocaltime}, therefore we can develop the argument disregarding the $V(\theta^*)$ component.
For a given $\theta$, we can write this problem as a convex combination of two vectors:
\[\alpha \Vec{x} + (1-\alpha) \Vec{y} = \Vec{0}\] 
where 
\begin{align*}
    \Vec{x}(\theta) = \begin{bmatrix}
           U(\theta,r(C,C)) \\
           U(\theta,r(D,D)) \\
         \end{bmatrix} \qquad \qquad
    &\Vec{y}(\theta) = \begin{bmatrix}
           U(\theta,r(C,D)) \\
           U(\theta,r(D,C)) \\
         \end{bmatrix}
\end{align*}
Let $\theta^1$ be the value of $\theta$ such that $U(\theta^1,r(C,C))=0$. Then, using the monotonicity of $U$ in rewards, the two vectors $\Vec{x},\Vec{y}$ computed in $\theta^1$ will be positioned as in \Cref{fig:homotopy}. Let $\theta^{2}$ instead be the value of $\theta$ such that $U(\theta^2,r(C,D))=0$. Again, the two vectors $\Vec{x},\Vec{y}$ are positioned as in \Cref{fig:homotopy}. Notice that by monotonicity of $U$ in its first component, $\theta^1>\theta^2$. The same assumption guarantees that the lines $\vec{y}(\theta^1 +(\theta^2-\theta^1)t)$ and $\vec{x}(\theta^1 +(\theta^2-\theta^1)t)$ lie on the left and right of the vertical axis, respectively.

\begin{figure}
    \centering
    \includegraphics[scale=1.7]{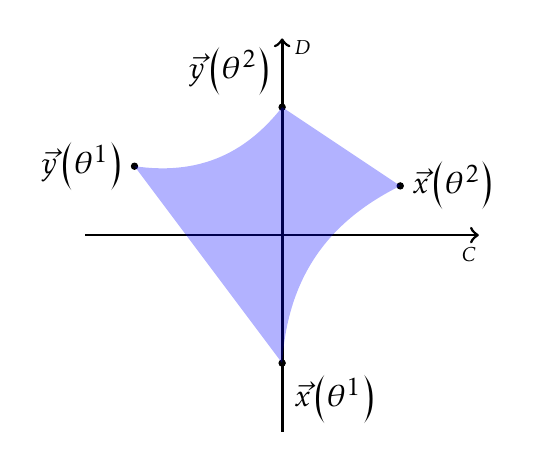}
    \caption{Homotopy}
    \label{fig:homotopy}
\end{figure}

Define $f\colon [0,1] \to \mathbb{R}^2$ as
\begin{align*}
f(t) = 
\begin{cases} 
(1-4t)\Vec{x}(\theta^1) + 4t \Vec{y}(\theta^1) &t \in \big[0,\frac{1}{4}\big] \\
\Vec{y}\big(\theta^1 + (\theta^2-\theta^1)(4t-1)\big) &t \in \big[\frac{1}{4},\frac{1}{2}\big] \\
(3-4t)\Vec{y}(\theta^2) + (4t-2)\Vec{x}(\theta^2) &t \in \big[ \frac{1}{2},\frac{3}{4}\big] \\
\Vec{x}\big(\theta^1 + 4(\theta^2-\theta^1)(1-t)\big) &t \in \big[\frac{3}{4},1\big]
\end{cases}
\end{align*}

Note that, by continuity of $U$, $f(t)$ is a loop based in $x(\theta^1)$. We can show that $f$ is null-homotopic by providing the following simple homotopy $H \colon [0,1] \times [0,1] \to \mathbb{R}^2$.
\begin{equation*}
    H(t,s) =
    \begin{cases}
            (1-4ts)\Vec{x}(\theta^1) + 4ts \Vec{y}(\theta^1) &t \in \big[0,\frac{1}{4}\big] \\
            (1-s) \Vec{x}\big(\theta^1 + s(4t-1)(\theta^2-\theta^1)\big) + s\Vec{y}\big(\theta^1 + s(4t-1)(\theta^2-\theta^1)\big) &t \in \big[\frac{1}{4},\frac{1}{2}\big] \\
            (1-s(3-4t))\Vec{x}\big((1-s)\theta^1 + s\theta^2\big) + s(3-4t)\Vec{y}\big((1-s)\theta^1 + s\theta^2\big)&t \in \big[ \frac{1}{2},\frac{3}{4}\big] \\
            \Vec{x}\big(\theta^1 + 4s(\theta^2-\theta^1)(1-t)\big) &t \in \big[\frac{3}{4},1\big]
    \end{cases}
\end{equation*}
We verify that $H$ is a homotopy between $f(t)$ and the constant path at $\Vec{x}(\theta^1)$.
\begin{itemize}
    \item $H(0,s) = H(1,s) = \Vec{x}(\theta^1)$
    \item $H(t,0) = \Vec{x}(\theta^1)$
    \item $H(t,1) = f(t)$
\end{itemize}

Suppose now that there is no pair $t,s$ such that $H(t,s) = (0,0)$. Then, by continuity it must be the case that there is an open neighborhood $V \ni (0,0)$ such that for all points $z \in V$, $z \notin Im(H)$.
Note that we can restrict the loop $f$ to a the domain $\mathbb{R}^2\setminus V$, and because $V \notin Im(H)$ we can restrict the homotopy to a homotopy $H \colon [0,1]^2 \to \mathbb{R}^2 \setminus V$. Thus we proved that a loop around the open set $V$ is null-homotopic, which is equivalent to proving that $\mathbb{R}^2\setminus V$ is simply connected, a contradiction.

Therefore, there exists a pair $\overline{t},\overline{s}$ such that $H(\overline{t},\overline{s}) = (0,0)$. There is a bijective transformation between $t,s$ and $\theta,\alpha$ over their respective domains, which guarantees that the system of \Cref{eq:nolocaltime} is satisfied.

Notice that we allowed for $\alpha \in (0,1)$ so far. However, it makes little sense to allow for $\alpha < \frac{1}{2}$: the algorithm would learn about actions it considers suboptimal faster than those he considers best. Fortunately, we observe the following:
\[\frac{1}{2} U(\theta, r(C,C)) + \frac{1}{2} U(\theta,r(C,D) < \frac{1}{2} U(\theta, r(D,C)) + \frac{1}{2} U(\theta,r(D,D). \]
This inequality follows from the ordering of rewards in a Prisoner's Dilemma game. This in particular implies that the line $\frac{1}{2} \vec{x}(\theta^1 +(\theta^2-\theta^1)t) +\frac{1}{2} \vec{y}(\theta^1 +(\theta^2-\theta^1)t)  $ lies above the 45 degree line. Thus, we can restrict the homotopy to the parameter space $\alpha \in [\frac{1}{2},1]$ without affecting the result. This guarantees the existence of a parameter $\alpha >\frac{1}{2}$ which sets the sliding vector field to zero.

Now, notice that in the previous construction in \Cref{eq:nolocaltime} we solved for $\alpha$ and $\theta^*$ such that the local time was equal on both sides of the switching boundary. We can relax this assumption so that \Cref{eq:nolocaltime} becomes
\begin{equation}\label{eq:withlocaltime}
\begin{cases}
\alpha \tau U(\theta^*,r(C,C)) + (1-\alpha) (1-\tau) U(\theta^*,r(C,D)) = 0\\
(1-\alpha) \tau U(\theta^*,r(D,C)) + \alpha (1-\tau) U(\theta^*,r(D,D)) = 0.
\end{cases}\end{equation}
Following the previous construction, we get 
\begin{align*}
    \Vec{x}(\theta) = \begin{bmatrix}
           \tau U(\theta,r(C,C)) \\
           (1-\tau) U(\theta,r(D,D)) \\
         \end{bmatrix} \qquad \qquad
    &\Vec{y}(\theta) = \begin{bmatrix}
           (1-\tau) U(\theta,r(C,D)) \\
           \tau U(\theta,r(D,C)) \\
         \end{bmatrix}
\end{align*}
The points $\vec{x}(\theta^i), \vec{y}(\theta^i)$ for $i=1,2$ each remain on their respective quadrants, enabling us to construct the same, rescaled, homotopy for this case. Therefore, for each $\tau$ we can identify a pair $\alpha^\tau, \theta^*(\tau)$ such that \Cref{eq:withlocaltime} is satisfied. 

Not all solutions to \Cref{eq:withlocaltime} are steady-states however. Recall from our construction in \Cref{sec:PD} that we need to verify a sliding condition: the normal components of the two vector fields to the switching surface must have opposite sign and must be attractive. 
In equations, this translates to the following system which needs to be satisfied:
\begin{equation}\label{eq:sliding}
    \begin{cases}
    (1-\alpha^\tau) \tau U(\theta^*(\tau),r(D,C)) - \alpha\tau  U(\theta^*(\tau),r(C,C)) \geq 0\\
    \alpha^\tau (1-\tau) U(\theta^*(\tau),r(D,D)) - (1-\alpha^\tau) (1-\tau)  U(\theta^*(\tau),r(C,D)) \leq 0.
    \end{cases}
\end{equation}
We need a preparatory lemma:
\begin{lemma}\label{lemma:regions}
The only region of $\theta$ where \Cref{eq:withlocaltime} can be verified is such that 
\[U(\theta,r(D,C))>U(\theta,r(C,C))\geq 0 >U(\theta,r(D,D)) >U(\theta,r(C,D))\]
\end{lemma}
\begin{proof}
The statement follows immediately from inspection of \Cref{fig:homotopy}. Since the path $\vec{y}(\theta^1+(\theta^2-\theta^1)t)$ for all $t$ falls within the second quadrant, any $\vec{x}(\theta^*)$ which satisfies \Cref{eq:withlocaltime} must fall within the fourth quadrant, which directly implies the result.
\end{proof}
Now, rearranging \Cref{eq:withlocaltime} we derive the following equalities:
\begin{equation*}
    \begin{cases}
        -(1-\alpha^\tau)(1-\tau)U(\theta^*(\tau),r(C,D)) = \alpha^\tau \tau U(\theta^*(\tau),r(C,C)) \\
        -(1-\alpha^\tau) \tau U(\theta^*(\tau),r(D,C)) = \alpha^\tau (1-\tau) U(\theta^*(\tau),r(D,D)) 
    \end{cases}
\end{equation*}
Substituting these in \Cref{eq:sliding} and rearranging, we obtain
\begin{equation*}
    \begin{cases}
        \alpha^\tau \tau U(\theta^*(\tau),r(C,C)) + \alpha^\tau (1-\tau) U(\theta^*(\tau),r(D,D)) \leq 0 \\   \alpha^\tau (1-\tau) U(\theta^*(\tau),r(D,D)) + \alpha^\tau \tau U(\theta^*(\tau),r(C,C)) \leq 0
    \end{cases}
\end{equation*}
Thus, only one condition needs to be satisfied to guarantee sliding:
\[\tau U(\theta^*(\tau),r(C,C)) + (1-\tau)U(\theta^*(\tau),r(D,D)) \leq 0\]
\Cref{lemma:regions} guarantee that a solution to this inequality exists for $\tau$ sufficiently close to $0$. In particular, there exists a $\overline{\tau}$ such that for all $\tau \leq \overline{\tau}$ we obtain sliding. Therefore there exists an open set of parameters $\{\alpha^\tau \ |\  \tau \leq \overline{\tau}\}$ such that a sliding steady-state exists. 
Now, we can perturb the value of $\alpha$ on either side of the sliding boundary. The region we derived depends continuously from $\alpha$, in particular from $\alpha^C$ and $\alpha^D$. Therefore,  the same result holds for small perturbations of $\alpha$ into different $\alpha^C$ and $\alpha^D$, concluding the proof of the Theorem.

\qed

\paragraph{Proof of \Cref{thm:domstr}.}
We set to prove that, in the limit, the statistic $\theta^n$ of the dominant action dominates the statistic $\theta^i$ of any other non-dominant action. 
First of all, since $\alpha^{a_i}$ is identical across actions, the evolution in time of $\theta$ is shifted by $V(\theta)$, but $V$ won't affect the relative rankings of the actions' estimates. Therefore, we drop $V$ in the rest of the proof, and we focus on $U$.

Regardless of the opponent's policy, the reinforcer in every step observes a return in hindsight for every action. We denote by $r_n(t)$ the return from playing action $n$ in period $t$, whatever the opponents' actions are. We consider the evolution of the weights pairwise: $\theta^n$ and $\theta^i$ for all $i$. By assumption, for any sequence of actions taken by the opponent(s), $r_n(t) \geq r_i(t)$. In particular, \Cref{manualasm:thickness} guarantees that there exists a $T >0$ such that $r_n(t) > r_i(t)$ for any $t>T$. Thus, in the limit the derivative $\dot{\theta}^n$ strictly dominates $\dot{\theta}^i$ in each point $(x,t)$: $U(x,r_n(t))> U(x,,r_i(t))$. In particular, there exists a $\varepsilon$ such that $U(x,r_n(t))> U(x,r_i(t)) + \varepsilon$. 

Suppose now that for some $t\geq 0$, $\theta^n(t)= \theta^i(t)$. We prove that it can never be the case that $\theta^i(T) > \theta^n(T)$ for some $T>t$. $\theta^n(t)$ is a solution to the ODE given by $\dot{\theta}^n(t) = U(\theta^n(t),r_n(t))$ and 
\[U(\theta^n(t),r_n(t)) \geq U(\theta^n(t),r_i(t)) =U(\theta^i(t),r_i(t))\]
Thus, $\dot{\theta}^i(t) < U(\theta^i(t),r_n(t))=\dot{\theta}^n(t)$, which implies that $\theta^n(t+\Delta)>\theta^i(t+\Delta)$ for $\Delta$ small enough. Thus, for any $T>t$ there can only be two cases: either $\theta^n(T)>\theta^i(T)$, or $\theta^n(T)=\theta^i(T)$; but we have just shown that if the latter case occurs, $\theta^i$ will stay below $\theta^n$ again.

Consider instead the case where $\theta^i(T) >\theta^n(T)$ for some $T>0$. From \Cref{def:regularity}, we have that 
\[U(\theta^i(T),r_i(T)) \leq U(\theta^i(T),r_n(T))< U(\theta^n(T),r_n(T))\]
We want to show that there exists a $t> T$ such that $\theta^n(t) = \theta^i(t)$. To this end, suppose by contradiction that $\forall t>T$, $\theta^i(t) > \theta^n(t)$. Observe that the previous inequalities imply that
\begin{equation}\label{eq:derivative}
    \frac{d}{ds}\bigg|_{s=t}(\theta^i(s) - \theta^n(s)) <0  \qquad \forall t >T.
\end{equation}
Then, $(\theta^i(t) - \theta^n(t))$ is a monotone decreasing function of time. Because the algorithm is bounded,  it must be the case that 
\begin{equation}\label{eq:limit=b}
    \lim_{t\to +\infty} (\theta^i(t) - \theta^n(t)) = b \geq 0.
\end{equation}
From the definition of limit and the monotonicity of the derivative, for any $\epsilon$ there exists a $t'>T$ such that, for all $t>t'$, \[\theta^n(t) + b \leq \theta^i(t) < \theta^n(t) + b + \epsilon.\] 
However, strict monotonicity of the reinforcer's update implies that there exists a $\delta>0$ such that  
\begin{equation}\label{eq:positivedifference}
    U(\theta^i(t),r_i(t)) \leq U(\theta^n(t),r_i(t)) < U(\theta^n(t),r_n(t)) - \delta.
\end{equation}
Note that, by \Cref{eq:limit=b} and the monotonicity given by \Cref{eq:derivative}, the limit of the derivative of the difference $\thetaî - \theta^n$ be $0$:
\[\lim_{t\to +\infty} (\dot{\theta}^i(t) - \dot{\theta}^n(t)) = \lim_{t\to +\infty} (U(\theta^i(t),r_i(t)) -  U(\theta^n(t),r_n(t))) = 0\]
This is a contradiction of  \Cref{eq:positivedifference}. Then, there must exist a $t$ such that $\theta^i(t) = \theta^n(t)$. From the first part of the proof, this implies that $\forall t'  >t$ we have $\theta^i(t') \leq \theta^n(t')$. In summary then, either $\theta^n(0)\geq \theta^i(0)$, which implies $\theta^n(t) \geq \theta^i(t)$ for all $t>0$, or $\theta^n(0)<\theta^i(0)$, which implies that there exists $T$ such that $\theta^n(t) \geq \theta^i(t)$ for all $t\geq T$. This concludes the proof.  \qed

\paragraph{Proof of \Cref{thm:IESDS}.}
The proof is largely based upon \Cref{thm:domstr}. We will show that there is always a $\mathcal{T}$ such that the actions taken after $\mathcal{T}$ survive an IESDS procedure. 

Let $a^i$ be a strategy for player $i$ which is dominated by $b^i$ in the game $G$. With the same argument of the proof of \Cref{thm:domstr} we can show that it must be the case that there exists a $\mathcal{T}_1$ such that for all $t\geq \mathcal{T}_1$ we have $\theta^{b^i}(t) \geq \theta^{a^i}(t)$. Notice that if $\theta^{b^i}(t) = \theta^{a^i}(t)$, then $\dot{\theta}^{b^i}(t) > \dot{\theta}^{a_i}(t)$, which implies that for all $t'>t$ we have a strict inequality $\theta^{b^i}(t') > \theta^{a^i}(t')$.

Therefore, after time $\mathcal{T}_1$ it is as if the agents were playing in a reduced game $G^1 = (N,(A^1_i)_i,(u_i)_i)$, where $A^1_i = A_i\setminus\{a^i\}$ and $A^1_j = A_j$ for all $j \neq i$.\footnote{Of course, \Cref{manualasm:thickness} guarantees that even action $a^i$ will be played with some positive probability, but the statement of \Cref{thm:IESDS} guarantees that the probability is small enough to not affect the order of the expected rewards.} We can now apply the same idea to this new reduced game $G^1$ and eliminate a strictly dominated strategy in this reduced game. While IESDS eliminates strategies ``in place'', the algorithms abandon dominated strategies over time, reducing the effective game played. Of course, the components of $\theta$ that correspond to dominated actions keep getting updated, but note that if an action is dominated given a larger subset of opponent's strategies, it is also (weakly) dominated given a smaller subset. Therefore, following the usual differential inequality argument, the $\theta$ corresponding to a dominated action will always remain below that of actions surviving IESDS. We then define $\mathcal{T}$ as the largest among the $\mathcal{T}_k$ that correspond to an action being eliminated, and this concludes the proof. \qed

\section{}\label{app:chaos}

In this Appendix we discuss the chaotic theory of the system in \Cref{sec:PD}. Chaos theory studies non-linear dynamical systems whose trajectories appear stochastic, but are the result of deterministic laws of motion. The most commonly accepted definition of chaotic behavior is high sensitivity to small perturbations in initial conditions: a system is said to be chaotic if the distance between two trajectories originating from different yet arbitrarily close points grows unboundedly (\citet{holden2014chaos}).

Let us consider again the dynamical system of \Cref{eq:allODE}, where we now relax \Cref{manualasm:symmetry} and allow Alice' and Bob's Q-values to be initialized in different points. Following the above definition of chaotic system, we plot in \Cref{fig:perturbation} the evolution of $Q^A_C(t)$ for two instances with slightly different $Q^A_C(0)$, as well as the distance between the two trajectories. Despite their very similar initial point the two trajectories vary wildy, and their distance increases by 20 orders of magnitude: it is clear that the 4-dimensional system outside the scope of \Cref{manualasm:symmetry} exhibits chaotic behavior.

\begin{figure}[hbtp]
    \centering
    \includegraphics[width = \textwidth]{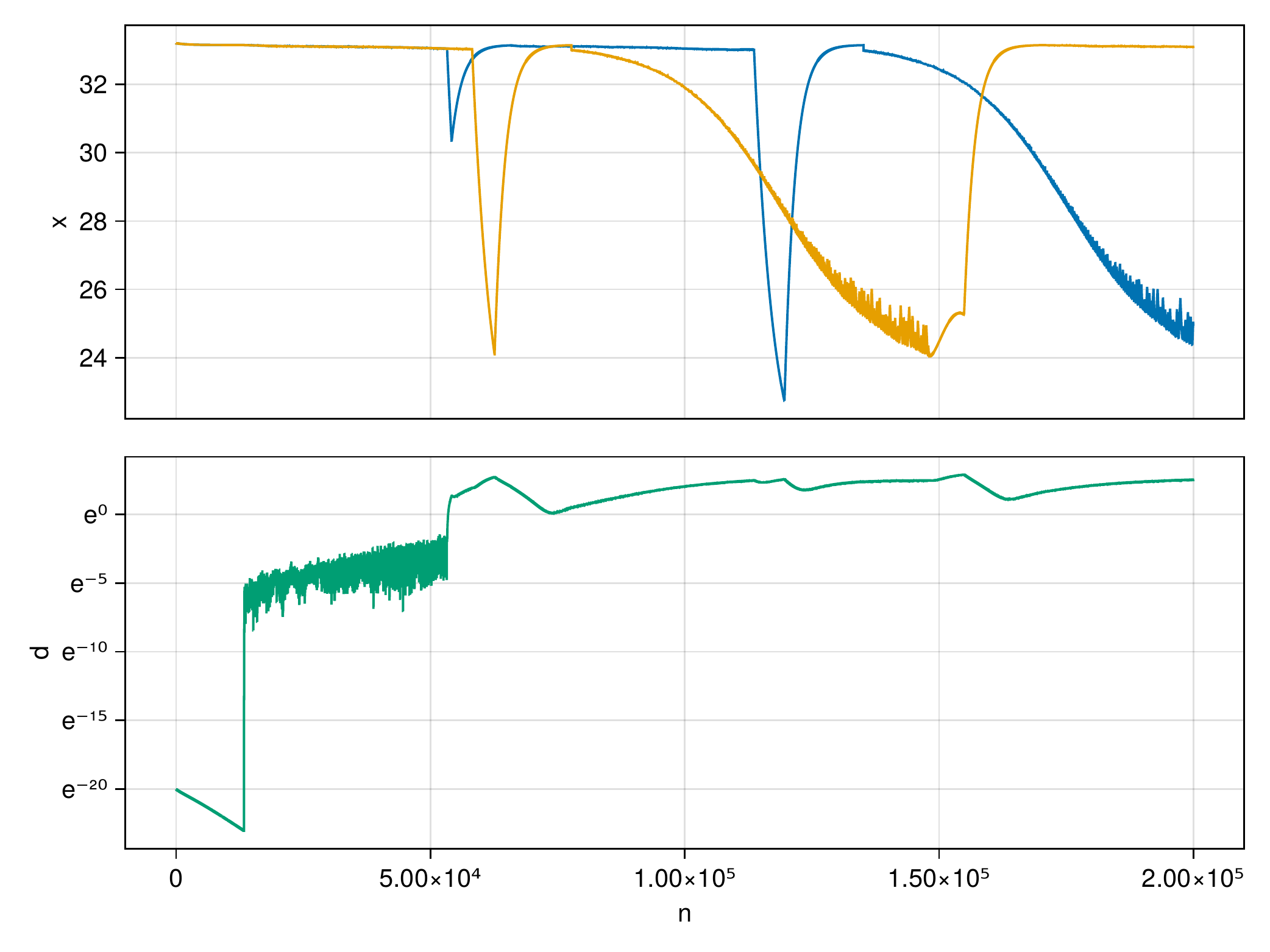}
    \caption{The top figure depicts the evolution of (the first component of) two trajectories with minimal perturbations of their initial conditions. The bottom figure represents the absolute distance between these two trajectories, in logarithmic scale. After an initial convergence, hitting the boundary leads to exponential divergence (the green line has constant slope upwards in the logarithmic scale). Finally the trajectories saturate, i.e. the boundedness of trajectories limits their absolute divergence.}
    \label{fig:perturbation}
\end{figure}

The deterministic chaos we observe makes any stability analysis impossible. In \Cref{fig:chaotic_trajectory} we plot a single trajectory of the 4-D system in the space of differences, and we observe that it seems to behave erratically and upredictably: the ``butterfly'' traced by the trajectory appears hard to characterize. However, the heatmap shows that most of the time is spent in a region where the estimates of cooperation and defection coincide. In other words, most of the time the system will show equal values for cooperation and defection for both Alice and Bob. Therefore, it seems that the sliding analysis we carried out in the symmetric case is not too far off the mark in the 4-D case as well, as one can also see in \Cref{fig:Local Time}: the gap between the predicted local time and the realized local time is due to the asymmetry that always arises in the discrete algorithmic system. It appears that the coordination bias's cycles take place in 4 dimensions --- they do not collapse on a pseudo-equilibrium.
 
\begin{figure}[hbtp]
    \centering
    \begin{subfigure}{.47\textwidth}
    \centering
    \includegraphics[scale=0.9]{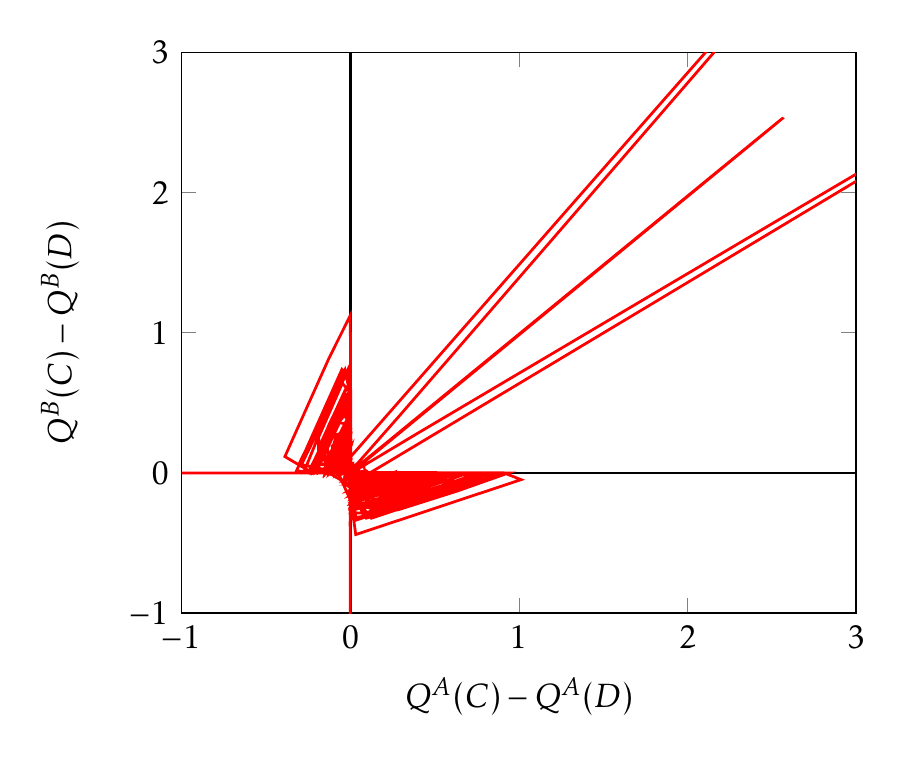}
    \end{subfigure}%
    \begin{subfigure}{.54\textwidth}
    \centering
     \includegraphics[scale=0.9]{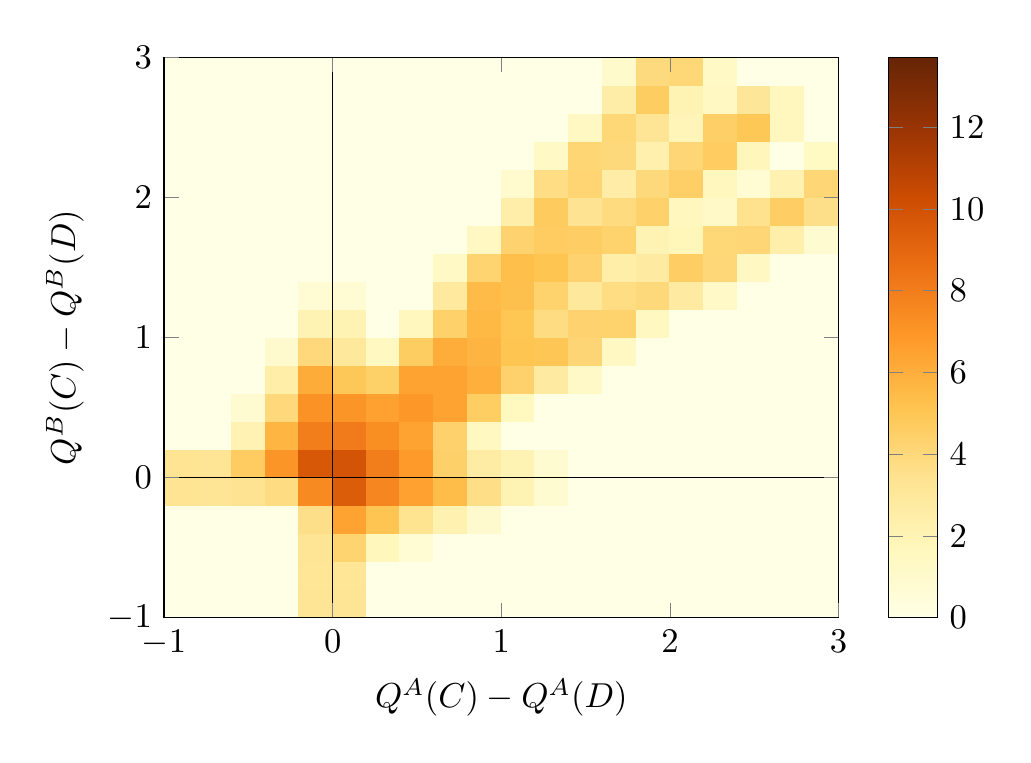}
    \end{subfigure}
    \caption{A chaotic trajectory of the system of \Cref{eq:allODE} in a 2-D representation. On the left, we depict the motion in the difference space. Agents appear chaotically attracted to a butterfly-like motion around the origin. The figure on the right represents the logarithmic density of time spent in a given square cell of side $0.2$. Effectively, over $90\%$ of time is spent in the square centered on the origin.}
    \label{fig:chaotic_trajectory}
\end{figure}

\subsection{Sliding Surfaces and Chaos}\label{app:chaos math}
We now formalize some of our claims above. Consider the system of \Cref{eq:allODE} and let $F_{a,a^\prime}$ denote the vector field in region $\omega_{a,a^\prime}$. There exist two switching surfaces, one for Alice ($\Sigma^A$) and one for Bob ($\Sigma^B$), given by:
\begin{align*}
    &\Sigma^A=\left\{\mathbf{Q} \in \mathbb{R}^4:\;Q^A_C=Q^A_D\right\}\\
    &\Sigma^B=\left\{\mathbf{Q} \in \mathbb{R}^4:\;Q^B_C=Q^B_D\right\}.
\end{align*}
Using our notation, $\Sigma^A=\left(\overline{\omega}_{C,C} \cap \overline{\omega}_{D,C}\right) \bigcup \left(\overline{\omega}_{C,D} \cap \overline{\omega}_{D,D}\right)$, and $\Sigma^B=\left(\overline{\omega}_{C,C} \cap \overline{\omega}_{C,D}\right) \bigcup \left(\overline{\omega}_{D,C} \cap \overline{\omega}_{D,D}\right)$. Sliding motion, if any, occurs on these hyperplanes, provided the vector fields on each side ``point'' in the right direction. In particular, if the system is on $\Sigma=\Sigma^A \cap \Sigma^B$, there might be \emph{double sliding}. Intuitively, double sliding occurs when the vector fields in all four regions ``push'' the system towards surface $\Sigma$. We make this formal following the definition of \citet{dieci2011}.

Let $\mathbf{\hat{N}^i}$ denote a vector in $\mathbb{R}^4$ orthogonal to $\Sigma^i$, where, without loss of generality, we choose $\mathbf{\hat{N}^i}$ such that it points towards the region of space where $Q^i_C>Q^i_D$. Surface $\Sigma$ is said to be \emph{nodally} attractive if the following hold on, and in a neighbourhood of, $\Sigma$:
\begin{equation}\label{eq:nodal attractivity}
    \begin{matrix}
    \mathbf{\hat{N}^A} \cdot F_{D,D}>0, & \mathbf{\hat{N}^A} \cdot F_{D,C}>0, & \mathbf{\hat{N}^A} \cdot F_{C,D}<0, & \mathbf{\hat{N}^A} \cdot F_{C,C}<0\\
    \mathbf{\hat{N}^B} \cdot F_{D,D}>0, & \mathbf{\hat{N}^B} \cdot F_{D,C}<0, & \mathbf{\hat{N}^B} \cdot F_{C,D}>0, & \mathbf{\hat{N}^A} \cdot F_{C,C}<0.\\
\end{matrix}
\end{equation}

In words, \Cref{eq:nodal attractivity} requires each surface $\Sigma^i$ to be sliding irrespective of the action of the opponent: if the system is initialized in a region of the space where $\Sigma$ is nodally attractive, it will first slide on either $\Sigma^i$ towards $\Sigma$, where it finally begins the double sliding motion.

Suppose for simplicity that $\gamma=0$.\footnote{The same analysis goes through with $\gamma>0$, but equations are somewhat more involved and less transparent to parse.} When the system is on $\Sigma$, $\mathbf{Q}=[q,q,q^\prime,q^\prime]^T$ for some $q,q^\prime \in \mathbb{R}$ and the conditions in \Cref{eq:nodal attractivity} become
\begin{align*}
    &q> \overline{q} = g+1+\dfrac{\frac{\varepsilon^2}{2}g+1-\varepsilon g}{1-\varepsilon}\\
    &q^\prime> \overline{q} = g+1+\dfrac{\frac{\varepsilon^2}{2}g+1-\varepsilon g}{1-\varepsilon}
\end{align*}
which square well with \Cref{fig:vector fields}, which shows that in the symmetric case, the switching surface is sliding only for large enough values of the estimates. The right-hand sides of the above conditions are the same because, even if Alice's and Bob's estimates are no longer symmetric, their laws of motion are. In fact, $\overline{q}$ can be equivalently characterized as
\begin{equation}\label{eq:q bar}
    \overline{q}=\max \left\{q:\; \exists (a,a^\prime) \in \{C,D\}^2 \text{ s.t. } \mathbf{\hat{N}^A} \cdot F_{a,a^\prime}\left([q,q,x,y]^T\right)=0,\; [q,q,x,y]^T \in \omega_{a,a^\prime}\right\},
\end{equation}
or, in words, the largest value of the estimates such that if Alice were sliding on her switching surface $\Sigma^A$, her vector field in at least one of the partitions of the space becomes parallel to $\Sigma^A$ ,i.e., $\dot{Q}^A(C)=\dot{Q}^A(D)$. Since the game is symmetric, we could replace Alice with Bob in \Cref{eq:q bar} and obtain the same result. The structure of payoffs of any Prisoner's Dilemma implies that the $\omega$ where Alice's vector field becomes parallel to her switching surface is $\omega_{D,C}$: it is exactly when she defects most of the time and Bob cooperates that Alice obtains the largest payoff.

The codimension of $\Sigma$ is 2; while the theory of \citet{Filippov1988} is sufficient to uniquely pin down the sliding vector field for sliding surfaces of codimension 1, it is well-known that for codimension greater than one this is no longer possible. \citet{dieci2011} suggests defining the weighted vector field using a bilinear formulation and proves that if the sliding surface satisfies nodal attractivity, the vector field thus obtained is unique. Formally, let $\mathcal{C}^A(\mathbf{Q}) \in [0,1]$ and $\mathcal{C}^B(\mathbf{Q})\in [0,1]$ be some functions of the state of the system: we interpret $\mathcal{C}^A(\mathbf{Q})$ as the weight associated to any $F_{C,a^\prime}$, where  for $a^\prime=C,D$, and similarly for $\mathcal{C}^B(\mathbf{Q})$. \citet{dieci2011} proves that there exist unique functions $\mathcal{C}^A(\mathbf{Q})$ and $\mathcal{C}^B(\mathbf{Q})$ such that for every $\mathbf{Q} \in \Sigma$ satisfying nodal attractivity the following holds:\footnote{Omitting dependence on $\mathbf{Q}$ for readability.}
\begin{equation}\label{eq:double sliding}
    \begin{cases}
        \mathcal{C}^A \mathcal{C}^B \mathbf{\hat{N}^A} \cdot F_{C,C}+(1-\mathcal{C}^A) \mathcal{C}^B \mathbf{\hat{N}^A} \cdot F_{D,C} + (1-\mathcal{C}^B)\mathcal{C}^A \mathbf{\hat{N}^A} \cdot F_{C,D} + (1- \mathcal{C}^A) (1-\mathcal{C}^B) \mathbf{\hat{N}^A} \cdot F_{D,D} = 0\\
        \\
        \mathcal{C}^A \mathcal{C}^B \mathbf{\hat{N}^A} \cdot F_{C,C} + (1-\mathcal{C}^A) \mathcal{C}^B \mathbf{\hat{N}^B} \cdot F_{D,C} + (1-\mathcal{C}^B)\mathcal{C}^A \mathbf{\hat{N}^B} \cdot F_{C,D} + (1- \mathcal{C}^A) (1-\mathcal{C}^B) \mathbf{\hat{N}^B} \cdot F_{D,D} = 0
    \end{cases}
\end{equation}
This condition has the same intuition of the construction mentioned in \Cref{sec:sketch}: by appropriately selecting $\mathcal{C}^A(\mathbf{Q})$ and $\mathcal{C}^B(\mathbf{Q})$ one can obtain a vector field that is parallel to $\Sigma$ so that the trajectories of the system remain contrained on the double sliding surface as long as this is nodally attractive.

In our system, there exists no steady state on the sliding boundary. In fact, since we argued that $\overline{q}$ is such that $$\hat{\mathbf{N}}^A \cdot F_{D,C}\left([\overline{q},\overline{q},x,y]^T\right)=0,$$ it must be that all vector fields have strictly negative parallel components to the double sliding surface $\Sigma$ as long as the estimates are above $\overline{q}$. Finally, this implies that there cannot exist any mixing of the type of \Cref{eq:double sliding} such that the combined vector field on $\Sigma$ vanishes. Formally, when Alice's estimates reach $\overline{q}$, the system leaves the 2-dimensional surface $\Sigma$ so that its motion can no longer be described by the nodally attractive sliding vector field. However, notice that the description of spontaneous coupling developed in \Cref{sec:coupling} does not require a full characterization of the motion; while based on the insights derived from the sliding motion of the symmetric system, it is more general since it relies only on understanding how the estimates move across the $\omega$ sets. As we argue in \Cref{fig:chaotic_trajectory}, the system indeed spends the vast majority of the time in a small neighbourhood of the switching surface $\Sigma$: while we cannot characterize the laws of motion \emph{on} $\Sigma$, we know the incentives that move the estimates from one $\omega$ to the other. Because of chaos, the ``cycling'' patterns we observe in the discrete system also remain in the continuous-time counterpart instead of collapsing onto a single point;\footnote{We use quotes because, formally, a chaotic system cannot display period behaviour.} nevertheless, they are consistent with our higher level understanding of spontaneous coupling.

\section{}\label{app:feedback}
In this Appendix we formalize the somewhat loose definitions given in \Cref{sec:learnable} and we prove \Cref{thm:LRSM}.

First, note that we can define an equivalence relation on the outcome space, $\sim_i$, such that $x \sim_i y$ if and only if $u^i(x,\lambda^i) = u^i(y, \lambda^i)$ for all $\lambda^i \in \Lambda_{i}$.
Let $\mathcal{X}_i = \mathcal{X}/\sim_i$ be the quotient of the outcome space with respect to this equivalence relation. We will refer to an element of $\mathcal{X}_i$, an equivalence class $[x]_i$, as an $i$-outcome.

Let us define formally what it means to provide ex-post feedback.

\begin{definition}
    A \emph{feedback policy} for mechanism $f$ and agent $i$ is a factorization of $f$ through a partition. It is composed of the following elements:
    \begin{itemize}
        \item A \emph{signal space} $S$, which is a partition of $\Lambda_{-i}$;
        \item A map $\phi_i \colon \Lambda_{-i} \to S$;
        \item A map $g \colon \Lambda_i \times S \to \mathcal{X}_i$;
    \end{itemize}
    such that
    \begin{itemize}
        \item The diagram commutes, i.e. $[f(\lambda^i,\lambda^{-i})]_i = g(\lambda^i,\phi_i(\lambda^{-i}))$ for all $\lambda^i, \lambda^{-i}$;
        \item The map $\phi_i$ is a partition map, i.e.
        $\phi_i(\lambda^i) \ni \lambda^i$.
    \end{itemize}
    We denote a feedback policy by its map $\phi_i$.
    A collection $(\phi_i)_i$ of feedback policies for each agent $i$ is a \emph{feedback structure}.
\end{definition}

Remember the loose definition of \Cref{sec:learnable}: a feedback policy for agent $i$ is a partition of the space of opponents' types, $\Lambda_{-i}$. This partition is such that agent $i$ can evaluate what outcome he could have enforced had he unilaterally deviated to a different report $\lambda^i$. We request that a feedback policy allows agent $i$ to compute all of his $i$-outcomes for any given report.

Next, we formalize the desire for reduced communication and revelation.  We define the following partial order on feedback policies: 
\begin{definition}
    Feedback policy $\phi_i$ is \emph{more private} than feedback policy $\psi_i$, denoted $\phi_i \trianglerighteq \psi_i$, if $\phi_i(\lambda^{-i}) \supseteq \psi_i(\lambda^{-i})$ for all $\lambda^{-i}.$
\end{definition}

Any two type profiles that are indistinguishable under a less private rule should be indistinguishable under a more private rule. As mentioned, the privacy order is a weak partial order: not all feedback policies are comparable. However, it turns out that under the privacy order maximum and minimum are well-defined: the feedback policies together with the privacy order form a lattice.

\begin{proposition}\label{prop:lattice}
Feedback policies together with the privacy order form a complete lattice.
\end{proposition}
\begin{proof}
We simply need to show that for any two elements $\phi_i, \psi_i$ there exist a join $\phi_i \join \psi_i$ and a meet $\phi_i \meet \psi_i$ which satisfy the feedback policy definition. The argument follows from the lattice structure of the set of partitions with the partial order \emph{coarser-than}. Note that $\phi_i^{-1}(\{x \colon x \in  \Lambda_{-i}\})$ defines a partition of the space $\Lambda_{-i}$, and the same is true for the $\psi_i$. We then require the preimage of join $(\phi_i \join \psi_i)$ to be the finest partition which is coarser than both the preimages of $\phi_i$ and $\psi_i$.  Formally, let $A \subset \phi_i^{-1}(\{x \colon x \in \Lambda_{-i}\}) \join_P \psi_i^{-1}(\{x \colon x \in \Lambda_{-i}\})$, then 
\[(\phi_i \join \psi_i)(\lambda^{-i}) = (\phi_i \join \psi_i)(\hat{\lambda}^{-i}) \text{ for all } \lambda^{-i},\hat{\lambda}^{-i} \in  A\]

Similarly, define the meet as the function $\phi_i \meet \psi_i$ such that it is constant over the meet of the two partitions. Again, let $B \subset \phi_i^{-1}(\{x \colon x \in \Lambda_{-i}\}) \meet_P \psi_i^{-1}(\{x \colon x \in \Lambda_{-i}\})$, then 
\[(\phi_i \meet \psi_i)(\lambda^{-i}) = (\phi_i \meet \psi_i)(\hat{\lambda}^{-i}) \text{ for all } \lambda^{-i},\hat{\lambda}^{-i} \in  B\]
The completeness of the lattice structure descends directly from the completeness of the partition lattice.
\end{proof}

\Cref{prop:lattice} implies that there exist both a minimally-  and a maximally-private feedback policy. The minimally-private policy is the full-revelation feedback policy: it reveals all information, and it is clearly less private than any other feedback policy. The maximally-private rule instead is a menu description.\footnote{Note that we defined the privacy lattice for feedback policies, not for feedback structures. When we analyze feedback structures, we implicitly consider the product lattice on the product space of feedback policies. This is correct because we assume private communication, but the analysis would likely change if we allowed public communication.}

\begin{definition}\label{def:menu}
Let $[x]_i$ be the equivalence class of outcome $x$ in $\mathcal{X}_i$.
A \emph{menu} for mechanism $f$ and agent $i$, given reports $\lambda^{-i}$ of the opponents, is the set 
\[\mathcal{M}_{\lambda^{-i}} = \left\{\ [f(\hat{\lambda}^i,\lambda^{-i})]_i \ | \ \hat{\lambda}^i \in \Lambda_i \ \right\}.\]
That is, the menu is the set of outcomes such that agent $i$ could have received had he reported any type in his type space.
\end{definition}

We can show that from the collection of all possible menus $\left\{ \mathcal{M}_{\lambda^{-i}}\right\}_{\lambda^{-i}}$ we can construct a feedback policy, and it is the maximally private feedback policy for agent $i$.

\begin{lemma}\label{lemma:optimalmenu}
    There exists a feedback policy $\mu^i$ corresponding to the collection of menus $\left\{ \mathcal{M}_{\lambda^{-i}}\right\}_{\lambda^{-i}}$, and $\mu^i$ is the maximally private feedback policy for agent $i$. 
\end{lemma}
\begin{proof}
Consider the space $\Lambda_{-i}$ with the following equivalence relation:
\[ \lambda^{-i} \sim \hat{\lambda}^{-i} \text{ iff } \mathcal{M}_{\lambda^{-i}} = \mathcal{M}_{\hat{\lambda}^{-i}}\]
and denote its quotient by $\Lambda_{-i}/\sim$. An element of the quotient is an equivalence class of opponents' types, denoted by $\left[\lambda^{-i}\right]$. Since $\sim$ is an equivalence relation, the quotient set $\Lambda_{-i}/\sim$ is a partition of $\Lambda_{-i}$. Let 
\[
\mu^i \colon \begin{array}[t]{c c c} 
          \Lambda_{-i} &\rightarrow& \Lambda_{-i}/\sim \\ 
          \lambda^{-i} &\mapsto&  \left[\lambda^{-i}\right]
         \end{array}
\]
Let us show first that $\mu^i$ satisfies the definition of feedback policy. Of course, $\mu^i$ is a partition map: an element always belongs to its equivalence class. Define $g$ as the function $g\left(\lambda^i,\left[ \lambda^{-i}\right]\right) := [f(\lambda^i,\lambda^{-i})]_i.$ Commutativity then follows from the fact that for all $\lambda^{-i} \in \left[ \lambda^{-i}\right]$ the menus $\mathcal{M}_{\lambda^{-i}}$ coincide. 

Now, we need to show that there is no feedback policy which is more private than $\mu^i$. Equivalently, we can show that there is no feedback policy $\phi^i$ such that $\phi^i \triangleright \mu^i$. Suppose instead such a $\phi^i$ exists. Then, it must be that there exists a $\lambda^{-i}$ such that $\phi^i(\lambda^{-i}) \supset \mu^i(\lambda^{-i})$. Then there exists a $\hat{\lambda}^{-i} \in \phi(\lambda^{-i})$ such that $\hat{\lambda}^{-i}$ is not in the same equivalence class of $\lambda^{-i}$. This implies that $\mathcal{M}_{\lambda^{-i}}$ is a different menu than $\mathcal{M}_{\hat{\lambda}^{-i}}$. Then, there exists a $\lambda^i$ such that $[f(\lambda^i,\lambda^{-i})]_i \neq [f(\lambda^i,\hat{\lambda}^{-i})]_i$. 
We have then that $g(\lambda^i,\phi(\lambda^{-i})) = [f(\lambda^i, \lambda^{-i})]_i \neq [f(\lambda^i, \hat{\lambda}^{-i})]_i = g(\lambda^i,\phi(\lambda^{-i})$, a contradiction.
\end{proof}

We observe another interesting property of the privacy order that stems from its connection with the set of partitions of the power set of $\Lambda_{-i}$: 
\begin{corollary}
    The communication complexity of a feedback policy is monotonically decreasing in the privacy order.
\end{corollary}

The maximally private feedback policy has the lowest communication complexity and it maintains the highest level of privacy. Our characterization says that the most private mechanisms are not particularly esoteric: they provide menu descriptions, which are well-known for their simplicity.

\section{}\label{app:PD}

The contribution game of \Cref{sec:PD} can be seen as a particular one-dimensional parametrization of the Prisoner's Dilemma. In this Appendix we extend our analysis to general symmetric Prisoner's Dilemma games. We parametrize them as described in \Cref{fig:PDpayoffs}. We normalize the cooperation payoff to $1$ and the sucker's payoff to $0$, while we vary the payoff to deviation $x$ and the payoff to mutual defection $y$.

\begin{figure}[hbtp]
    \centering
    \begin{subfigure}[b]{0.34\textwidth}
    \centering
    \imagebox{55mm}{\begin{game}{2}{2}[Alice][Bob]
    \> $C$ \> $D$\\
    $C$ \>$1,1$ \>$0,x$\\
    $D$ \>$x,0$ \>$y,y$
    \end{game}}
    \caption{Payoffs of the stage game. \\$1< x < 2$, $0<y<1$.}
    \label{fig:PDpayoffs}
    \end{subfigure}%
    \begin{subfigure}[b]{0.64\textwidth}
    \centering
    \includegraphics[width=\textwidth]{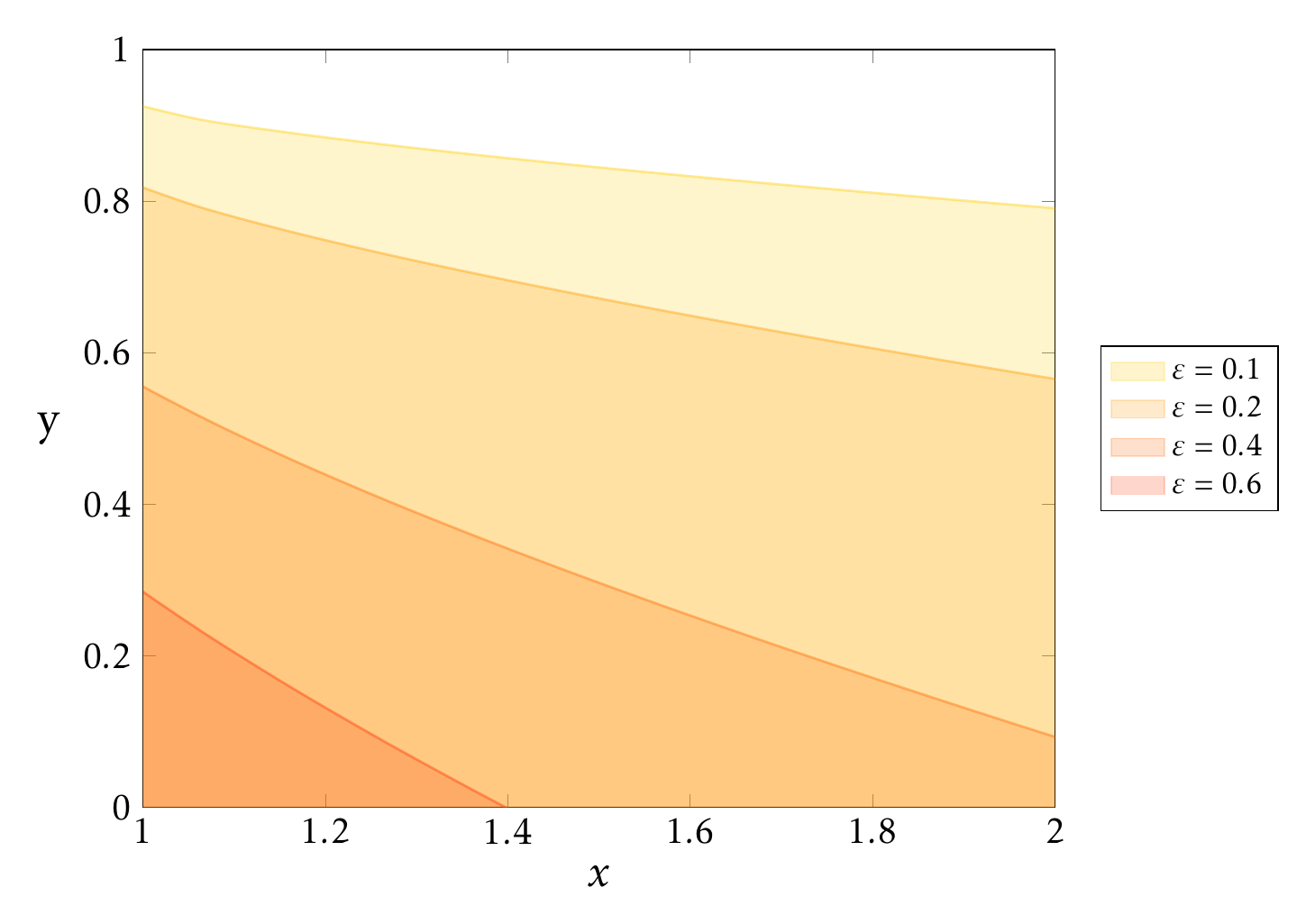}
    \caption{Existence of cooperative pseudo-steady-state in the Prisoner's Dilemma parameter space.}
    \label{fig:PDss}
    \end{subfigure}
    \caption{}
    \label{fig:generalPD}
\end{figure}

 We can replicate the analysis carried out for the contribution game in this more general setting, and we reach similar conclusions.

\begin{proposition}\label{prop:PD}
Consider a Prisoner's Dilemma with payoffs as in \Cref{fig:PDpayoffs} played by $\varepsilon$-greedy Q-learning algorithms. The forward limit set of $\mathbf{Q}$ is a singleton for any initial condition. 
Suppose the following inequalities are satisfied:
\begin{equation}\label{eq:existence}
    \begin{cases}
        1 < x < \frac{4 + 2\varepsilon - \varepsilon^2}{2\varepsilon - \varepsilon^2} - 4\sqrt{\frac{1}{2\varepsilon - \varepsilon^2}},\\
        0 \leq y \leq -4\sqrt{\frac{(\varepsilon-2)^2\varepsilon^2[\varepsilon^2(x-2) - 2\varepsilon(x-2) + 4(x-1)]}{(4-2\varepsilon + \varepsilon^2)^4}} + \frac{16-4\varepsilon^3(x-1)+\varepsilon^4(x-1) - 8\varepsilon(x+2) + 4\varepsilon^2(1+2x)}{(4-2\varepsilon+\varepsilon^2)^2}
    \end{cases}
\end{equation}
Then, there are two regions of attractions, $R_C$ and $R_D$. Initial conditions in either region are attracted to two different steady states, $q_C$ and $q_D$ respectively. 
The steady-state $q_D$ lies in $\omega_{D,D}$, while $q_C$ is a pseudo-steady-state --- it lies in $\overline{\omega}_{C,C} \cap \overline{\omega}_{D,D}$.

If \Cref{eq:existence} is not satisfied, all initial conditions are attracted to the steady-state $q_D$.
 
\end{proposition}

The conditions for existence of a pseudo-steady-state appear complex, but the visualization in \Cref{fig:PDss} helps disentangling the various forces at play.

The higher the exploration rate, the more extreme the parameters $x,y$ need to be to sustain the cooperative equilibrium. When both $x$ and $y$ are large the payoff from mutual defection is close to mutual cooperation, and a one-period defection provides large unilateral benefits. These are the cases providing the strongest incentives for defection, while when both $x$ and $y$ are low the opposite is true. The Figure  reflects these intuitions for various levels of exploration.
\end{appendices}

\end{document}